\documentclass[11pt]{article}
\usepackage{geometry}                % See geometry.pdf to learn the layout options. There are lots.
\usepackage{graphicx}
\usepackage{epsfig}
\usepackage{color}
\usepackage{subfigure}
\usepackage{graphics}

\usepackage{amssymb}
\usepackage{amsthm}
\usepackage{amsmath}
\usepackage{amscd}
\usepackage{epstopdf}

\usepackage[ruled,longend]{algorithm2e}

\usepackage{fullpage}

\usepackage[ruled]{algorithm2e}

\newtheorem{theorem}{Theorem}[section]
\newtheorem{lemma}[theorem]{Lemma}
\newtheorem{proposition}[theorem]{Proposition}
\newtheorem{corollary}[theorem]{Corollary}

\newtheorem{definition}[theorem]{Definition}

\newtheorem{invariant}[theorem]{Invariant}

\newtheorem{remark}[theorem]{Remark}
\newtheorem{observation}[theorem]{Observation}

\def\G{{Gr}}
\title{Practical and Efficient Split Decomposition \\ via Graph-Labelled Trees}

\author{
 
Emeric Gioan\footnotemark[1], Christophe Paul\footnotemark[1], Marc Tedder\footnotemark[2],  Derek Corneil\footnotemark[2]
}

\date{}                                           % Activate to display a given date or no date

\begin{document}
\maketitle

\footnotetext[1]{CNRS - LIRMM, Univ. Montpellier II France; \{emeric.gioan,christophe.paul\}@lirmm.fr; financial support was received from the French ANR project ANR-O6-BLAN-0148-01: \emph{Graph Decomposition and Algorithms} (GRAAL).}
\footnotetext[2]{Department of Computer Science, University of Toronto; \{mtedder,dgc\}@cs.toronto.edu; financial support was received from Canada's Natural Sciences and Engineering Research Council (NSERC).}
\addtocounter{footnote}{2}

\begin{abstract}
Split decomposition of graphs was introduced by Cunningham (under the name join decomposition) as a generalization of the modular decomposition.   This paper undertakes an investigation into the algorithmic properties of split decomposition.  We do so in the context of graph-labelled trees (GLTs), a new combinatorial object designed to simplify its consideration.  GLTs are used to derive an incremental characterization of split decomposition, with a simple combinatorial description, and to explore its properties with respect to Lexicographic Breadth-First Search (LBFS).  Applying the incremental characterization to an LBFS ordering results in a split decomposition algorithm that runs 
in time $O(n+m)\alpha(n+m)$, where $\alpha$ is the
inverse Ackermann function,
whose value is smaller than 4 for any 
practical graph.  
Compared to Dahlhaus' linear time split decomposition algorithm \cite{Dah00}, which does not rely on an incremental construction, our algorithm is just as fast in all but the asymptotic sense and full implementation details are given in this paper.  Also, our algorithm extends to circle graph recognition, whereas no such extension is known for Dahlhaus' algorithm. 
The companion paper~\cite{GPTC11b} uses our algorithm to derive the first sub-quadratic circle graph recognition algorithm.  
\end{abstract}

%----------------------------------------------------------------------------------------------------------------------
%----------------------------------------------------------------------------------------------------------------------
\section{Introduction}

\emph{Split decomposition} ranks among the classical hierarchical graph decomposition techniques, and can be seen as a generalization of modular decomposition~\cite{Gal67,MR84,HP10} and the decomposition of a graph into $3$-connected components~\cite{Tut66}. It was introduced by Cunningham and Edmonds~\cite{Cun82,CE80} as a special case of the more general framework of bipartitive families.
%EME - nobody seems to know a reference for the (fake?) precision below
%, which applied not only to graphs, but also to hypergraphs and matroids. 
Since then, a number of  extensions and applications have been developed. For example, the decomposition scheme used in the proof of the Strong Perfect Graph Theorem~\cite{CRST05} and in the recognition of Berge graphs~\cite{CCL05} is based in part on the 2-join decomposition, which generalizes split decomposition.  Also, clique-width theory~\cite{CER93} and rank-width theory~\cite{Oum05} can be considered generalizations of split decomposition theory. Indeed, split decomposition is one of the important subroutines in the polynomial-time recognition of clique-width $3$ graphs~\cite{CHL00}. Moreover, the graphs of rank-width one are precisely the graphs that are totally decomposable by split decomposition (i.e. the distance-hereditary graphs~\cite{How77} or completely-separable graphs~\cite{HM90}).

%Xtof--04-10--------
As with distance hereditary graphs~\cite{HM90}, parity graphs can be characterized by their split 
decomposition~\cite{BU84,CS99}.  In~\cite{CS99b}, split decomposition is used to define a hierarchy of graph families  between distance hereditary and parity graphs.
%Xtof--04-10------
%Distance-hereditary graphs are generalized by parity graphs, and both families can be characterized by their split 
%REF1
%decomposition~\cite{BU84,HM90}.  
%Xtof--04-10--------
Split decomposition also appears in the recognition of circular arc graphs~\cite{Hsu95} and in structure theorems of various graph classes (see e.g.~\cite{TV10}).
%Xtof--04-10--------
One of the more important applications of split decomposition is with respect to circle graphs; these are the intersection graphs of chords inscribing a circle.  Prime circle graphs -- those indecomposable by split decomposition -- have unique chord representations (up to reflection)~\cite{Bou87} (see also~\cite{Cou06}).  All of the fastest circle graph recognition algorithms are based on this fact~\cite{Bou87, GHS89, Spi94}.  Recent work has focused on their connection to rank-width and vertex-minors~\cite{Bou94,Oum05}.  
%Xtof--04-10--------
For a brief introduction to split decomposition, the reader may refer to~\cite{Spi03}.
%Xtof--04-10--------

%EME - sentences below are out of purpose in this paper
%It is conjectured that circle graphs are related to vertex-minors in the same way that planar graphs are to graph-minors: just as large tree-width implies a graph has a large grid as a graph-minor, it is believed that large rank-width implies a graph has a large circle graph as a vertex-minor~\cite{Oum08}.  The conjecture has already been proved for line graphs~\cite{}.  

The first polynomial-time algorithm for split decomposition appeared in~\cite{Cun82}, and ran in time $O(nm)$.  Ma and Spinrad later developed an $O(n^2)$ algorithm~\cite{MS94}, which yields an $O(n^2)$ circle graph recognition algorithm when combined with their prime testing procedure in~\cite{Spi94}.  The only linear time algorithms for split decomposition are due to Dahlhaus~\cite{Dah00} and, more recently, Montgolfier et al.~\cite{CMR09}.  
%REF2
% MT 08/27/12 - Updating to reflect Christophe's recommendation to avoid explicitly highlighting things missing from other papers.  
However, so far neither of these linear time algorithms seems to extend to circle graph recognition.  This paper develops a split decomposition algorithm that runs in time $O(n+m) \alpha(n+m)$, where $\alpha$ is the inverse Ackermann function \cite{CLR01,Tar75} (we point out that this function is so slowly growing that it is bounded by $4$ for all practical purposes.%
\footnote{
\font\sixrm=cmr10 scaled 600
 Let us  mention that several definitions exist for this function, either with two variables, including some variants, or with one variable. For simplicity, we choose to use the version with one variable. This makes no practical difference since all of them could be used in our complexity bound, and they are all essentially constant. As an example, the two variable function considered in \cite{CLR01} satisifies $\alpha(k,n)\leq 4$ for all integer $k$ and for all $n\leq \underbrace{2^{.^{.^{.^{2}}}}}_{17 \hbox{\sixrm { times}}}$.
})
Hence, there is essentially no running time tradeoff in using our algorithm.  Moreover, the algorithm presented here is used by the companion paper~\cite{GPTC11b} to derive the first sub-quadratic circle graph recognition algorithm.

Our algorithm benefits from the recent reformulation of split decomposition in terms of graph-labelled trees (GLTs), introduced in~\cite{GP07, GP08} (see Section 2).  That paper enabled the authors to derive fully-dynamic recognition algorithms for distance-hereditary graphs and various subfamilies. 
%EME - below is exactly what GLTs bring, but check my english
GLTs are a combinatorial structure
designed to capture precisely the underlying structure of split decomposition \cite{Cun82} and in other similar reformulations that have been considered in the literature, for instance
in a logical context~\cite{Cou06} or in a distance-hereditary graph drawing context~\cite{EGY05}.
%EME - I am not so sure about what claims the next sentence
GLTs can also be understood as a special case of a term in a graph grammar~\cite{EO97}.
They are valuable here for greatly simplifying the consideration of split decomposition and providing the insight for the results in this paper.

 The overview of our algorithm appears as Algorithm \ref{alg:over}, where $G_0$ refers to the empty graph, $G_i$ denotes the subgraph of $G$ induced on $\{x_1, \cdots, x_i\}$, and $ST(G_i)$ denotes the GLT (called the split-tree) that captures the split decomposition of $G_i$.

\begin{algorithm} 
\KwIn{A connected graph $G$ with $n$ vertices.}
\KwOut{$ST(G)$, the split-tree of $G$.}
%\KwOut{A numbering of $V(G)$ inducing the ordering $\sigma$.}

\BlankLine
$ST(G_0)  \gets$ null\;
Using Algorithm \ref{alg:LBFS}, do an LBFS on $G$ to produce ordering $x_1, x_2 , \cdots, x_n$\;
%\lForEach{$x \in V(G)$} {label($x$) $\gets$ null\;}
%\SetKw{DownTo}{downto}
\For{$i = 1$ \KwTo $n$} {
          $ST(G_i) \gets ST(G_{i-1}) + x_i$, using Algorithm \ref{alg:vertexinsertion}\;
	}
\KwRet{$ST(G_n)$}\;
\caption{The Split Decomposition Algorithm} \label{alg:over}
\end{algorithm}

We use GLTs to derive a combinatorial incremental characterization of split decomposition, generalizing that given for distance-hereditary graphs in~\cite{GP07, GP08} (see Section 4).  Note that in Theorem \ref{th:cases} and its subsequent propositions,
we characterize all possible
ways in which $ST(G_{i-1})$ is modified to produce $ST(G_i)$.  GLTs are also used to demonstrate properties of split decomposition with respect to Lexicographic Breadth-First Search (LBFS)~\cite{RTL76} (see Section 3).  
Sections 3 and 4 are independent, and their content provides general results and constructions that may be useful on their own.
%Xtof--04-10--------
Notably, the results of Section 4 easily yield an efficient split decomposition dynamic algorithm supporting vertex insertion and deletion.
%Xtof--04-10--------

By applying the incremental characterization to an LBFS ordering we achieve a split decomposition algorithm that is conceptually straightforward, but requires a careful and detailed explanation of the implementation in order to achieve
the stated running time (see Section 5).  
%A charging argument we develop based on the structure of GLTs amortizes the cost of inserting each vertex.  We use it to prove the $O(\alpha(n+m)(n+m))$ running-time (see Section 6).  
%
%%%EMERIC-29-06
%A charging argument we develop based on the structure of GLTs allows us to evaluate the amortized cost of inserting each vertex.  We use it to prove the $O(n+m)\cdot \alpha(n+m)$ running-time (see Section 6).  
%Furthermore, our algorithm extends to circle graph recognition:  the companion paper~\cite{GPTC-circle} uses it to develop the first sub-quadratic circle graph recognition algorithm, which also runs in $O(n+m)\cdot \alpha(n+m)$ time.      
%%%EMERIC-29-06
%
We develop a charging argument based on the structure of GLTs  that allows us to evaluate the amortized cost of inserting each vertex, according to an LBFS ordering.  
%
%EME-Ack-update
%We use it to prove the $O((n+m)\alpha(n+m))$ running-time (see Section 6). 
We use it to prove the $O(n+m)\alpha(n+m)$ running time (see Section 6). Furthermore, our algorithm extends to circle graph recognition; the companion paper~\cite{GPTC11b} uses it to develop the first sub-quadratic circle graph recognition algorithm, 
%EME-Ack-update
%which also runs in $O((n+m)\alpha(n+m))$ time.    
which also runs in $O(n+m)\alpha(n+m)$ time.    
Note that different
versions of both the split decomposition algorithm and the circle graph recognition algorithm appear in \cite{Ted11}.

%----------------------------------------------------------------------------------------------------------------------
%\newpage
%----------------------------------------------------------------------------------------------------------------------
%----------------------------------------------------------------------------------------------------------------------
\section{Preliminaries}

All graphs in this document are simple, undirected, and connected.  The set of vertices in the graph $G$ is denoted $V(G)$ and the set of edges by $E(G)$.  The graph \emph{induced} on the set of vertices $S$ is signified by $G[S]$. We let $N_G(x)$, or simply $N(x)$, denote the neighbours of vertex $x$, and for $S$ a set of vertices $N(S)=(\cup_{x\in S} N(x))\setminus S$.
A vertex is \emph{universal to} a set of vertices $S$ if $S\subseteq N(x)$; it is \emph{isolated} from $S$ if $N(x)\cap S=\emptyset$.  %The size of $S$ is $|S|$.  
A vertex is \emph{universal in} a graph if it is adjacent to every other vertex in the graph.  We use $N[x] = N(x) \cup \{x\}$ to denote the \emph{closed neighbourhood of a vertex}.  Two vertices $x$ and $y$ are \emph{twins} if $N(x) \setminus \{y\} = N(y) \setminus  \{x\}$.  A \emph{pendant} is a vertex of degree one.  A \emph{clique} is a graph in which every pair of vertices is adjacent.  A \emph{star} is a graph with at least three vertices in which one vertex, called its \emph{centre}, is universal, and no other edges exist; the vertices other than the centre are called its {\emph{degree-1 vertices}}.  The clique on $n$ vertices is denoted $K_n$; the star on $n$ vertices is denoted $S_n$.

The graph $G+(x,N(x))$ is formed by adding the vertex $x$ to the graph $G$ adjacent to the subset $N(x)$ of vertices, its neighbourhood; when $N(x)$ is clear from the context, we simply write $G+x$.  The graph $G-x$ is formed from $G$ by removing $x$ and all its incident edges.  

The non-leaf vertices of a tree $T$ are called its \emph{nodes}.  The edges in a tree not incident to leaves are its \emph{internal} edges.  If $S$ is a set of leaves of $T$, then $T(S)$ denotes the smallest connected subtree spanning $S$.  If $T$ is a tree, then $|T|$ represents its number of nodes and leaves. In a rooted tree $T$, every node or leaf $x$ (except the root) has a unique \emph{parent}, namely its neighbour on the path to the root. A \emph{child} of a node $x$ is a neighbour of $x$ distinct from its parent.

%----------------------------------------------------------------------------------------------------------------------
\subsection{Split decomposition}

This subsection recalls original definitions from \cite{Cun82}.
%EME - paragraph below deleted since useless. 
%The origin and importance of split decomposition were discussed in the introduction.  Roughly speaking, split decomposition recursively decomposes the two parts in a \emph{split}:
%
%EME - I considered adding a paragraph like the one below, but it sounded too redundant with the last paragraph
%%For sake of clarity and completeness of the paper, we briefly present here the original definition of the split decomposition as considered in \cite{}. However, except for the first definitions of splits, prime and degenerate graphs, this subsection is useless for what follows, since split decomposition will be immediately reformulated more conveniently for our purpose. 

% MT 08/27/12 - Reviewer rightly pointed out the second item is redundant in the given definition.  
\begin{definition}
A \emph{split} of a connected graph $G=(V,E)$ is a bipartition $(A,B)$ of $V$, where $|A|,|B| > 1$ such that every vertex in $A' = N(B)$ is universal to $B' = N(A)$.  The sets $A'$ and $B'$ are called the \emph{frontiers} of the split.  

%\begin{enumerate}
%\item every vertex in $A'=N(B)$ is universal to $B'=N(A)$;
%\item no other edges exist between vertices in $A$ and $B$.
%\end{enumerate}

%The sets $A'$ and $B'$ are called the \emph{frontiers} of the split.  
\end{definition}

A graph not containing a split is called \emph{prime}. 
%EME - useless to put a figure for this
%METTRE UNE FIGURE AVEC C_5, HOUSE  or GEM
%(see figure~\ref{fig:}.) 
A bipartition is \emph{trivial} if one of its parts is the empty set or a singleton. Cliques and stars are called \emph{degenerate} since every non-trivial bipartition of their vertices is a split:

\begin{remark} \label{degenerate}
Let $(A,B)$ be a bipartition of the vertices in a clique or a star such that $|A|,|B|>1$.  Then $(A,B)$ is a split.
\end{remark}

Degenerate graphs and prime graphs represent the base cases in the process defining split decomposition:

\begin{definition} \label{def:split-dec-tree}
\emph{Split Decomposition} is a recursive process decomposing a given graph $G$ into a set of disjoint graphs $\{G_1,\dots G_k\}$, called \emph{split components}, each of which is either prime or degenerate. There are two cases:

\begin{enumerate}
\item if $G$ is prime or degenerate, then return the set $\{G\}$;
\item if $G$ is neither prime nor degenerate, it contains a split $(A,B)$, with frontiers $A'$ and $B'$.  The
split decomposition of $G$ is then the union of the split decompositions of the graphs $G[A] + a$ and $G[B] + b$, where $a$ and $b$ are new vertices, called \emph{markers}, such that $N_{G[A] + a}(a) = A'$ and $N_{G[B] + b}(b) = B'$.
\end{enumerate}
\end{definition}

Notice that during the split decomposition process, the marker vertices can be matched by so called \emph{split edges}. Then given a split decomposition, provided the marker vertices and their matchings are specified, the input graph $G$ can be reconstructed without ambiguity. The set of split edges merely defines the \emph{split decomposition tree} whose nodes are the components of the split decomposition. 

Cunningham showed that every graph has a canonical split decomposition tree~\cite{Cun82}. As Cunningham's original work was on the decomposition of a graph by a family of bipartitions of the vertex set, his paper focuses on the tree representation of the family of splits to obtain a canonical tree rather than on how the graph's adjacencies can be retrieved from its split decomposition tree. At first sight, it is not immediately clear how the graph's adjacencies are encoded by the split decomposition tree, and what role the marker vertices play in determining them.  Tellingly, the base case treats prime and degenerate graphs the same; looking at the tree, the viewer is left to guess which one applied. In recent papers~\cite{GP03,Cou06}, split decomposition is represented by the \emph{skeleton graph} which is the union of the split components connected by the split edges. The fact that $G$'s vertices and the marker vertices are mixed is a drawback of this representation.

%EME - paragraph below rewritten
%XTof - WATCH OUT: ADAPT THANKS TO THE SIDMA REPORT ON DH
%This goes some way to explaining how little is known about split decomposition (pg. ?,~\cite{}).  A recent reformulation of split decomposition in terms of graph-labelled trees (GLTs) aims to change this~\cite{}.  Our investigation of split decomposition takes place entirely in this new GLT setting, which is described below.
%
A recent reformulation of split decomposition in terms of graph-labelled trees (GLTs) aims to clarify this~\cite{GP07, GP08}.  Our investigation of split decomposition takes place entirely in this new GLT setting, which is described below.

%----------------------------------------------------------------------------------------------------------------------
\subsection{Graph-labelled trees} \label{sec:GLT}

This subsection recalls definitions from \cite{GP07, GP08} and adds useful terminology.
%EME useless
%Graph-labelled trees (GLTs) were first formalized in~\cite{}; a list of some of their applications appears in the introduction.  Their definition follows:

\begin{definition} [\cite{GP07, GP08}]
A \emph{graph-labelled tree} (GLT) is a pair $(T,\mathcal{F})$, where $T$ is a tree and $\mathcal{F}$ a set of graphs, such that each node $u$ of $T$ is \emph{labelled} by the graph $G(u) \in \mathcal{F}$, and there exists a bijection $\rho_u$ between the edges of $T$ incident to $u$ and the vertices of $G(u)$.  (See Figure~\ref{fig:GLTexample}.)
\end{definition}

%EME slight modification since such sentences are an open door to ambiguity and missunderstanding
%When we refer to a node $u$ in a GLT, we could mean either the node itself in $T$, or its label $G(u) \in \mathcal{F}$; the meaning will be clear from context.  Notation will be simplified by saying $V(u) = V(G(u))$.
%
When we refer to a node $u$ in a GLT, we usually mean the node itself in $T$ (non-leaf vertex). We may sometimes use the notation $u$ as a shortcut for its label $G(u) \in \mathcal{F}$; the meaning will be clear from the context.  
For instance, notation will be simplified by saying $V(u) = V(G(u))$.  
The vertices in $V(u)$ are called \emph{marker vertices}, %\footnote{The dual-use of this term is intentional, as we'll see.}
and the edges between them in $G(u)$ are called \emph{label-edges}. 
%EME added def below
For a label-edge $e=uv$ we may say that $u$ and $v$ are \emph{the vertices of} $e$.
The edges of $T$ are \emph{tree-edges}.  The marker vertices $\rho_u(e)$ and $\rho_v(e)$
of the internal tree-edge $e = uv$ are called the \emph{extremities} of $e$ .  Furthermore, $\rho_v(e)$ is the \emph{opposite} of $\rho_u(e)$ (and vice versa).  A leaf is also considered an \emph{extremity} of its incident edge, and its opposite is the other extremity of the edge (marker vertex or leaf). For convenience, we will use the term \emph{\it adjacent} between: a tree-edge and one of its extremities; a label-edge and one of its vertices; two extremities of a tree-edge, etc., as long as the context is clear. The most important notion for GLTs with respect to split decomposition is that of \emph{accessibility}:

%EME - I have improved the def below, by extending it immediately to leaves (without waiting a half page), and by noting that the sequence forms a path (which was not said and left to reader's thinking). But still I am not fully satisfied, maybe it is not enough natural yet.

\begin{definition} [\cite{GP07, GP08}]
Let $(T,\mathcal{F})$ be a GLT.  The marker vertices $q$ and $q'$ 
%(possibly $q = q'$) 
are \emph{accessible} from one another if there is a sequence $\Pi$ of marker vertices $q,\ldots,q'$ such that:

\begin{enumerate}
\item every two consecutive elements of $\Pi$ are either the vertices of a label-edge or the extremities of a tree-edge;
\item the edges thus defined alternate between tree-edges and label-edges.
\end{enumerate}

% MT 08/27/12 - Circle reviewer had difficulty just how "natural" this definition extends.Updating below to be consistent with changes made in circle paper.
Two leaves are accessible from one another if their opposite marker vertices are accessible; similarly for a leaf and marker vertex being accessible from one another; see Figure~\ref{fig:GLTexample} where the leaves accessible from
$q$ include both 3 and 15 but neither 2 nor 11.
%This definition naturally extends when the first and/or the last element of $\Pi$ is a leaf (providing \emph{accessibility} between two leaves, or between a leaf and a marker vertex). 
By convention, a leaf or marker vertex is accessible from itself.
\end{definition}

Note that, obviously, if two leaves or marker vertices are accessible from one another, then the sequence $\Pi$ with the required properties is unique, and the set of tree-edges in $\Pi$ forms a path in the tree $T$.

%EME added the macro \G below, can be modified, but the notation G(T,F) for this function was not formally good, since possibly confusing with the notation G for a graph.

\begin{definition}[\cite{GP07, GP08}]
Let $(T,\mathcal{F})$ be a GLT.  Then its \emph{accessibility graph}, denoted $\G(T,\mathcal{F})$, is the graph whose vertices are the leaves of $T$, with an edge between two distinct leaves $\ell$ and $\ell'$ 
%(for $\ell\not=\ell'$) 
if and only if they are accessible from one another. Conversely, we may say that $(T,\mathcal{F})$ \emph{is a GLT of} $\G(T,\mathcal{F})$.
\end{definition}

%\begin{figure} %[htbh]
%\begin{center}
%\includegraphics[scale=0.75]{SPLIT---FIG_1(b)---graph.eps}
%\end{center}
%%\vspace{-0.2in}
%\caption{The accessibility graph for the GLT in figure~\ref{fig:GLTexample}.}
%\label{fig:accessibilityExample}
%\end{figure}

Accessibility allows us to view GLTs as encoding graphs; an example appears in Figure~\ref{fig:GLTexample}.  
\bigskip

\begin{figure}[htbh]
\begin{center}
\includegraphics[scale=0.75]{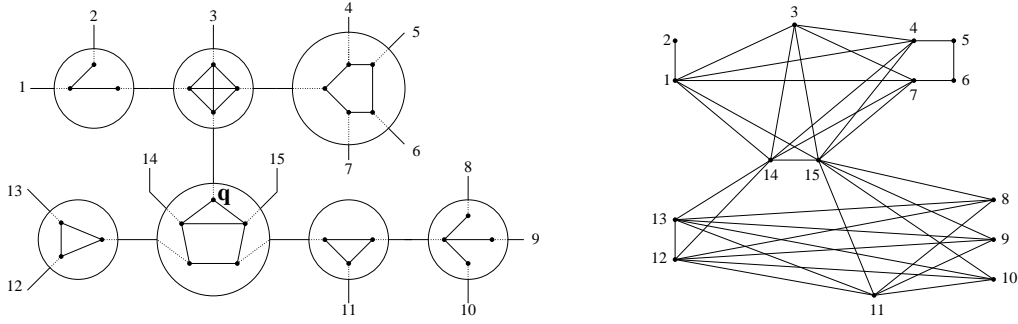}

\end{center}
\caption{A graph-labelled tree $(T,\mathcal{F})$ and its accessibility graph $\G(T,\mathcal{F})$.}
\label{fig:GLTexample}
\end{figure}

%The following extended notion of accessibility helps establish the connection to split decomposition:
%
%\begin{definition}
%Let $(T,\mathcal{F})$ be a GLT.  The leaf $\ell$ and the marker vertex $q$ are \emph{accessible} from one another if $\ell$'s opposite and $q$ are accessible from one another.  
%\end{definition}

%This notation is generalized as follows:

%\begin{notation}
Let $(T,\mathcal{F})$ be a GLT, and let $q$ be a marker vertex belonging to the node $u$ of $T$ and corresponding to the tree-edge $e$ of $T$.  
Then we denote:

- $L(q)$ the set of leaves of $T$ from which there is a path to $u$ using $e$;

- $A(q)$ the subset of leaves of $L(q)$ that are accessible from $q$; 

- $T(q)=T(L(q))$ the smallest subtree of $T$ that spans the leaves $L(q)$; note that $q \notin T(q)$.  

%EME not sure that the definitions of L(l) and T(l) below are useful%
%To unify our notation, for a leaf $\ell$ of $T$ incident to the tree-edge $e$, the sets $L(\ell)$, $A(\ell)$, $T(\ell)$ are similarly defined. 
%Note that $A(\ell)=N(\ell)$, $\ell$'s neighbourhood in $\G(T,\mathcal{F})$.
%
To unify our notation, for a leaf $\ell$ of $T$, the sets $L(\ell)$, $A(\ell)$, $T(\ell)$ can be similarly defined, 
so that $A(\ell)=N_G(\ell)$, $\ell$'s neighbourhood in $G=\G(T,\mathcal{F})$, and  $L(\ell)=V(G)\setminus \{\ell\}$.
%\end{notation}  

\begin{definition} \label{def:descendant}
Let $(T,\mathcal{F})$ be a GLT and let $q$ and $p$ be distinct marker vertices. Then $p$ is a \emph{descendant} of $q$ if $L(p)\subset L(q)$, that is if $T(p)$ is a subtree of $T(q)$.
\end{definition}

%EME precision below is useless I guess
%Notice that if $\ell$ is a leaf belonging to the subtree $T(q)$ for some marker vertex $q$, then $\ell$ is not a descendant of $q$. 

The above notation and definitions are illustrated in Figure~\ref{fig:treeExample}.  Also note that a leaf is never a 
descendant of a leaf or a marker vertex.
\bigskip

\begin{figure}[htbh]
\begin{center}
\includegraphics[scale=0.75]{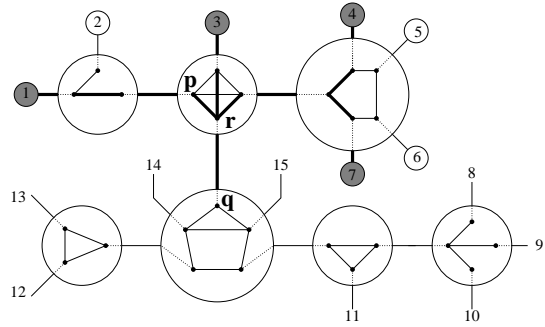}
\end{center}
\caption{A marker vertex $q$, with $L(q)=\{1,2,3,4,5,6,7\}$ and $A(q)=\{1,3,4,7\}$. Bold edges are those used in sequences certifying the accessibility between $q$ and elements of $A(q)$. The vertices of the subtree $T(q)$ are the elements of $L(q)$ together with the three upper nodes. The marker vertex $p$ with $L(p)=\{1,2\}$ is a descendant of $q$.
The marker vertex $r$ is the opposite of $q$, and is not a descendant of $q$.}
%It is not a descendant of $q$, and neither are the leaves of $T(r)$.}
\label{fig:treeExample}
\end{figure}

%As an example, consider once more the GLT in Figure~\ref{fig:GLTexample}, and let $u$ be the node adjacent to the leaves 12 and 13, with $q \in V(u)$ being the marker vertex not opposite these leaves.  Then $L(q) = \{1,2,\ldots,10,11,14,15\}$, whereas $A(q) = \{8,9,10,11,14\}$. 

\medskip
We conclude this subsection by a series of remarks following directly from the definitions.

%The next remark follows immediately:

\begin{remark} \label{connectedLabels}
If a graph $G$ is connected, then every label in a GLT of $G$ is connected.
\end{remark}  

\begin{remark} 
For any marker vertex $q$ in a GLT of a connected graph, $A(q) \ne \emptyset$.
\end{remark}  

As a consequence, by choosing one element of $A(q)$ for every marker vertex $q$ in the label we see that every label in a GLT of a connected graph $G$ is an induced subgraph of $G$.

\begin{remark} \label{rem:hered}
Let  $p$ and $q$ be two marker vertices of a GLT such that $p$ is a decendent of $q$. If $p$ and $q$ are accessible from one another, then $A(q)\cap L(p)=A(p)$. If $p$ and $q$ are non-accessible from one another, then $A(q)\cap L(p)=\emptyset$.
\end{remark}

%\begin{remark}  \label{rem:edge-split}
%Let $q$ and $r$ be the extremities of an internal tree-edge in the GLT $(T,\mathcal{F})$.  Then $(L(q),L(r))$ is a split in $\G(T,\mathcal{F})$.
%\end{remark}

%The next section extends this remark to show how GLTs can be restricted to completely capture the split decomposition.

%----------------------------------------------------------------------------------------------------------------------
\subsection{The split-tree}
\label{subsec:split-tree}

This subsection reformulates split decomposition \cite{Cun82} %main result \cite{C82}
in the GLT setting, as done in \cite{GP07, GP08}.

%It is easy to see that GLTs do not uniquely encode graphs.  For example, the graph in Figure~\ref{fig:GLTexample} can equally be represented by a GLT consisting of a single node labelled by the graph itself.  Another way to see this is through the \emph{node-join} and \emph{node-split} operations:

\begin{definition} 
Let $e$ be a tree-edge incident to nodes $u$ and $u'$ in a GLT, and let $q \in V(u)$ and $q' \in V(u')$ be the extremities of $e$.  The \emph{node-join} of $u,u'$
%(for ``node-join'') 
replaces $u$ and $u'$ with a new node $v$ labelled by the graph formed from $G(u)$ and $G(u')$ as follows: all possible label-edges are added between $N(q)$ and $N(q')$, and then $q$ and $q'$ are deleted.  
See Figure~\ref{fig:nJoinNsplitExample}.
\end{definition}

\begin{definition} \label{def:node-split}
The \emph{node-split} 
%(for ``node-split'') 
is the inverse of the node-join.  More precisely, let $v$ be a node such that $G(v)$ contains the split $(A,B)$ with frontiers $A'$ and $B'$.  The node-split with respect to $(A,B)$ replaces $v$ with two new adjacent nodes $u$ and $u'$ labelled by $G[A] + q$ and $G[B] + q'$, respectively, where $q$ and $q'$ are the extremities of the new tree-edge thus created, $q$ being universal to $A'$, and $q'$ being universal to $B'$. The extremities of the tree-edges incident to $v$ remain unchanged. 
See Figure~\ref{fig:nJoinNsplitExample}.
\end{definition}

%These two definitions are illustrated by Figure~\ref{fig:nJoinNsplitExample}.
%\bigskip

When a node-split or a node-join operation is performed, 
a marker vertex of the initial GLT is \emph{inherited by} the resulting GLT
through the operation if its corresponding tree-edge has not been affected by the operation, i.e. if its corresponding tree-edge is not 
created or deleted in one of the above definitions. 
%DGC1  For convenience, we may use the same notation for a marker vertex and its inherited~one.
%We say that a marker vertex is \emph{inherited} through a node-split or a node-join if its corresponding tree-edge has not been affected by the operation. For convenience, we may sometimes identify a marker vertex of a GLT and its inherited one in the resulting GLT.
%

\medskip

The key property to observe is: 

\begin{observation}

The node-join operation and the node-split operation preserve the accessibility graph of the GLT.

\end{observation}

%\bigskip

\begin{figure}[htbh]
\begin{center}
\includegraphics[scale=0.70]{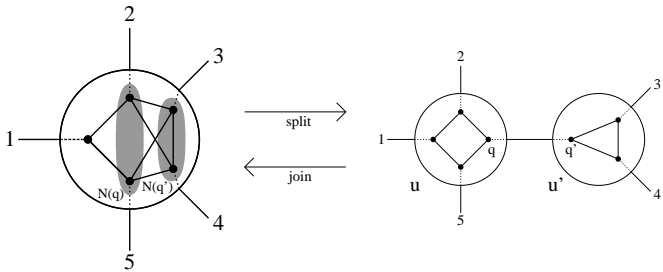}
\end{center}
%\caption{Example of the node-join and node-split.}
\caption{Example of the node-join and node-split.}
\label{fig:nJoinNsplitExample}
\end{figure}

Hence, GLTs do not uniquely encode graphs. 
In particular, recursive application of the node-join on every edge of a GLT of $G$ leads to the GLT with a unique node labelled by the accessibility graph $G$. And conversely, any GLT of a graph $G$ can be obtained by recursive application of the node-split from the GLT consisting of a unique node labelled by $G$.
%A single node labelled by the accessibility graph can be obtained by recursive application of the node-join; the original GLT can be obtained by recursive application of the inverse node-split operations.  

Also, observe that, as a consequence, the accessibility graph $G$ of a GLT and the tree structure of the GLT (with leaves labelled by $V(G)$) completely determine the node labels of the GLT.
Therefore, transforming a GLT into another GLT using node-splits and node-joins can be done using any ordering for such operations.  In particular, performing a set of node-joins can be done in any order without changing the result (the final tree structure is obtained by contracting edges from the initial tree).
And concerning node-splits, creating two tree-edges using these operations can be done equally by creating first one tree-edge or the other.
We emphasize these two remarks, as they will guarantee the consistency of further constructive statements.

%EME - I would think that the two detailed remarks below are useless, since their combinatorial meaning (as explained above) is more simple and general than alluded to in these remarks. But I let these remarks here since Christophe and Derek cared about them according to Christophe.

%It follows from the definition that the \emph{node-join} operation is \emph{communtative}:

\begin{remark} \label{rem:node-join-commutative}
%MT 08/27/12 - Reviewer wanted clearer wording  to emphasize order invariant
Applying a sequence of node-joins on a GLT yields the same GLT, regardless of the order of the node-joins.
%Applying successive \emph{node-joins} on a GLT -- in whatever order -- yields the same GLT.
\end{remark}

%As for the \emph{node-join} a commutativity property can be observed for the \emph{node-split}, but with some restriction (we use the notations of definition~\ref{def:node-split}):

\begin{remark} \label{rem:node-split-commutative}
Recalling the notation of Definition \ref{def:node-split}, let $v$ be a node of a GLT and let $(A,B)$ and $(C,D)$ be two splits of $G(v)$ such that $A\subset C$. Then applying the \emph{node-split} on $v$ with respect to $(A,B)$ and then on node $u'$ with respect to $((C\setminus A)\cup\{q'\},D)$ or applying the \emph{node-split} on $v$ with respect to $(C,D)$ and then on node $u$ with respect to $(A,(B\setminus D)\cup\{q\})$ yields the same GLT. 
\end{remark}

%Notice that the underlying accessibility graph remains unchanged by node-join and node-split operations.  A single node labelled by the accessibility graph can be obtained by recursive application of the node-join; the original GLT can be obtained by recursive application of the inverse node-split operations.  

\medskip
Of special interest are those node-joins/splits involving degenerate nodes.  The \emph{clique-join} is a node-join involving adjacent cliques: its result is a clique node; the \emph{clique-split} is its inverse operation.  The \emph{star-join} is  a node-join involving adjacent stars whose common incident 
tree-edge has exactly one extremity that is the centre of its star: its result is a star node; the \emph{star-split} is its inverse operation.  Figure~\ref{fig:cliqueStarExample} provides examples. 
%EME - removed transitory sentence (I removed a lot in the paper... they were redundant with the result statements placed just after them)
% A unique GLT representation exists when no clique-join or star-join is possible:

\begin{figure}[htbh]
\begin{center}
\includegraphics[scale=1]{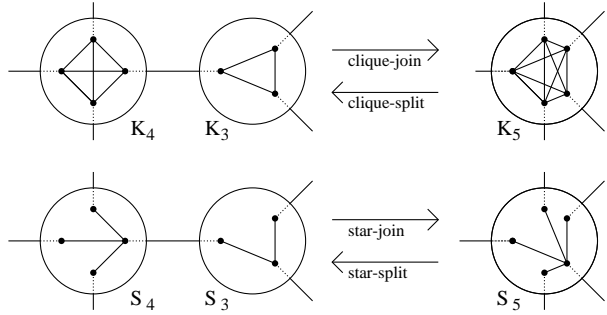}
\end{center}
\caption{Examples of the clique-join/split and star-join/split.}
\label{fig:cliqueStarExample}
\end{figure}

\begin{definition}
A GLT is \emph{reduced} if all its labels are either prime or degenerate, and no clique-join or star-join is possible.
\end{definition} 

%EME - I put ref C82 everywhere here, I guess it is the correct reference, but not absolutely sure

We can now state the main result of \cite{Cun82}, as reformulated in \cite{GP07, GP08}.

\begin{theorem} %[reformulation of \cite{C82} in \cite{GP07}\cite{GP08}]
[\cite{Cun82, GP07, GP08}] \label{theo:reduced-GLT}
For any connected graph $G$, there exists a unique, reduced graph-labelled tree $(T,\mathcal{F})$ such that  $G=\G(T,\mathcal{F})$.
\end{theorem}  

The unique GLT guaranteed by the previous theorem is the \emph{split-tree}, and is denoted $ST(G)$.  
%\medskip
%
The GLT in Figure~\ref{fig:GLTexample} is the split-tree for the accessibility graph pictured there.  The split-tree is the intended replacement for Cunningham's split decomposition tree. The following theorem first appeared in Cunningham's seminal paper~\cite{Cun82} in an equivalent form. We phrase it in terms of GLTs and the split-tree:

%\begin{theorem} \cite{???} \label{prop:split_list}
%Let $ST(G) = (T,\mathcal{F})$.  Any split of $G$ is the bipartition (of leaves) induced by removing an internal tree-edge from $\tilde T$, where $\tilde T = T$, or $\tilde T$ is obtained from $T$ by exactly one star-split or clique-split.
%\end{theorem}

\begin{theorem}[\cite{Cun82}] \label{prop:split_list}
Let $G$ be a connected graph. A bipartition $(A,B)$ is a split of $G$ if and only if either there exists an internal tree-edge of $ST(G)$ with extremities $p$ and $q$ such that $A=L(p)$ and $B=L(q)$, or there exists a degenerate node $u$ and a split $(A_u,B_u)$ of $G(u)$ such that $A=\cup_{p\in A_u} L(p)$ and $B=\cup_{p\in B_u} L(p)$.
\end{theorem}

In order words, the split-tree can be understood as a compact representation of the family of splits of a connected graph. Indeed it is easy to show that the size of the split-tree $ST(G)$ of a graph is linear in the size of $G$ (the sum of the sizes of label graphs of $ST(G)$ is linear in the size of $G$), 
%DGC2:   size of its graphs
whereas a graph can have exponentially many splits (it is the case for the clique and the star). The following corollary is simply a rephrasing of Theorem~\ref{prop:split_list} based on the node-split operation.

%EME - it is maybe be awkawrd to cite \cite{GP07}\cite{GP08} almost as much as cunningham's paper ! don't know

\begin{corollary} [\cite{Cun82, GP07, GP08}]\label{cor:split_list}
Let $ST(G) = (T,\mathcal{F})$.  Any split of $G$ is the bipartition (of leaves) induced by removing an internal tree-edge from $\tilde T$, where $\tilde T = T$, or $\tilde T$ is obtained from $T$ by exactly one %star-split or clique-split.
node-split of a degenerate node.
\end{corollary}

%Note that the case where $\tilde T$ is obtained by a star-split or clique-split deals with the split of a degenerate node.  
Compared to Cunnigham's split decomposition tree or the skeleton graph representation (see the remarks following Definition~\ref{def:split-dec-tree}), the advantage of the split-tree is manifest.
The adjacency relation in the underlying graph is now explicitly represented by the accessibility relation, and the role played by the marker vertices (and their own adjacencies) is established.  
%EME - sentence below removed: I would have thought that these base cases were already identified in \cite{C82}
%Moreover, the base cases are now highlighted as either prime or degenerate based on the labels present.  
All this added information comes with no space-tradeoff:

\begin{lemma} [\cite{GP07, GP08}] \label{treeSize}
Let $ST(G) = (T,\mathcal{F})$.  If $x \in V(G)$, then $|T(N(x))| \le 2 \cdot |N(x)|$.
\end{lemma}
   
%EME paragraph below useless (the paper is already so long...)
%In fact, the split-tree is merely a reorganization of Cunningham's split decomposition tree.  The marker vertices have just been collected at the nodes, and their adjacencies stressed through the introduction of label edges.  However, these small changes lead to the results in this paper.  Amongst them is a series of results linking split decomposition with Lexicographic Breadth-First Search.  The latter is introduced next.

%----------------------------------------------------------------------------------------------------------------------
%\newpage
%----------------------------------------------------------------------------------------------------------------------
%----------------------------------------------------------------------------------------------------------------------
\section{Lexicographic breadth-first search} \label{sec:LBFS}

As mentioned in the introduction, our algorithm incrementally builds the split-tree by adding vertices one at a time from the input graph. %Now, adding a single vertex to the split-tree can require $O(n)$ changes, as demonstrated in Figure~\ref{fig:badExample}.  
%EME below MIEUX
Now, adding a single vertex %with just two neighbours 
to the split-tree of a graph with $n$ vertices can require 
%MT 08/27/12 - Reviewer rightly points out it's trivially O(n) -- so make it \theta(n)
$\Theta(n)$ 
%$O(n)$ 
changes, as demonstrated in Figure~\ref{fig:badExample}. 
%EME above rewritten below since FIg 9 and its comment should go in section 3, since completely about building ST(G+x) and good example of its limitation
%But, as mentioned above, the incremental construction developed earlier is not sufficent for complexity purpose.
However, later in the paper we prove that if vertices are added according to a Lexicographic Breadth-First Search (LBFS) ordering, then the total cost of inserting all vertices of $G$ can be amortized to linear time up to inverse Ackermann function.  

\begin{figure}[htbh]
\begin{center}
\includegraphics[scale=0.75]{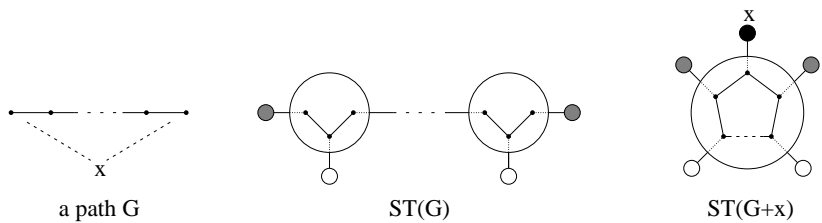}
\end{center}
\caption{Adding a single vertex adjacent to the ends of a path requires 
%MT 08/27/12 -- As above, it's trivially O(n) so make it \theta(n)
$\Theta(n)$ 
%$O(n)$ 
changes to the split-tree (the neighbours of $x$ appear as the grey leaves).}
\label{fig:badExample}
\end{figure}

\bigskip

This section presents new LBFS results on the split decomposition and more generally on GLTs. We first present the LBFS algorithm and some known results.

%----------------------------------------------------------------------------------------------------------------------
\subsection{LBFS orderings}

Lexicographic Breadth-First Search (LBFS) was developed by Rose, Tarjan, and Lueker for the recognition of chordal graphs~\cite{RTL76} and has since become a standard tool in algorithmic graph theory~\cite{Cor04}.  
%It appears here as Algorithm~\ref{alg:LBFS}.

%EME - terminology for "ordering" and "numbering" and their use were mixed and confusing in this subsection. I tried to be more correct

%EME consider the integers below was necessary to define properly the restriction \sigma[S],
%
An \emph{ordering} $\sigma$ of a graph $G$ is a linear ordering of its set of vertices $V(G)$.
Formally, we can define it either as an injective mapping from $V(G)$ to the integers, or as an ordering binary relation. We slightly abuse notation by allowing $\sigma$ to represent such a mapping as well as the ordering, and we let $<_\sigma$ denote the binary relation: $x<_{\sigma} y$ is equivalent to $\sigma(x)<\sigma(y)$.
%
%EME heavy precisions below added for consistency in the way that are written the statement and proof of the third Lemma \ref{inducedLBFS} and the definition before it (could be more improved...)
%
In such a case, we say that ``$x$ appears before $y$'', or ``earlier than $y$'', in $\sigma$. Similarly,  by ``first'', ``last'' and ``penultimate'', we denote respectively,  the smallest element of $<_\sigma$, the greatest element and the element appearing immediately before the last one. 
%

%EME in algo 1, we cannot define an ordering as a mapping, since in this case the restriction to S would not be such a mapping onto {1,...,|S|}

 By an LBFS ordering of the graph $G$, we mean any ordering produced by Algorithm~\ref{alg:LBFS} on input  graph $G$. 
%EME precision below added
Notice that such an ordering can be built in linear time (see e.g.~\cite{Gol80, HMP00}).  

\begin{algorithm} 
\KwIn{A graph $G$ with $n$ vertices.}
\KwOut{An ordering $\sigma$ defined by a mapping $\sigma : V(G) \rightarrow \{1,\ldots,n\}$.}
%\KwOut{A numbering of $V(G)$ inducing the ordering $\sigma$.}

\BlankLine
\lForEach{$x \in V(G)$} {label($x$) $\gets$ null\;}
\SetKw{DownTo}{downto}
\For{$i = 1$ \KwTo $n$} {
	pick an unnumbered vertex $x$ with lexicographically largest label\;
	$\sigma(x) \gets i$ \tcp*[l]{assign $x$ the number $i$} 
	\lForEach{unnumbered vertex $y \in N(x)$} {append $n - i + 1$ to label($y$)\;}
	}

\caption{Lexicographic Breadth-First Search} \label{alg:LBFS}
\end{algorithm}

\medskip
  The next result characterizes LBFS orderings:

\begin{lemma} [\cite{DNB96, Gol80}] \label{4vertex}
An ordering $\sigma$ of a graph $G$ is an LBFS ordering if and only if for any triple of vertices $a <_{\sigma} b <_{\sigma} c$ with $ac \in E(G)$, $ab \notin E(G)$, there is a vertex $d <_{\sigma} a$ such that $db \in E(G)$, $dc \notin E(G)$.
\end{lemma}  

%EME definition below is less formal, I changed it. Also added def of prefix.
% For a subset $S$ of vertices, the ordering $\sigma[S]$ is the same as $\sigma$, but with all vertices not in $S$ removed.

For a subset $S$ of $V(G)$, $\sigma[S]$ denotes the restriction of $\sigma$ to $S$.
A {\it prefix} of an ordering $\sigma$ is a set $S$ such that $x<_{\sigma} y$ and $y\in S$ implies $x\in S$.
One obvious result is the following:

\begin{remark} \label{prefix}
Let $S$ be a prefix of any LBFS ordering $\sigma$ of connected graph $G$.  Then $\sigma[S]$ is an LBFS ordering of $G[S]$, and $G[S]$ is connected.
\end{remark}

%----------------------------------------------------------------------------------------------------------------------
\subsection{LBFS and split decomposition}
\label{sub:LBFS}

We now introduce a general lemma about split decomposition, followed by 
lemmas relating LBFS orderings %(with last vertex $x$)
and split decomposition.  
%EME sentences below : redundant with lemma statements
%The first, although not specifically related to an LBFS ordering, immediately applies by Remark~\ref{prefix}.  The second establishes properties of a twin of vertex $x$ in prime graph $G$.  

\begin{lemma} \label{newSplit}
Let $G$ and $G+x$ be two connected graphs such that $G$ is prime but $G+x$ is not. Then either $x$ is a pendant vertex or $x$ has a twin.
\end{lemma}

\begin{proof}
Since $G + x$ is not prime, it has a split $(A,B)$.  Let $A'$ and $B'$ be the frontiers of the split.  Without loss of generality, assume that $x \in A$.  Since $(A \setminus \{x\},B)$ is not a split in $G$, we know that $|A| = 2$.  If $A' = \{x\}$, then $G$ is disconnected.
If $A' = \{x,y\}$, then $y$ is a twin of $x$.  If $A' = \{y\}, y \ne x$, then $N(x) = \{y\}$, since $G+x$ is connected.  Therefore $x$ is a pendant.
\end{proof}

\begin{lemma} \label{LBFStwin}
Let $G$ and $G+x$ be two connected graphs and let $\sigma$ be an LBFS ordering of $G+x$ in which $x$ appears last.  If $G$ is prime and $x$ has a twin $y$, then $y$ is either universal in $G$ or is the penultimate vertex in $\sigma$.
\end{lemma}

\begin{proof}
Observe that if $|V(G)| > 3$, then $y$ is unique since $G$ is prime. 
Consider an execution of Algorithm~\ref{alg:LBFS} that produced the ordering $\sigma$.  Let $S$ be the set of vertices with the same label as $y$ at the time $y$ is numbered by Algorithm~\ref{alg:LBFS} (of course $y\in S$). As $x$ and $y$ are twins, we must have $x \in S$. We can assume that $S\setminus\{y,x\}\neq\emptyset$ as otherwise $y$ would be the penultimate vertex of $\sigma$. Let $B$ be the set of vertices numbered before $y$ by Algorithm~\ref{alg:LBFS}. Observe that $|B|\leqslant 1$ as otherwise $G$ would contain the split $(B,S\setminus\{x\})$.

%MT 08/27/12 - Reviewer wanted application of 4 vertex lemma made explicit
Consider the case where $B = \emptyset$.  Then $y$ is the first vertex in $\sigma$ and immediately following $y$ are the vertices in $N(y)$.  If $y$ is not universal in $G$, then the set $Z = V(G) - N(y)$ is non-empty and in $\sigma$, all of its vertices appear after those in $N(y)$.  We claim that there is a join between $N(y)$ and $Z$.  Suppose for contradiction that there is no such join.  Then there is a vertex $w \in N(y)$ that is not universal to $Z$.  Consider some vertex $z \in Z$ such that $wz \notin E(G)$.  With $x$ and $y$ twins, it follows that $wx \in E(G)$.  Hence, $w <_\sigma z <_\sigma x$, and $wx \in E(G)$ but $wz \notin E(G)$.  Therefore, by Lemma~\ref{4vertex}, there is a vertex $d <_\sigma w$ such that $dz \in E(G)$ but $dx \notin E(G)$.  But $x$ is universal to $N(y)$ since it is $y$'s twin, and thus $d \notin N(y)$.  Given the restrictions on $\sigma$ noted above, it follows that $d = y$.  But then $dz \notin E(G)$, since $z \in Z$, providing the desired contradiction.  

That means there is a join between $N(y)$ and $Z$.  
%If $B = \emptyset$ (i.e. $y$ is the first vertex of $\sigma$), then to avoid $y$ being universal in $G$, $Z$, the set of vertices of $G$ not adjacent 
%to $y$ is not empty and thus $Z$ occurs after $N(y)$ in $\sigma$.  Since $x$ is the last vertex of $\sigma$ and is adjacent to precisely $N(y)$, all vertices in $Z$ must also be universal to $N(y)$, by Lemma~\ref{4vertex}.  
But now, unless $|N(y)| = 1$, $(N(y), Z \cup \{y\})$ is a split in $G$ contradicting $G$ being prime.  When $|N(y)| = 1$, if $|Z| = 1$, then $G$ is a star on three vertices and $G + x$ is a star on four vertices, where the penultimate vertex is a twin of $x$ (note that by Lemma~\ref{4vertex}, $xy \notin E$); if $|Z| > 1$, $(N(y) \cup \{y\}, Z)$ is a split in $G$.  Thus $B$ is not empty.

So let $s$ be the unique vertex of $B$. 
%Now let $S$ consist of $y$ followed by $Z_1 = N(y) \cap S$, followed by $Z_2$ (vertices of $S \cap G$ not adjacent to $y$), followed by $x$. 
Now $S$, as ordered by $\sigma$, consists of $y$ followed by $Z_1 = N(y) \cap S$, followed by  $Z_2$ (vertices of $S \cap V(G)$ not adjacent to $y$); $x$ is the last element of $S$
As $S\setminus\{y,x\}\neq\emptyset$, we have that $Z_1\cup Z_2\neq\emptyset$. Since $x$ and $y$ are twins, $x$ is universal to $Z_1$ and not adjacent to any vertices in $Z_2$.  Since $y$ is not universal in $G$, $|Z_2| >0$ and thus by Lemma~\ref{4vertex}, $x$ is not adjacent to $y$.  If $|Z_1| =0$, then unless $|Z_2| = 1$, $G$ has the split $(\{y,s\}, Z_2)$.  If $|Z_2| = 1$, then, as in the $B = \emptyset$ case, $G$ is a star on three vertices and $G+x$ is a star on four vertices, where the penultimate vertex is a twin of $x$.  Thus $|Z_1|> 0$ and $|Z_2| >0$.  By the application of Lemma~\ref{4vertex} to $(z_1, z_2, x)$, where $z_1 \in Z_1, z_2 \in Z_2$, we see that $z_1z_2 \in E(G)$ and thus $(\{s\} \cup Z_1, \{y\} \cup Z_2)$ is a split of $G$, contradicting $G$ being prime.
%-Derek 26/06
%Now let $S$ consist of $y$ followed by $Z_1 = N(y) \cap S$, followed by $Z_2$ (vertices of $S$ not adjacent to $y$), followed by $x$.  Since $x$ and $y$ are twins, $x$ is universal to $Z_1$ and not adjacent to any vertices in $Z_2$.  Since $y$ is not universal in $G$, $|Z_2| >0$ and thus by Lemma~\ref{4vertex}, $x$ is not adjacent to $y$.  If $|Z_1| =0$, then unless $|Z_2| = 1$, $G$ has the split $(\{y,s\}, Z_2)$.  If $|Z_2| = 1$, then, as in the $B = \emptyset$ case, $G$ is a star on three vertices and $G+x$ is a star on four vertices, where the penultimate vertex is a twin of $x$.  Thus $|Z_1|> 0$ and $|Z_2| >0$.  By the application of Lemma~\ref{4vertex} to $(z_1, z_2, x)$, where $z_1 \in Z_1, z_2 \in Z_2$, we see that $z_1z_2 \in E(G)$ and thus $(\{s\} \cup Z_1, \{y\} \cup Z_2)$ is a split of $G$, contradicting $G$ being prime.
%-Derek 26/06
%As $x$ appears last in $\sigma$, $B$ and $S$ partition $V(G+x)$; moreover, if $S = \{x,y\}$, then $y$ is the penultimate vertex in $\sigma$.  So assume that $|S| > 2$.  If $|B| > 1$, then $(B,S \setminus \{x\})$ is a split in $G$, contradicting $G$ being prime.  Thus, $|B| = 1$, meaning $B$ consists of a single vertex $s$ universal to $S$ and thus $s$ is universal in $G + x$.  Now, the proof proceeds as in the $B = \emptyset$ case.  Note that $s \in N(y)$ but that $N(y) \cap S$ could be empty.
\end{proof}

%EME below : useless explanation
%To show how to relate a LBFS ordering of $G =\G(T,\mathcal{F})$ to one for $G(u)$, with $u$ a node of $(T,\mathcal{F})$, we first need define an appropriate ordering on $V(u)$.

% MT 08/27/12 - Reviewer wanted to clarify that q \ne r either before or after the definition below.
Let $u$ be a node in a GLT $(T,\mathcal{F})$.  Notice that the sets $L(q), q \in V(u)$ partition the leaves of $T$.  In other words, each marker vertex can be associated with a distinct leaf in $T$.  This allows us to define a type of induced LBFS ordering on $G(u)$ as demonstrated below.

%EME below : not much useful in my opinion, but can be let
%The next lemma will show how graph labels in a GLT inherit, in some sense, the LBFS ordering of the entire graph.

%EME the notation $\sigma[G(u)]$ is a little confusing with $\sigma[S]$,  but why not. 

\begin{definition}
Let $u$ be a node of a GLT $(T,\mathcal{F})$ and let $\sigma$ be an ordering of $G=\G(T,\mathcal{F})$. 
For any marker vertex $p$, let $x_p$ be the earliest vertex of $A(p)$ in $\sigma$.
Define $\sigma[G(u)]$ to be the ordering of $G(u)$ such that for $q,r \in V(u)$, $q<_{\sigma[G(u)]} r$ if $x_q<_{\sigma} x_r$.
%and only if $A(q)$ contains a vertex appearing earlier in $\sigma$ than any vertex of $A(r)$.
\end{definition}

\begin{lemma} \label{inducedLBFS}
Let $\sigma$ be an LBFS ordering of a connected graph $G=\G(T,\mathcal{F})$, and let $u$ be  a node in $(T,\mathcal{F})$. Then $\sigma[G(u)]$ is an LBFS ordering of $G(u)$.
\end{lemma}

\begin{proof}
First observe that if we collect in a set $S$ one leaf $\ell_q$ of $A(q)$ for every marker vertex $q\in V(u)$, then the induced subgraph $G[S]$ is isomorphic to $G(u)$. Notice that $\sigma[G(u)]$ is then the ordering $\sigma[S]$ if each selected leaf $\ell_q$ is chosen to be the earliest in $\sigma$. We prove by induction on the number of nodes in $T$ that $\sigma[S]=\sigma[G(u)]$ is an LBFS ordering of $G(u)$. To that aim, we use Lemma~\ref{4vertex}.

As an induction hypothesis, assume the lemma holds for all graphs whose split-tree has fewer nodes than $ST(G)$.  
The lemma clearly holds if $(T,\mathcal{F})$ contains only one node, because $G[S]$ is isomorphic to $G$ in this case.

So assume that $(T,\mathcal{F})$ contains more than one node.  Then there is a $q \in V(u)$ such that $T(q)$ contains at least one node.  Let $\ell_q \in A(q)$ be the leaf associated with $q$ in $\sigma[G(u)]$.  Let $G' = G[(V(G)\setminus L(q))\cup\{\ell_q\}]$.  Remove $T(q)$ from $(T,\mathcal{F})$, choosing $\ell_q$ to be the leaf that replaces its nodes; let $(T',\mathcal{F}')$ be the resulting GLT.  Clearly $G'=\G(T',\mathcal{F}')$. For simplicity, let $\sigma' = \sigma[V(G')]$. Suppose $a$, $b$ and $c$ form a triple of vertices of $V(G')$ as in Lemma~\ref{4vertex}. As $\sigma$ is an LBFS ordering of $G$, there exists $d\in V(G)$ appearing earlier than $a$ in $\sigma$ which is adjacent to $b$ but not to $c$. Suppose that $d$ does not belong to $V(G')$, i.e. $d\neq \ell_q$ and $d \in L(q)$. Let $p$ be $q$'s opposite in $(T,\mathcal{F})$. As $(L(q),L(p))$ is a split of $G$, the vertex $b$ either belongs to $L(p)$ or $L(q)$.  In the former case, since $b$ is adjacent to a vertex in $L(q)$, $b \in A(p)$ and thus $d \in A(q)$.  By the choice of $\ell_q$, it can replace vertex $d$.  
%DGC3 - added sentence:
(Note that $c \in L(p) \setminus A(p)$ and thus $l_q$ and $c$ are not adjacent.)
In the latter case, since $\ell_q$ is the only $L(q)$ vertex in $G'$, $b = \ell_q$ and $b \in A(q)$.
Moreover, by the choice of $\ell_q$, $d$ belongs to $L(q)\setminus A(q)$. We now prove that $a$ cannot appear before $b$ in $\sigma$, yielding a contradiction. As $b$ is the only vertex of $L(q)$ present in $V(G')$, so vertex $a$ belongs to $L(p) \setminus A(p)$. By the choice of $\ell_q$, no vertex of $A(q)$ appears before $a$ in $\sigma$. 
By Remark~\ref{prefix}, the subgraph of $G$ induced on the vertices of $\sigma$ up to, and including $a$ is connected.  But, there can be no path in the subgraph connecting $d \in L(q) \setminus A(q)$ and $a$ since $A(q)$ is a separator for $d$ and $a$, and $b = \ell_q$ is the earliest vertex of $A(q)$ in $\sigma$. Thus $a$ cannot appear before $b$ in $\sigma$, thereby contradicting the existence of the triple $\{a,b,c\}$.
It follows that $\sigma'$ is an LBFS ordering of $G'$.
% since any triple of vertices such as $a$, $b$ and $c$ satisfies the characterization of lemma~\ref{4vertex}.

Of course, $(T',\mathcal{F}')$ has fewer nodes than $(T,\mathcal{F})$.  We can therefore apply our induction hypothesis.  Hence, $\sigma[S]$ is an LBFS ordering of $G'[S]$.  But notice that $G'[S]$ is isomorphic to $G[S]$ which is isomorphic to $G(u)$. The induction step follows.
\end{proof}

%EME corollary below : too obvious reformulation  of lemma above, and it served only once in what follows !
%\begin{corollary} \label{nodeLBFS}
%Let $\sigma$ be an LBFS ordering of a connected graph $G$, and let $u$ be a node in $ST(G)$.  Then $\sigma[G(u)]$ is an LBFS ordering of $G(u)$. 
%\end{corollary}

%----------------------------------------------------------------------------------------------------------------------
%\newpage
%----------------------------------------------------------------------------------------------------------------------
%----------------------------------------------------------------------------------------------------------------------
\section{Incremental split decomposition} \label{sec:characterization}

Throughout this section we assume that the graphs $G$ and $G + x$ are both connected.  
We provide a simple combinatorial description of the updates required in $ST(G)$ to arrive at $ST(G+x)$.
%
%The proof is obtained by an analysis of the properties of $ST(G)$ when removing $x$ from $ST(G+x)$, which turns out to be easily inversible. 
%VAR
The proof is obtained by a case by case analysis of the properties of $ST(G)$ when removing $x$ from $ST(G+x)$, which turns out to be easily invertible. 
 
%EME put in section 5
%As mentioned previously, the implementation of our algorithm, developed in section~\ref{sec:algo}, assumes that the given graph is connected and that the vertices are added following an LBFS ordering.  Thus from Remark~\ref{prefix} all iterations of the Algorithm satisfy the condition that $G$ and $G + x$ are connected.

% NEW 9 END

%----------------------------------------------------------------------------------------------------------------------
\subsection{State assignment}

Most results in the paper rely on the next definition.
%
%EME precision below not necessary, inherited from previous versions
Intuitively, its aim is to allow a characterization of the portions of the split-tree that change or fail to change under the insertion of a new~vertex.
%Many of the results in the rest of the paper will be phrased in terms of \emph{states} assigned to leaves and marker vertices, as defined here below:

\begin{definition} \label{def:states}
Let $(T,\mathcal{F})$ be a GLT, and let $q$ be one of its leaves or marker vertices.  Let $S$ be a subset of $T$'s leaves.  Then the \emph{state (with respect to} $S$\emph{) of} $q$ is: 

- \emph{perfect} if $S \cap L(q) = A(q)$; 

- \emph{empty} if $S \cap L(q) = \emptyset$; 

- and \emph{mixed} otherwise.  

For a node $u$, define the sets $P(u) = \{ q \in V(u)~ |~ q \textrm{ perfect}\}$, $M(u) = \{ q \in V(u)~ |~ q \textrm{ mixed}\}$, and $NE(u) = P(u) \cup M(u)$ (``NE'' for ``Not-empty'').  See Figure~\ref{fig:stateExample}.
\end{definition}

\begin{figure}[htbh]
\begin{center}
\includegraphics[scale=0.80]{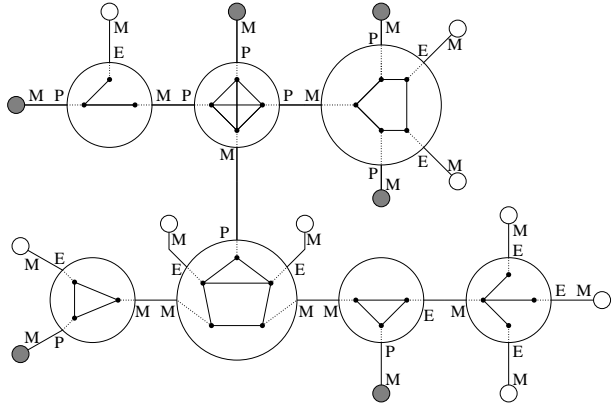}
\end{center}
%\vspace{-0.2in}
\caption{A GLT and states assigned according to the shaded leaves (``P'' for ``perfect'', ``M'' for ``mixed'', and ``E'' for ``empty''.)}
\label{fig:stateExample}
\end{figure}

\begin{remark}\label{rk:inherited-states}
The state of a marker vertex before and after a node-split or a node-join is the same.
\end{remark}

%EME useless
%The $P$ stands for ``perfect'', the $M$ for ``mixed'', and the $NE$ for ``not-empty''.  
Notice that the opposite of any leaf $l$ must be either perfect (if $l\in S)$ or empty (if $l\not\in S)$.  We extend the state definition to subtrees: if a marker vertex (or leaf) $q$ is perfect (respectively empty, mixed), then the \emph{subtree $T(q)$ is perfect} (respectively \emph{empty, mixed}) as well. 
%EME removed, a voir dans intro de section
%Intuitively, perfect and empty subtrees define those portions of $ST(G)$ that remain unchanged in $ST(G+x)$ (a formal proof of this assertion is provided later). 
%
%
%
%
%
%Finally, we will also use the following definitions.
%
%\begin{definition}
%Let $(T,\mathcal{F})$ be a GLT. 
A node $u$ of $T$ is \emph{hybrid} if every marker vertex $q\in V(u)$ is either perfect or empty and $q$'s opposite is mixed. 
A tree-edge $e$ of $T$ is  \emph{fully-mixed} if both of its extremities are mixed. 
% MT 08/27/12 - Reviewer wanted "maximal" as it made more sense
A \emph{fully-mixed subtree} $T'$ of $T$ 
%A subtree $T'$ of $T$ is \emph{fully-mixed} if it 
is one that contains at least one tree-edge, all of its tree-edges are fully-mixed, and it is maximal for inclusion with respect to this property.
%\end{definition}   
% 
For a degenerate node $u$, we denote:
\begin{eqnarray*}
P^*(u) & =  & \{ q \in V(u) | \; q \textrm{ perfect and not the centre of a star} \}, \\
E^*(u) & = & \{ q \in V(u) | \; q \textrm{ empty, or } q \textrm{ perfect and the centre of a star} \}. 
\end{eqnarray*}
We now describe some basic properties.
% of a GLT \emph{marked with respect to a subset $S$ of its leaves}. 
The first key lemma follows directly from Remark~\ref{rem:hered}, and implies the subsequent corollary. 
%the next results follow from it and from one another.
%EME smplified
%We now describe some basic properties of a GLT \emph{marked with respect to a subset $S$ of its leaves}. The first follows directly from Remark~\ref{rem:hered}, each other follows from the previous one.
%EME useless
%These properties will help to prove our combinatorial incremental characterization of the split-tree. 
%EME useless
%Recall that, thanks to definition~\ref{def:descendant}, a leaf belonging to the subtree $T(q)$ for a marker vertex $q$ is not a descendant of $q$.

\begin{lemma}[Hereditary property] \label{lem:hereditary}
Let $(T,\mathcal{F})$ be a GLT marked with respect to a subset of leaves $S$. Then
\begin{enumerate}
\item a marker vertex $q$ is perfect if and only if every accessible descendant of $q$ is perfect and every non-accessible descendant of $q$ is empty.
\item a marker vertex $q$ is empty if and only if every descendant of $q$ is empty.%\hfill$\square$
\end{enumerate}

\end{lemma}
%\begin{proof}
%This is a direct consequence of Remark~\ref{rem:hered}.
%\end{proof}

%\begin{corollary} \label{lem:PE-descendant}
%Let $(T,\mathcal{F})$ be a GLT marked with respect to $S$. A marker vertex  $q$ descendant of a perfect or empty marker vertex $p$ cannot be mixed.
%\hfill$\square$
%\end{corollary}

\begin{corollary} \label{cor:hereditary}
Let $(T,\mathcal{F})$ be a GLT marked with respect to a subset of leaves $S$. 

\begin{enumerate}
%\renewcommand{\labelenumi}{\roman{enumi}.}
%EME better variant
%\item A marker vertex  descendant of a perfect or empty marker vertex  cannot be mixed.

\item If marker vertex $q$ is mixed, then every marker vertex having $q$ as a descendant is mixed.

\item If a tree-edge has a perfect or empty extremity $q$  with a mixed opposite, then, for every tree-edge in $T(q)$, the extremity that is a descendant of $q$ is perfect or empty and its opposite is mixed.

\item If there exists a hybrid node, then it is unique.

\item  In a clique node, if every marker vertex is perfect, then every opposite of a marker vertex is also perfect.

\item In a star node, if every marker vertex is empty, except the centre which is perfect, then every opposite of a marker vertex is perfect, except the opposite of the centre which is empty.
%\hfill$\square$
\end{enumerate}

\end{corollary}

%\begin{proof}
%Follows from Lemma~\ref{lem:hereditary}.
%\end{proof}

%EME corollary below used once in christophe's proof of main theorem
%\begin{corollary} \label{lem:PP-PE-edge}
%Let $(T,\mathcal{F})$ be a GLT marked with respect to $S$. If $T$ contains a tree-edge $e$ the extremities of which are either perfect or empty, then every tree-edge has a perfect or empty ex-\break tremity.
%\hfill$\square$
%\end{corollary}

%\begin{proof}
%Follows from Lemma~\ref{lem:PE-descendant}.
%\end{proof}

%----------------------------------------------------------------------------------------------------------------------
\subsection{Lemmas deriving $ST(G)$ from $ST(G+x)$}

This subsection is devoted to technical lemmas, which aim to enumerate and characterize in terms of states all possible cases for the deletion of $x$ from $ST(G+x)$. Their proofs rely on an extensive use of the hereditary property (Lemma~\ref{lem:hereditary}) and Corollary~\ref{cor:hereditary}. These lemmas will only be used in the proofs of the next subsection to describe how to update $ST(G)$ when inserting a new vertex. 

\medskip
In this subsection, we let $u$ denote the node of $ST(G+x)$ to which the leaf $x$ is attached. 
Let  $(T_x,{\mathcal{F}}_x)$ be the GLT, obtained from $ST(G+x)$ by removing leaf $x$ and $x'$ its opposite marker vertex in the label of $u$, and let $u_x$ be the node corresponding to $u$ in $(T_x,\mathcal{F}_x)$, such that $G(u_x)=G(u)-x'$.  Note that the accessibility graph of $(T_x,\mathcal{F}_x)$ is $G$. For convenience, but contrary to the definition, the GLT $(T_x,\mathcal{F}_x)$ is allowed to have a binary node $u_x$ in the case where $u$ was ternary; in this case, ``{\it contraction of $u_x$ to $e$}'' refers to  the operation of replacing $u_x$ and its two adjacent tree-edges by a single tree-edge $e$.
To simplify, we may identify a marker vertex in $(T_x,\mathcal{F}_x)$ with the corresponding marker vertex in $ST(G+x)$. Finally, we assume that $ST(G)$ and $(T_x,\mathcal{F}_x)$ are marked with respect to $S=N_G(x)$. Notice we consider $x$ to have the perfect state and thus the states of the descendants of $x$ in $(T_x,\mathcal{F}_x)$ are determined to be either perfect or empty
by applying Lemma \ref{lem:hereditary}-1 in $ST(G+x)$. To shorten statements, a tree-edge is said to be $PP$, $PE$, $PM$, $EM$, or $MM$ (i.e. fully-mixed),  depending on the states of its two extremities, where $P$, $E$, and $M$, stands respectively for perfect, empty, and mixed.

In the following subsections, we deal with all possibilities of $u$, the node in $ST(G+x)$ to which $x$ is adjacent.

\subsubsection{$u$ is a clique}

\begin{lemma} \label{lem:x_to_clique_states}
Assume $x$ is adjacent to a clique $u$ in $ST(G+x)$. Then every tree-edge of $(T_x,\mathcal{F}_x)$ incident to $u_x$ is $PP$, and every other edge in $(T_x,\mathcal{F}_x)$ is either $PM$ or $EM$.
%Assume $x$ is adjacent to a clique node $u$ in $ST(G+x)$ and let $e$ be tree-edge of $(T_x,\mathcal{F}_x)$ marked with respect to $S$. If $e$ is incident to $u_x$, then $e$ is $PP$, otherwise $e$ is either $PM$ or $EM$.
\end{lemma}

\begin{proof}
Every marker vertex of $u_x$ is a descendant of $x$ in $ST(G+x)$ and hence it is perfect by the hereditary property (Lemma~\ref{lem:hereditary}-1). Then by Corollary~\ref{cor:hereditary}-4, every opposite $p$ of a marker vertex $t$ of $u_x$ is perfect. So every tree-edge incident to $u_x$ is $PP$. 

Let $v$ be a node adjacent to $u_x$ by the tree-edge $e$, and let $p$ and $t$ be respectively the extremities of $e$ in $v$ and in $u_x$. Let $r$ be the opposite of a marker vertex $q$ of $v$ distinct from $p$. Observe that $T(r)$ contains the node $u_x$ and thus $r$ has a perfect descendant. So by the hereditary property (Lemma~\ref{lem:hereditary}-2), $r$ cannot be empty. 

We now prove that if $r$ is perfect then, by the hereditary property (Lemma~\ref{lem:hereditary}-1), $v$ is a clique node. Observe first that Lemma~\ref{lem:hereditary}-1 applied on $r$ and $p$ implies that $p$ and $q$ are adjacent. Since $G(v)$ is connected and contains at least $3$ marker vertices, $v$ contains a marker vertex distinct from $p$ and $q$ adjacent to at least one of 
$p$ or $q$. As every such vertex $s$ is a descendant of $t$ and $r$ (both being perfect),  Lemma~\ref{lem:hereditary}-1 implies that $s$ is adjacent to both $p$ and $q$. It follows that either $(\{p,q\}, V(G(v))\setminus \{p,q\})$ forms a split in $G(v)$ or $v$ is ternary.  Since $ST(G+x)$ is reduced, in both cases $v$ is degenerate and by the adjacencies between $p$, $q$ and 
$s$, $v$ is a clique node. 

So we proved that if $r$ is perfect, then $ST(G+x)$ contains two adjacent clique nodes: contradiction.
It follows that $r$ is mixed.  By the hereditary property (lemma~\ref{lem:hereditary}-1), since $t$ is perfect, $q$ is either perfect or empty.  Hence, every 
tree-edge not incident to $u_x$ is $PM$ or $EM$ by Corollary \ref{cor:hereditary}-2.
\end{proof}

\begin{lemma} \label{lem:x_to_clique}
Assume $x$ is adjacent to a clique node $u$ in $ST(G+x)$.
\begin{enumerate}
\item If $u$ is ternary, let $(T,\mathcal{F})$ be the GLT resulting from the contraction of $u_x$ to $e$ in $(T_x,\mathcal{F}_x)$.
\begin{enumerate}
\item If $(T,\mathcal{F})$ is reduced then $ST(G)=(T,\mathcal{F})$ and $e$ is the unique $PP$ tree-edge of $ST(G)$, every other tree-edge is $PM$ or $EM$.

\item Otherwise, $ST(G)$ results from the star-join in $(T,\mathcal{F})$ of the nodes incident to $e$ (let $v$ be the resulting node). Then $v$ is the unique hybrid node of $ST(G)$, and every tree-edge is $PM$ or $EM$.
\end{enumerate}
%obtained by contracting $u_x$ to $e$ in $(T_x,\mathcal{F}_x)$ is reduced, then $ST(G)$ equals this GLT. In $ST(G)$ marked with respect to $S$, $e$ is its unique $PP$ tree-edge, and every other tree-edge is $PM$ or $EM$.

%\item If $u$ is ternary and the GLT obtained by contracting $u_x$ to $e$ in $(T_x,\mathcal{F}_x)$ is not reduced, then $ST(G)$ is obtained from this GLT by performing a star-join on $e$.  Let $v$ be the resulting node. In $ST(G)$ marked with respect to $S$, $v$ is the unique hybrid node, and every tree-edge is $PM$ or $EM$.

\item If $u$ is not ternary then $ST(G)=(T_x,\mathcal{F}_x)$, $u_x$ is the unique clique node of $ST(G)$
whose marker vertices are all perfect, the tree-edges incident to $u_x$ are $PP$ and every other tree-edge is $PM$ or $EM$.
\end{enumerate}
\end{lemma}

\begin{proof}
The correctness of the construction of $ST(G)$ follows  directly from the definition of the split-tree since the involved operations preserve the accessibility graph $G$ and yield a reduced GLT. The state properties of tree-edges come directly from Lemma \ref{lem:x_to_clique_states}
since, by Remark \ref{rk:inherited-states}, states of marker vertices are preserved by the involved operations.  This conclude the proof for cases 1(a) and 2 since uniqueness follows in both cases from Lemma~\ref{lem:x_to_clique_states}.

In case 1(b), let $p$ and $q$ denote the two marker vertices of $u_x$. Observe that as $u$ is a clique node and as $(T,\mathcal{F})$ is not reduced, the two neighbours $v_1$ and $v_2$ of $u_x$ are star nodes such that the centre of $v_1$ is the opposite of $p$, whereas the centre of $v_2$ is not the opposite of $q$. It follows that $ST(G)$ results from a star-join of $v_1$ and $v_2$. Note that the node $v$ (resulting from the star-join) inherits from $v_1$ and $v_2$ the descendants of $x$ in $ST(G+x)$. It follows by the hereditary property (Lemma~\ref{lem:hereditary}-1), that the marker vertices of $v$ are perfect or empty. Finally observe that 
$v$ contains empty and perfect marker marker vertices: the non-centre marker vertices inherited from $v_2$ are empty, all the others are perfect. It follows that $v$ is a hybrid node and it is unique (by Corollary  \ref{cor:hereditary}-3).
\end{proof}

\subsubsection{$u$ is a star node}

\begin{lemma} \label{lem:x_to_star_states}
Assume $x$ is adjacent to a star node $u$ in $ST(G+x)$. Then every tree-edge of $(T_x,\mathcal{F}_x)$ incident to $u_x$ is $PE$, and every other edge in $(T_x,\mathcal{F}_x)$ is either $PM$ or $EM$.
\end{lemma}

\begin{proof}
Let $c$ be the centre of the star $G(u)$. Since $G$ is connected, $x'$ the opposite of $x$ is a degree-1 marker vertex of $G(u)$. It follows that $S=A(c)$, and thus, $c$ is perfect and, by Corollary \ref{cor:hereditary}-5, its opposite $p$ is empty. Now let $q$ be a marker vertex of $u_x$ distinct from $c$ and let $r$ be its opposite.
By the hereditary property (Lemma~\ref{lem:hereditary}-2), as a descendant of $p$, $q$ is empty. By Corollary \ref{cor:hereditary}-5, $r$ is perfect. So we proved that every tree-edge incident to $u_x$ is $PE$. 

We now prove that every tree-edge non-incident to $u_x$ is either PM or EM. Let $v$ be a node adjacent to $u_x$ by the tree-edge $e$, and let $p$ and $t$ be respectively the extremities of $e$ in $v$ and in $u_x$. Let $r$ be the opposite of a marker vertex $q$ of $v$ distinct from $p$.

Assume first that $t\neq c$. Then by Lemma~\ref{lem:hereditary}-2, since $c$ is a perfect descendant of $r$, $r$ is not empty. So suppose for contradiction that $r$ is perfect. Observe first that Lemma~\ref{lem:hereditary}-1 applied to $r$ and $p$ implies that $p$ and $q$ are adjacent and that by Lemma~\ref{lem:hereditary}-2, as a descendant of $t$, $q$ is empty. Since $G(v)$ is connected and contains at least $3$ marker vertices, $v$ contains a marker vertex $s$ distinct from $p$ and $q$ adjacent to 
at least one of $p$ or $q$. As $S=A(c)$, we have that $L(s)\cap S=\emptyset$ implying that $s$ is empty. As $s$ is a descendant of $r$, by Lemma~\ref{lem:hereditary}-1, $s$ is not adjacent to $q$ and thereby it is adjacent to $p$. It follows that in $G(v)$, the marker vertex $q$ has degree one. Then $v$ has to be a star node whose centre is $p$: this contradicts the fact that $ST(G+x)$ is reduced. It follows that $r$ is mixed (it can not be perfect or empty). 

Assume now that $t=c$. If $r$ is empty, by definition $L(r)\cap S=\emptyset$.  
For every neighbour $q'$ of $p$, there exists a leaf accessible from $c$ in $T(q')$, and hence an element of $S$ is in $T(q')$.
But now, for every $q' \ne q$, $T(q') \subseteq T(r)$ which contradicts $L(r)\cap S=\emptyset$.  Thus $q$
 is the only neighbour of $p$ in $G(v)$, and $v$ has to be a star node whose centre is $q$; this contradicts the fact that $ST(G+x)$ is reduced. So $r$ is perfect or mixed. Now assume $r$ is perfect. Since $p$ is an empty descendant of $r$, by 
 Lemma~\ref{lem:hereditary}-1, $p$ and $q$ are not adjacent in $G(v)$. Since $G(v)$ is connected and contains at least $3$ marker vertices, $v$ contains a marker vertex distinct from $p$ and $q$ adjacent to at least one of $p$ or $q$. As every such vertex $s$ is a descendant of $r$ and $c$, both being perfect, Lemma~\ref{lem:hereditary}-1 implies that $s$ is adjacent to $p$ and $q$. It follows that either $(\{p,q\}, V(G(v))\setminus \{p,q\})$ forms a split in $G(v)$ or $v$ is ternary. In both cases, $v$ is degenerate and by the adjacencies between $p$, $q$ and $s$, $v$ is a star node whose centre is $s$. This contradicts the fact that $ST(G+x)$ is reduced. It follows in this case also that $r$ is mixed (it can not be perfect or empty). 

So $r$ is always mixed and $q$ is perfect or empty by the hereditary property (Lemma \ref{lem:hereditary}-1 applied to the marker vertices of $u_x$). Then, every tree-edge not incident to $u_x$ is $PM$ or $EM$ by Corollary \ref{cor:hereditary}-2.
\end{proof}

\begin{lemma} \label{lem:x_to_star}
Assume $x$ is adjacent to a star node $u$ in $ST(G+x)$.
\begin{enumerate}
\item If $u$ is ternary, let $(T,\mathcal{F})$ be the GLT resulting from the contraction of $u_x$ to $e$ in $(T_x,\mathcal{F}_x)$.
\begin{enumerate}
\item If $(T,\mathcal{F})$ is reduced then $ST(G)=(T,\mathcal{F})$ and $e$ is the unique $PE$ tree-edge of $ST(G)$, and every other tree-edge is $PM$ or $EM$.

\item Otherwise, $ST(G)$ results from the star-join or a clique-join in $(T,\mathcal{F})$ of the nodes incident to $e$ (let $v$ be the resulting node). Then $v$ is the unique hybrid node of $ST(G)$, and every tree-edge is $PM$ or $EM$.
\end{enumerate}

\item If $u$ is not ternary then $ST(G)=(T_x,\mathcal{F}_x)$, $u_x$ is its unique star node
whose marker vertices are all empty except the centre, which is perfect, tree-edges adjacent to $u_x$ are $PE$, and all other tree-edges are $PM$ or $EM$.
\end{enumerate}
\end{lemma}

\begin{proof}
The proof follows the same lines as the proof of Lemma \ref{lem:x_to_clique}, using Lemma \ref{lem:x_to_star_states} instead of  Lemma \ref{lem:x_to_clique_states}. Only the arguments to show that $v$ is a hybrid node in case 1(b) differ slightly.

So assume case 1(b) holds. As $u$ is ternary and $(T,\mathcal{F})$ is not reduced, the two neighbours $v_1$ and $v_2$ of $u$ are either clique nodes or star nodes. Suppose that the centre $c$ of $u$ is the extremity of the tree-edge $uv_1$. Let: $p_1$ be the opposite of $c$; $p_2$ be the marker vertex of $v_2$'s  extremity of the tree-edge $uv_2$; and $q$ be the opposite of $p_2$. Observe that 
$p_1$ is universal in $G(v_1)$:  this is trivial if $G(v_1)$ is a clique node; if $G(v_1)$ is a star node, then the fact that $ST(G+x)$
is reduced implies that $p_1$ is the centre of the star.  Now, due to their adjacency with $x'$ (the opposite of $x$), 
 $c$ is perfect and $q$ is empty.
It follows by Lemma~\ref{lem:hereditary}-1, that the marker vertices of $v_1$ distinct from $p_1$ are perfect (since they are accessible descendants of $c$). Similarly, by Lemma~\ref{lem:hereditary}-2, the marker vertices of $v_2$ distinct from $p_2$ are empty (since they are descendant of $q$). As all these marker vertices are inherited by $v$, $v$ is an hybrid node which is unique (by Corollary  \ref{cor:hereditary}-3).
\end{proof}

\subsubsection{$u$ is a prime node}

Here we have two cases depending on whether or not $G(u_x)$ is also prime.

\begin{lemma} \label{lem:x_to_prime_hybrid}
Assume that $x$ is adjacent to a prime node $u$ in $ST(G+x)$ such that $G(u_x)$ is also prime.
Then $ST(G)=(T_x,\mathcal{F}_x)$, $u_x$ is its unique hybrid node, and every tree-edge is $PM$ or $EM$.
\end{lemma}

\begin{proof}
Observe that by construction, every marker vertex of $u_x$ is either perfect or empty. Let $p$ be the opposite of a marker vertex $t$ of $u_x$. Assume that $p$ is perfect. Then, in $G(u)$, $t$ is a twin of $x'$, the opposite of $x$: contradicting the fact that $u$ is a prime node of $ST(G+x)$. So assume that $p$ is empty. Then, by the hereditary property (Lemma \ref{lem:hereditary}-2), every marker vertex $q$ distinct from $t$ in $u_x$ is empty. It follows that $A(t)=S$ and thereby $t$ is the only neighbour of $x'$ in $G(u)$. This is again a contradiction with the fact that $u$ is a prime node of $ST(G+x)$, since a prime graph does not contain pendant vertices. Thus $p$ is mixed and now the proof follows from Corollary \ref{cor:hereditary} (-2 and -3).
\end{proof}

It remains to analyse the case where $x$ is adjacent to a prime node $u$ in $ST(G+x)$ but $G(u_x)$ is not prime. To that aim, we describe a three step construction that computes $ST(G)$ from $ST(G+x)$.  Note that this construction is not part of the
Split Decomposition Algorithm itself.

Let us first recall that when a node-join or a node-split has been performed on an initial GLT, then a marker vertex is \emph{inherited} by the resulting GLT if its corresponding tree-edge is not affected by the operation. We say that a tree-edge $e=uv$ of a GLT is  \emph{non-reduced} if a node-join on $u$ and $v$ yields a star node or a clique node. 

The announced construction is the following; it uses $(T_x,\mathcal{F}_x)$ as input:
\begin{enumerate}
\item While the current GLT contains a node $v$ which is neither prime nor degenerate, find a split in $G(v)$ and perform the node-split accordingly. 
\item While the resulting GLT contains a non-reduced tree-edge $e$ both extremities of which are not inherited from $(T_x,\mathcal{F}_x)$, perform the corresponding node-join. Let $(T'_x,\mathcal{F}'_x)$ be the resulting GLT.
\item While the current GLT contains a non-reduced tree-edge, perform the corresponding node-join. Let $(T,\mathcal{F})$ be the resulting GLT.
\end{enumerate}

% MT 08/27/12 - Reviewer wanted above construction named and referenced later for continuity.  All future references to this construction (that I could find) have been updated without further comment.
The rest of the results in this section should be interpreted in the context of the above construction, which will subsequently be referred to as the \emph{prime-splitting construction} for simplicity.  The following observation concerning the prime-splitting construction follows from the fact that a node-join and a node-split do not change the accessibility of a GLT and that $(T,\mathcal{F})$ is clearly reduced.

\begin{observation} \label{obs:st(g)}
The GLT $(T,\mathcal{F})$ resulting from the prime-splitting construction is the split-tree $ST(G)$.
\end{observation}

Intuitively, the GLT $(T'_x,\mathcal{F}'_x)$ is obtained from $T_x, \mathcal{F}_x$ by replacing the node $u_x$ with the split-tree $ST(G(u_x))$. Such a replacement is obtained by accurately identifying the leaves of $ST(G(u_x))$ with the marker vertices opposite the marker vertices of $u_x$.  Indeed, note that in the case where $G+x$ is prime, but not $G$, then $ST(G)=(T'_x,\mathcal{F}'_x)=(T,\mathcal{F})$.
To help the intuition of the following lemmas, we state the summarizing lemma (Lemma \ref{lem:x_to_prime}) in the context of this special case.

We will now describe the properties of $ST(G)$ and of the intermediate GLT $(T'_x,\mathcal{F}'_x)$  in terms of the states of their marker vertices. Recall that by Remark~\ref{rk:inherited-states}, the states of inherited marker vertices remain unchanged. Also observe that after a series of node-joins and node-splits, a tree-edge $e$ of the resulting GLT has its two extremities either both inherited or both non-inherited. In the former case, $e$ is an \emph{inherited tree-edge}, in the latter case a \emph{non-inherited tree-edge}. Finally observe that if $e$ is a non-inherited tree-edge, then it corresponds to a split $(A_x,B_x)$ of $G(u_x)$ since it results from the first and second steps of the prime-splitting construction. Intuitively, the non-inherited tree-edges correspond to the internal tree-edges of $ST(G(u_x))$.
% (see Corollary~\ref{cor:x_to_prime_states}).

\begin{lemma} \label{lem:x_to_prime_states}
Consider the prime-splitting construction.  Assume that $x$ is adjacent to a prime node $u$ in $ST(G+x)$.
% such that $G(u_x)$ is not prime.
Then every non-inherited tree-edge of $(T'_x,\mathcal{F}'_x)$ is MM and every inherited tree-edge is PM or EM.
\end{lemma}
\begin{proof}
Note that if $G(u_x)$ is prime, the result is trivial.
We first prove that inherited tree-edges are PM or EM. To that aim we first argue that every opposite $q$ of a marker vertex $p$ of $u_x$ is mixed in $(T_x,\mathcal{F}_x)$ marked with respect to $S$. Observe that $q$ cannot be perfect, since otherwise $p$ and $x'$ are twins in $G(u)$, contradicting node $u$ being a prime node (a prime graph cannot have a pair of twins). So assume that $q$ is empty. Then $p$ cannot be empty since otherwise we would have $S=\emptyset$ (as $L(q)=L(p)=\emptyset$). If $p$ is perfect, then $p$ has degree $1$ in $G(u)$ (since $L(q)=\emptyset$). Thus $p$ is a pendant vertex of $G(u)$: contradiction, a prime graph cannot have a pendant vertex. It follows by Corollary~\ref{cor:hereditary}-2 that every tree-edge of $(T_x,\mathcal{F}_x)$ is PM or EM. Now by Remark~\ref{rk:inherited-states},  the state of the inherited marker vertices are preserved under node-joins and node-splits. Thus every inherited tree-edge of $(T'_x,\mathcal{F}'_x)$ is PM or EM.

We now deal with non-inherited tree-edges. Let $p$ be the extremity of such a tree-edge $e$ in $(T'_x,\mathcal{F}'_x)$. Denote by $(A,B)$ the split of $G$ corresponding to $e$ with $L(p)=B$. As noticed before, $e$ also corresponds to a  split $(A_x,B_x)$ of $G(u_x)$ (we see that $A=\cup_{q\in A_x} L(q)$ and $B=\cup_{q\in B_x} L(q)$). We prove that if $p$ is not mixed, then $G(u)$ contains a split, a contradiction with $u$ being a prime node. So assume first that $p$ is empty. Then by definition $L(p)\cap S=B\cap S=\emptyset$. It follows that the bipartition $(A\cup\{x\},B)$ is a split of $G+x$ and thus $(A_x\cup\{x'\},B_x)$ is a split of $G(u)$. Assume now that $p$ is perfect, then $L(p)\cap S=B\cap S=A(p)$. It follows that  $(A\cup\{x\},B)$ is a split of $G+x$ (here $x$ belongs to the frontier of $A$). Thereby $(A_x\cup\{x'\},B_x)$ is again a split of $G(u)$.  Thus $p$ is mixed, as required.
\end{proof}

%\begin{corollary} \label{cor:x_to_prime_states}
%If $G+x$ is prime but not $G$, then every internal tree-edge of $ST(G)$ is MM.
%\end{corollary}

\begin{lemma} \label{lem:x_to_prime_cleaning}
Consider the prime-splitting construction.  Assume that $x$ is adjacent to a prime node $u$ in $ST(G+x)$.
% such that $G(u_x)$ is not prime. 
Let $w$ be a degenerate node incident to a non-inherited tree-edge in $(T'_x,\mathcal{F}'_x)$.
\begin{enumerate}
\item If $w$ is a star, the centre of which is perfect, then $w$ has no empty marker vertex and at most two perfect marker vertices. 
\item Otherwise $w$ has at most one empty marker vertex and at most one perfect marker vertex.
\end{enumerate}
\end{lemma}
\begin{proof}
First observe that the result is trivial if $G(u_x)$ is prime.  Note 
that every empty or perfect marker vertex $p$ is inherited. Otherwise $p$ would be the extremity of a tree-edge $e$ corresponding to a split $(A_x,B_x)$ of $G(u_x)$ and being empty or perfect would imply that $(A_x\cup\{x\},B_x)$ is a split of $G(u)$, contradicting $u$ being prime.
%that by construction $G(u_x)$ is the graph that would label to node resulting from the series node-joins on every pair $v$, $v'$ of nodes of $(T'_x,\mathcal{F}'_x)$ such that the tree-edge $vv'$ is non-inherited.
\begin{enumerate}
\item \emph{$w$ is a star node, the centre $c$ of which is perfect:} Suppose that $w$ has an empty marker vertex $q$ (distinct from $c$). As $q$ is inherited from $u_x$ and not accessible from every marker vertex inherited from $u_x$, $q$ is a pendant vertex in $G(u)$: contradicting $u$ being a prime node. So no marker vertex of $w$ is empty. Suppose now that $w$ has two perfect marker vertices $p$ and $p'$ distinct from $c$. Again $p$ and $p'$ are inherited from $u_x$. Moreover, every inherited vertex from $u_x$ accessible to $p$ is accessible to $p'$ (and vice versa). It follows that $p$ and $p'$ form a pair of twins in $G(u)$:  contradicting $u$ being a prime node. So $w$ contains at most two perfect marker vertices (including $c$).

\item \emph{otherwise:} Suppose that $w$ is a clique node containing two perfect (or two empty) marker vertices $p$ and $p'$. Again $p$ and $p'$ are inherited from $u_x$ and have the same accessibility set among the inherited marker vertices of $u_x$. Thereby $p$ and $p'$ are twins in $G(u)$: contradicting $u$ being a prime node.

Assume that $w$ is a star node, the centre $c$ of which is not perfect. If $c$ is mixed, then the same argument as for the clique node proves the property. So assume that $c$ is empty. Then the same argument as in case 1 applies. If $w$ contains an empty marker vertex $q$ distinct from $c$, then $q$ is pendant in $G(u)$. If $w$ contains two perfect marker vertices $p$ and $p'$ (distinct from $c$ by hypothesis), then $p$ and $p'$ are twins in $G(u)$. Both cases lead to a contradiction.
\end{enumerate}
\end{proof}

%\begin{corollary} \label{cor:x_to_prime_cleaning}
%Assume that $G+x$ is prime but not $G$. Let $w$ be a degenerate node of $ST(G)$.
%\begin{enumerate}
%\item If $w$ is a star node, the centre of which is perfect, then $w$ has no empty marker vertex and at most two perfect marker vertices. 
%\item Otherwise $w$ has at most one empty marker vertex and at most one perfect marker vertex.
%\end{enumerate}
%\end{corollary}

The following lemma summarizes what we have found so far and completes the picture.  Recall that $P^*(u)  =   \{ q \in V(u) | q$ is perfect and not the centre of a star$\}$ and that $
E^*(u)  =  \{ q \in V(u) | q$ is empty, or $q$ is perfect and the centre of a star$\}$.

\begin{lemma} \label{lem:x_to_prime}
Consider the prime-splitting construction.  Assume that $x$ is adjacent to a prime node $u$ in $ST(G+x)$ such that $G(u_x)$ is not prime. Then:
\begin{enumerate}
\item $ST(G)=(T,\mathcal{F})$;
\item every tree-edge of $ST(G)$ both extremities of which are not inherited from $(T_x,\mathcal{F}_x)$ is $MM$ and all other tree-edges are $PM$ or $EM$;
\item a degenerate node $v$ of $(T,\mathcal{F})$ incident to a non-inherited tree-edge results from at most two node-joins during step 3 of the construction and these node-joins respectively generate the split $(P^*(v),V(v) \setminus P^*(v))$ and/or $(E^*(v),V(v) \setminus E^*(v))$ of $G(v)$.
%\item every node-join proceeded to obtain $(T,\mathcal{F})$ from $(T'_x,\mathcal{F}'_x)$ corresponds to the split $(P^*(v),V(v) \setminus P^*(v))$ or $(E^*(v),V(v) \setminus E^*(v))$ of $G(v)$ for some node $v$ incident to an MM tree-edge.
\end{enumerate}
\end{lemma}

\begin{proof}
The first assertion is given by Observation~\ref{obs:st(g)} and the second follows from Lemma~\ref{lem:x_to_prime_states} and Remark~\ref{rk:inherited-states}. So it remains to prove the third property.

Let $v$ be a degenerate node of  $(T,\mathcal{F})$ incident to a non-inherited tree-edge (i.e. an $MM$ tree-edge). Observe that if $v$ results from the node-join of nodes $w$ and $w'$ during the third step, then the tree-edge $e=ww'$ is inherited (i.e. $EM$ or $PM$) and exactly one of the two nodes, say $w$ is incident to an $MM$ tree-edge. Let $p$ the extremity of $e$ in $w$ and $q$ be the (mixed) extremity of $e$ in $w'$.  We need to examine all the possibilities for $p$ and $w$. We provide all details for the first case; the others use similar arguments. 

If $w$ is a star node and $p$ its perfect centre, then $q$ is a degree-1 marker vertex of $w'$. The resulting star node $v$ contains the split $(E^*(v),V(v) \setminus E^*(v))$ where $E^*(v)$ is the set of inherited marker vertices of $w'$ ($E^*(v)$ contains the perfect  centre  and the empty degree-1 marker vertices of $w'$). This follows from Lemma~\ref{lem:x_to_prime_cleaning}-1 which shows that $q$ is the only empty marker vertex and from Remark~\ref{rk:inherited-states} that claims that the states of inherited marker vertices are unchanged under node-join. 

The other cases follow from Lemma~\ref{lem:x_to_prime_cleaning}-2. If $p$ is an empty marker  vertex of $w$ (in that case, the node $w$ can be a star or a clique as well), then node $v$ contains the split $(E^*(v),V(v) \setminus E^*(v))$ with $E^*(v)$ is the set of inherited marker vertices of $w'$.  Now if $p$ is a perfect marker vertex but not the centre of a star, then the resulting node contains the split $(P^*(v),V(v) \setminus P^*(v))$ where $P^*(v)$ is composed of the marker vertices inherited from $w'$.

Finally, by Lemma \ref{lem:x_to_prime_cleaning} a degenerate node incident to an $MM$ tree-edge contains at most two non-mixed marker vertices, and thus at most two node-joins are required to generate node $v$. We mention that forthcoming Figure \ref{fig:case4cleaning} illustrates the two inverse node-split operations on $v$.
\end{proof}

\subsection{Construction of $ST(G+x)$ from $ST(G)$}

%Before the description of the modifications of $ST(G)$ required by $x$'s insertion, we state our incremental characterization of the split-tree.To provide a complete description of a split-tree marked with respect to a subset $S$ of its leaves,  we introduce two new terms.

%As we will see, a degenerate hybrid node (if it exists) will need some updates in order to insert a new vertex neighbouring $S$ while the rest of the split-tree will remain unchanged. Likewise, if a fully-mixed tree-edge exists between two nodes $u$ and $v$, the insertion of $x$ will impose the node-join of $u$ and $v$. 

%%%%%VIEUX
%\begin{theorem}[Incremental characterization] \label{lem:cases}
%Let $ST(G) = (T,\mathcal{F})$ marked with respect to a subset $S$ of leaves. Then exactly one of the following condition holds:
%\begin{enumerate}
%\item $ST(G)$ contains a tree-edge $e$ both of which extremities are perfect;
%\item $ST(G)$ contains a tree-edge $e$ one of which extremity is perfect and the other empty;
%\item $ST(G)$ contains a unique hybrid node $u$ which is degenerate;
%\item $ST(G)$ contains a unique hybrid node $u$ which is prime;
%\item $ST(G)$ contains a fully-mixed tree-edge $e$.
%\end{enumerate}
%\end{theorem}
Having shown how $ST(G)$ can be derived from $ST(G+x)$, we now use these results to characterize how $ST(G+x)$ can be derived from $ST(G)$.  The various cases of the following theorem drive our Split Decomposition algorithm in Section \ref{sec:algo}.  Recall that by definition, a 
fully-mixed subtree is maximal.

\begin{theorem}
\label{th:cases}
Let $ST(G) = (T,\mathcal{F})$ be marked with respect to a subset $S$ of leaves. Then exactly one of the following conditions holds:

\begin{enumerate}
\item $ST(G)$ contains a clique node, whose marker vertices are all perfect, and this node is unique; %$u$ 
%the incident tree-edges to which are the set of tree-edges with two perfect extremities;
\item $ST(G)$ contains a star node, %$u$ 
whose marker vertices are all empty except the centre, which is perfect, and this node is unique;
%the incident tree-edges to which are the set of tree-edges with one perfect extremity and other empty;
\item $ST(G)$ contains a unique hybrid node, %$u$ 
and this node is prime;
\item $ST(G)$ contains a unique hybrid node, %$u$ 
and this node is degenerate;
\item $ST(G)$ contains a $PP$ tree-edge, %$e$ 
and this edge is unique;
\item $ST(G)$ contains a $PE$ tree-edge, %$e$ 
and this edge is unique;
%$ST(G)$ contains a unique tree-edge %$e$ 
%one extremity of which is perfect, the other empty;

%EME variant below more consisten I think
%\item $ST(G)$ contains a fully-mixed tree-edge $e$ and the set of fully-mixed edges is connected.
\item $ST(G)$ contains a unique fully-mixed subtree.
\end{enumerate}

Moreover, in every case, the unique node/edge/subtree is obtained from $T$ by deleting, for every tree-edge $e$ with a perfect or empty extremity $q$ whose opposite $r$ is mixed, the tree-edge $e$ and the node or leaf corresponding to $r$. In case 1 and case 2, the node, together with its adjacent edges, is obtained in this way.
\end{theorem}

%From the constructive point of view, building the unique object to consider (edge, node or fully-mixed subtree) can be done simply by deleting recursively the whole subtree $T(q)$ for every perfect or empty marker vertex $q$ whose opposite is mixed. Indeed, by the hereditary property and the unicity results in the theorem, or as well by the case by case analysis of the previous subsection, the tree-edges of such a subtree all have a mixed extremity and the other perfect or empty. This constructive method to detect the case to consider will be substantially the one used in next algorithmic section, with some adaptation for an efficient implementaion purpose.

\begin{proof}
By Lemmas \ref{lem:x_to_clique}, \ref{lem:x_to_star}, \ref{lem:x_to_prime_hybrid} or \ref{lem:x_to_prime}, applied to $G+x$ with $N(x)=S$, we directly know that (at least) one condition holds. A more careful look at these lemmas also proves that exactly one condition holds, implying directly with Corollary \ref{cor:hereditary}-2 the given construction by deletion of $PM$ and $EM$ edges. First, notice that the following conditions are mutually exclusive: there exists a $PP$ edge; there exists a $PE$ edge; there exists an $MM$ edge. Indeed, every time one of these conditions holds, all the tree-edges of another type are known to be $PM$ or $EM$ (Lemmas \ref{lem:x_to_clique}, \ref{lem:x_to_star}, and \ref{lem:x_to_prime}). These three cases are mutually exclusive from the existence of a hybrid node (Lemmas \ref{lem:x_to_clique}, \ref{lem:x_to_star}, and \ref{lem:x_to_prime_hybrid}). Together, these four cases -- existence of a PP edge, PE edge, MM edge, and hybrid node -- determine the cases in the theorem: if there is exactly one (respectively at least two) $PP$ edge(s), then case 5 (respectively case 1) holds;
if there is exactly one (respectively at least two) $PE$ edge(s), then case 6 (respectively case 2) holds; if there is a hybrid node, then it is either prime (case 3) or degenerate (case 4) but not both; if there is an $MM$ edge, then case 7 holds.
\end{proof}

\begin{proposition} [Cases 1, 2 and 3 of Theorem~\ref{th:cases}] \label{prop:unique-node-cases123}
Let $ST(G) = (T,\mathcal{F})$ be marked with respect to a subset $N(x)$ of leaves. If $ST(G)$ contains:
\begin{itemize}
\item {[case 1]} a unique clique node $u$ the marker vertices of which are all perfect, or
\item {[case 2]} a unique star node $u$ the marker vertices of which are all empty except its centre which is perfect, or
\item {[case 3]} a unique hybrid node $u$ which is prime,
\end{itemize}
then $ST(G+x)$ is obtained by  adding to node $u$ a marker vertex $q$ adjacent in $G(u)$ to $P(u)$ and making the leaf $x$ the opposite of $q$.
% a tree-edge $e$ both of which extremities are perfect, then either:
%\begin{itemize}
%\item $e$ is unique (case 1): $ST(G+x)$ is obtained by subdividing $e$ with a new clique node adjacent to $x$ and $e$'s extremities;
%% and none of its extremities is a marker of a clique node,
%\item $e$ is not unique (case 2):  $ST(G+x)$ is obtained by adding to the clique node incident a universal marker vertex $q$ and making $x$ $q$'s opposite.
%%or the set of such tree-edges is the set of tree-edges incident to a unique clique node $u$ and $ST(G+x)$ is obtained by adding to node $u$ a universal marker vertex $q$ and making $x$ $q$'s opposite.
%\end{itemize}
\end{proposition}

\begin{proof}
Trivial by the definition of the split-tree; the resulting GLT is reduced and its accessibility graph is $G+x$.
Notice that each of these three cases is the converse construction of the one provided in Lemma \ref{lem:x_to_clique}-2, or Lemma \ref{lem:x_to_star}-2, or Lemma \ref{lem:x_to_prime_hybrid}, respectively.
%Obviously the resulting GLT is reduced since the set of nodes and the type of $u$ are unchanged. By construction, we have $A(x)=S$ (in case 1, $P(u)=V(u)$; in case 2, the only perfect marker vertex of $u$ is its centre). It follows that the resulting GLT is $ST(G+x)$.
\end{proof}

%If $ST(G)$ contains a hybrid degenerate node $u$, prior to $x$'s insertion, node $u$ has to be split. This case corresponds in a sense to the first non-trivial update 

%EME it is a pity that the statement of proposition repeats the sateent of the next one... it would be more natural to put permutate the two propositions and apply the other to this one, just as done in figures. But this would mean dealing with case 5 and 6 before case 4, and the next proposition would not be self contained... I let it like this since christophe tend to prefer this, but feel free to change.
%
\begin{proposition}[Case 4 of Theorem~\ref{th:cases}]
\label{prop:hybrid-degen-case4}
Let $ST(G) = (T,\mathcal{F})$ be marked with respect to a subset $N(x)$ of leaves. If $ST(G)$ contains a unique hybrid node $u$ which is degenerate, then $ST(G+x)$ is obtained in two steps:
\begin{enumerate}
\item performing the node-split corresponding to 
%$(P^*(u),V(u)-P^*(u))$, or equally $(E^*(u),V(u)-E^*(u))$, 
$(P^*(u),E^*(u))$
thus creating a tree-edge $e$ both of whose extremities are perfect or empty (see Figure~\ref{fig:case3Example});
\item subdividing $e$ with a new ternary node adjacent to $x$ and $e$'s extremities, such that the node is a clique if both extremities of $e$ are perfect, and 
such that the node is a star whose centre is the opposite of $e$'s empty extremity otherwise (see Figure~\ref{fig:cases12Example}).

\end{enumerate}
\end{proposition}

\begin{proof}
First, observe that $(P^*(u),E^*(u))$ is a split since $|P^*(u)|>1$ and $|E^*(u))|>1$, otherwise the degenerate node $u$ would either be a clique adjacent to a $PP$ edge or a star adjacent to a $PE$ edge, contradicting $u$ being hybrid.
Then the construction follows easily from the definition of the split-tree: the resulting GLT is reduced and its accessibility graph is $G+x$.
Notice that this construction is the converse of the one provided in Lemma~\ref{lem:x_to_clique}-1(b) if the label is a clique, or Lemma \ref{lem:x_to_star}-1(b) if the label is a star. 
%The identification of the split as $(P^*(u),V(u)-P^*(u))=(V(u)-E^*(u), E^*(u))$ is straightforward in every case.
\end{proof}

\begin{figure}[htbh]
\begin{center}
\includegraphics[scale=1]{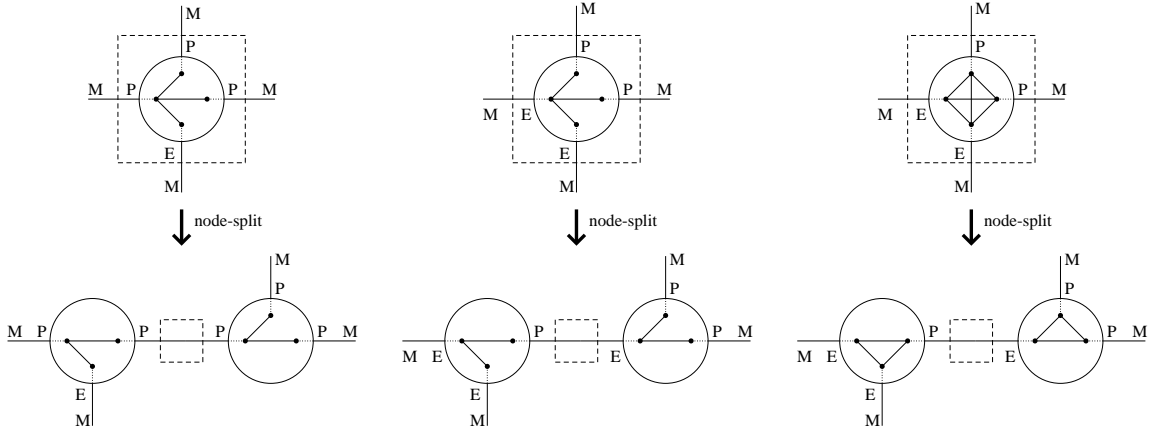}
\end{center}
%\vspace{-0.2in}
\caption{Node-split performed when there is a degenerate hybrid node  (case 4 of Theorem~\ref{th:cases}, first step of Proposition \ref{prop:hybrid-degen-case4}). 
%The following transformation is shown on Figure \ref{fig:cases12Example}. 
The dashed rectangle shows where the local transformation takes place, as described by the second step of Proposition \ref{prop:hybrid-degen-case4} (Figure \ref{fig:cases12Example}).}
%Theorem~\ref{th:cases}.}
\label{fig:case3Example}
\end{figure}

%DGC4 Figure~\ref{fig:case3Example} illustrates the node-split operation performed in Proposition \ref{prop:hybrid-degen-case4}.
%Figure~\ref{fig:cases12Example} illustrates the update performed in the second step of Proposition \ref{prop:hybrid-degen-case4}, as well as the update performed in Proposition \ref{prop:unique-edge-case56} below.

\begin{figure}[htbh]
\begin{center}
\includegraphics[scale=1]{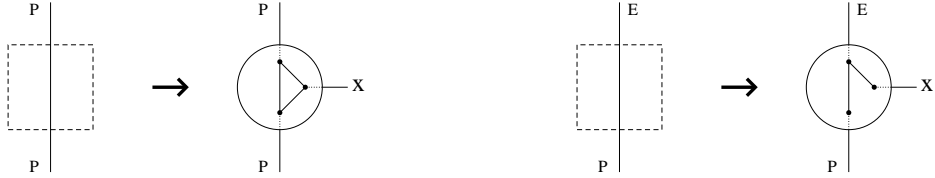}
\end{center}
%\vspace{-0.2in}
\caption{Insertion of vertex $x$ in case 4 of Theorem~\ref{th:cases}, second step of Proposition \ref{prop:hybrid-degen-case4} and when there is a unique edge with no mixed extremity (cases 5 or 6 of Theorem~\ref{th:cases}, Proposition \ref{prop:unique-edge-case56}). The dashed rectangle shows where the local transformation is made.}
\label{fig:cases12Example}
\end{figure}

\begin{proposition} [Cases 5 and 6 of Theorem~\ref{th:cases}] \label{prop:unique-edge-case56}
Let $ST(G) = (T,\mathcal{F})$ be marked with respect to a subset $N(x)$ of leaves. If $ST(G)$ contains:
\begin{itemize}
\item {[case 5]} a unique tree-edge $e$ both of whose extremities are perfect, then $ST(G+x)$ is obtained by subdividing $e$ with a new clique node adjacent to $x$ and $e$'s extremities (see Figure \ref{fig:cases12Example});
\item {[case 6]} a unique tree-edge $e$ one of whose extremities is perfect and the other empty, then $ST(G+x)$ is obtained by subdividing $e$ with a new star node adjacent to $x$ and $e$'s extremities, such that the centre of the star is opposite $e$'s empty extremity (see Figure \ref{fig:cases12Example}).
\end{itemize}
%by  adding to node $u$ a marker vertex $q$ adjacent in $G(u)$ to $P(u)$ and making $x$ the opposite of $q$.
% a tree-edge $e$ both of which extremities are perfect, then either:
%\begin{itemize}
%\item $e$ is unique (case 1): $ST(G+x)$ is obtained by subdividing $e$ with a new clique node adjacent to $x$ and $e$'s extremities;
%% and none of its extremities is a marker of a clique node,
%\item $e$ is not unique (case 2):  $ST(G+x)$ is obtained by adding to the clique node incident a universal marker vertex $q$ and making $x$ $q$'s opposite.
%%or the set of such tree-edges is the set of tree-edges incident to a unique clique node $u$ and $ST(G+x)$ is obtained by adding to node $u$ a universal marker vertex $q$ and making $x$ $q$'s opposite.
%\end{itemize}
\end{proposition}

\begin{proof}
Direct by the definition of the split-tree; the resulting GLT is reduced and its accessibility graph is $G+x$.
Notice that each of these two cases is the converse construction of the one provided in Lemma \ref{lem:x_to_clique}-1(a), or Lemma \ref{lem:x_to_star}-1(a), respectively.
\end{proof}

\begin{definition} \label{def:cleaning}
Let $(T,\mathcal{F})$ be a GLT marked with respect to a subset of leaves and having a fully-mixed subtree.
\emph{Cleaning} the GLT consists of performing, for every degenerate node $u$ of the fully-mixed subtree, the node-splits defined by
$(P^*(u),V(u) \setminus P^*(u))$ and/or $(E^*(u),V(u) \setminus E^*(u))$ as long as they are splits of $G(u)$.
% is a two-step process: first, every degenerate node $u$ containing the split $(P^*(u),V(u) - P^*(u))$ is node-split accordingly; then in the resulting GLT, every degenerate node $u$ containing the split $(E^*(u),V(u) - E^*(u))$ is node-split accordingly.  
The resulting GLT is denoted $c\ell(T,\mathcal{F})$.
\end{definition}

The above definition makes sense thanks to Remark \ref{rem:node-split-commutative} since $P^*(u)\cap E^*(u)=\emptyset$; the two node-splits corresponding to these splits can be done in any order with the same result.
Figure \ref{fig:case4cleaning} illustrates the possible local transformations at each node $u$. 

\begin{figure} %[htbh]
\begin{center}
\includegraphics[scale=0.9]{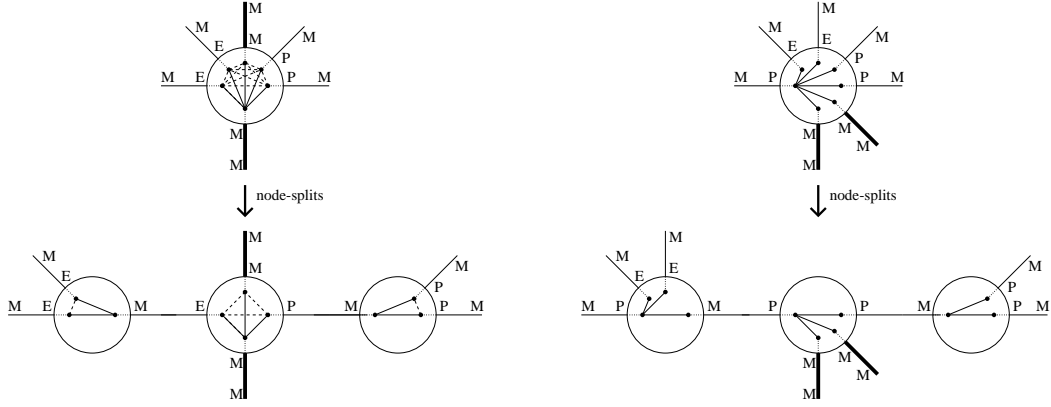}
\end{center}
%\vspace{-0.2in}
\caption{Cleaning of a degenerate node of the fully-mixed subtree (Definition \ref{def:cleaning}%, case 7 of Theorem~\ref{th:cases}). 
).
The left picture (as shown with dashed label-edges) applies equally to a clique or star node with a mixed centre. The right picture concerns a star node with a perfect centre. Bold edges are fully-mixed.}
\label{fig:case4cleaning}
\end{figure}

%EME I added the remark below, it is not used here, but may simplify considerations. Or it can be put in section 5, where the cleaning is considered (it just avoids a useless cardinality test in Algo 6)

\begin{remark} \label{rk:cleaning}
With Lemma \ref{lem:hereditary} and the fact that $u$ contains at least one mixed marker vertex whose opposite is mixed,
%With the above notations, 
one can easily show that
$(P^*(u),V(u) \setminus P^*(u))$, respectively $(E^*(u),V(u) \setminus E^*(u))$, is a split of $u$ if and only if $| P^*(u) |>1$, respectively $| E^*(u) |>1$. 
\end{remark}

\begin{proposition}[Case 7 of Theorem~\ref{th:cases}]
\label{prop:fully-mixed-case-7}
Let $ST(G) = (T,\mathcal{F})$ be marked with respect to a subset $N(x)$ of leaves. If $ST(G)$ contains a fully-mixed tree-edge $e$, then $ST(G+x)$ is obtained by:
\begin{itemize}
\item contracting, by a series of node-joins,  the fully-mixed  subtree of $c\ell(ST(G))$ into a single node~$u$;% (marked by $N(x)$)
\item adding to node $u$ a marker vertex $q_x$ adjacent in $G(u)$ to $P(u)$ and making $x$, $q_x$'s opposite. The resulting node $u$ is prime.  See Figure \ref{fig:case4Example} for an illustration of the whole process.

\end{itemize}
\end{proposition}

\begin{proof}
First, observe that the series of node-joins is well defined by Remark \ref{rem:node-join-commutative}.
This construction is exactly the inverse of the prime-splitting construction referenced in Lemma \ref{lem:x_to_prime}.
More precisely, the GLT $cl(ST(G))$ here is exactly the GLT $(T'_x,\mathcal{F}'_x)$ there.
So the fully-mixed subtree of  $cl(ST(G))$ here is the fully-mixed subtree
induced by $ST(G(u_x))$ there. And the series of node-joins applied to this subtree leads to the node labelled by $u_x$, to which $x$ is added naturally.
 \end{proof}

\begin{figure}[] %[htbh]
\begin{center}
\includegraphics[scale=0.9]{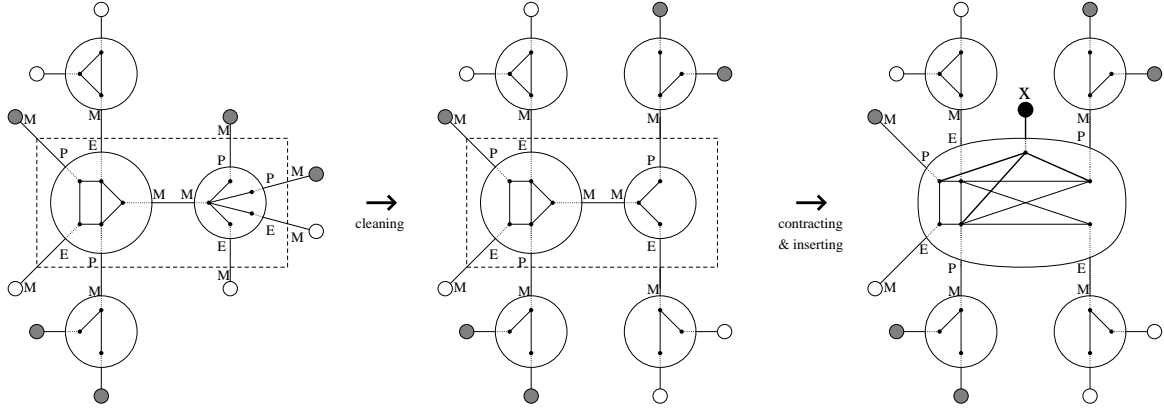}
\end{center}
%\vspace{-0.2in}
\caption{A complete example of the split-tree update where case 7 of Theorem~\ref{th:cases} applies (Proposition \ref{prop:fully-mixed-case-7}). The dashed rectangle contains the fully-mixed subtree of $ST(G)$ in the left picture, and the fully-mixed subtree of $cl(ST(G))$ in the middle one. }
\label{fig:case4Example}
\end{figure}

Throughout the rest of the paper, we will use the phrase \emph{contraction step}, or simply \emph{contraction}, to refer to the procedure involved in Proposition \ref{prop:fully-mixed-case-7} which transforms the fully-mixed subtree of $c\ell(ST(G))$ into a prime node that has the new vertex  $x$ attached.
%, that is including the insertion of the new vertex. 

To end this section, we point out a number of observations that follow from the results in this section.  First, the construction provided by  Propositions~\ref{prop:unique-node-cases123}, \ref{prop:hybrid-degen-case4}, and~\ref{prop:unique-edge-case56} applied to a distance hereditary graph (i.e. when every node is degenerate), amounts to the one provided in \cite{GP07,GP08}. Thus, the present construction is a generalization to arbitrary graphs.
 
Secondly, note that we chose to separate the cases in Theorem \ref{th:cases} for consistency with our next algorithm. But other shorter and equivalent presentations would have been possible; for instance: case 1 and case 5 (respectively case 2 and case 6) could be grouped and treated the same way as they are the only cases where there exists a $PP$ (respectively $PE$) edge, with a clique-join or star-join after insertion if the edge was not unique; case 4 comes to cases 5 and 6 by making a $PP$ or $PE$ edge  appear after splitting the node; case 1 could be considered as a trivial subcase of case 7.

Finally, the results of this subsection, together with Lemmas \ref{lem:x_to_prime_states} and \ref{lem:x_to_prime_cleaning}
yield the following theorem which plays an important role in our circle graph recognition algorithm \cite{GPTC11b}.

\begin{theorem}\label{circle}
A graph $G + x$ is a prime graph if and only if $ST(G)$, marked with respect to $N(x)$, satisfies the following:
\begin{enumerate}
\item  Every marker vertex not opposite a leaf is mixed,
\item Let $w$ be a degenerate node.  If $w$ is a star node, the centre of which is perfect, then $w$ has no empty marker vertex
and at most two perfect marker vertices;  otherwise, $w$ has at most one empty marker vertex and at most one perfect marker
vertex.
\end{enumerate}
\end{theorem}
\section{An LBFS incremental split decomposition algorithm} \label{sec:algorithm}
\label{sec:algo}

Our combinatorial characterization, described by Theorem~\ref{th:cases} and the subsequent Propositions~\ref{prop:unique-node-cases123}, \ref{prop:hybrid-degen-case4},~\ref{prop:unique-edge-case56} and~\ref{prop:fully-mixed-case-7}, immediately suggests an incremental split decomposition algorithm.
%EME not understood
%, the more so because of the algorithmic nature of its components: node-splits and cleaning as well as node-joins and contraction. 
Our characterization makes no assumption about the order vertices are to be inserted.  For the sake of complexity issues, we choose to add vertices according to an LBFS ordering $\sigma$, which we assume to be built by a preprocess. 
%
%EME vient de section 3 a reprendre
%As mentioned previously,the implementation of our algorithm, developed in section~\ref{sec:algo}, assumes that the given graph is connected and that the vertices are added following an LBFS ordering.  
%
%Let us first point out that the incremental construction of Section \ref{sec:characterization} assumes that the graphs $G$ and $G+x$ are both connected, which is compatible with the insertion of new vertices according to a LBFS ordering: from Remark~\ref{prefix} all iterations of the algorithm satisfy the condition that $G$ and $G + x$ are connected.
From Remark~\ref{prefix}, such an ordering is compatible with the assumption made in Section \ref{sec:characterization}: all iterations of the algorithm
satisfy the condition that $G$ and $G + x$ are connected.
% that the graphs $G$ and $G+x$ are both connected.
%
Roughly speaking, the LBFS ordering will play two crucial parts: first, it permits a costless twin test allowing us to avoid ``touching'' non-neighbours of the new vertex when identifying states (Subsection \ref{subsec:cases});
second, it means that successive updates of the split-tree have an efficient amortized cost (Section \ref{sec:runningTime}).
\medskip

As in the previous section, we assume throughout that $ST(G) = (T,\mathcal{F})$, and that leaves and marker vertices in $ST(G)$ are assigned states according to the set $N(x)$; then we consider the changes required to form $ST(G+x)$.   
Algorithm~\ref{alg:vertexinsertion} outlines how the split-tree is updated to insert the last vertex of an LBFS ordering. 
%The rest of the section details its implementation.

\smallskip

\begin{algorithm}[ht]
\KwIn{A graph $G$, a vertex $x\notin V(G)$ which is the last vertex in an LBFS ordering of $G+x$, and the split-tree $ST(G)=(T,\mathcal{F})$}
\KwOut{The split-tree $ST(G+x)$}
\BlankLine

\lnl{line:cases} Determine whether $ST(G)$ contains either a tree-edge neither of whose extremities is mixed, or a hybrid node, 
%or a fully-mixed tree-edge\;
or a fully-mixed subtree\;
%according to Theorem \ref{th:cases}\;

%modif eme abandonnee
% \uIf{$ST(G)$ contains a tree-edge $e$ none of which extremities is mixed}{
%\lnl{line:PP-PE}	\lIf{$e$ is not adjacent to a degenerate node}{
%		update $ST(G)$ according to Proposition~\ref{prop:unique-edge-case56}\;
%		}
%	\lElse{update $ST(G)$ according to Proposition~\ref{prop:unique-node-cases123}\;}
%	}
 \uIf{$ST(G)$ contains a tree-edge $e$ neither of whose extremities is mixed}{
\lnl{line:PP-PE}	\lIf{$e$ is unique}{update $ST(G)$ according to Proposition~\ref{prop:unique-edge-case56}\;
		}
	\lElse{update $ST(G)$ according to Proposition~\ref{prop:unique-node-cases123}\;}
	}
	
\uIf{$ST(G)$ contains a  hybrid node $u$}{
\lnl{line:split}\lIf{$u$ is degenerate}{update $ST(G)$ according to Proposition~\ref{prop:hybrid-degen-case4}\;}
	\lElse{update $ST(G)$ according to Proposition~\ref{prop:unique-node-cases123}\;}
	}
	
%\uIf{$ST(G)$ contains a fully-mixed tree-edge}{
\uIf{$ST(G)$ contains a fully-mixed subtree}{
\lnl{line:contraction}	compute and update  $c\ell(ST(G))$ according to Proposition~\ref{prop:fully-mixed-case-7}\;
	}
\caption{Vertex insertion} \label{alg:vertexinsertion}
\end{algorithm}

The first task consists of identifying which of the cases of Theorem~\ref{th:cases} holds, at line \ref{line:cases} of Algorithm~\ref{alg:vertexinsertion}.  The implementation of line~\ref{line:cases} is by a procedure which also returns the states of the involved marker vertices 
%(Algorithm \ref{alg:perfectPruning}), 
(see Subection \ref{subsec:cases}),
and hence allows us to apply the constructions provided by the propositions.
At line~\ref{line:PP-PE}, testing the uniqueness of the tree-edge $e$ amounts to a check as to whether $e$ is incident to a clique or a star node. This is required to discriminate between cases 1 and 5 or cases 2 and 6.
More precisely, if $e$ has two perfect extremities and is adjacent to a clique $u$, then all the marker vertices of $u$ are perfect by Lemma \ref{lem:hereditary}, and then Proposition~\ref{prop:unique-node-cases123} is applied to this clique node. If $e$ has a perfect extremity $q$ and its opposite $r$ is empty, and if $q$ is the centre of a star or $r$ is a degree-1 vertex of a star, then by Lemma \ref{lem:hereditary}, this star has all of its marker vertices empty except the centre, which is perfect,
and then Proposition~\ref{prop:unique-node-cases123} is applied to this star node.
%At line \ref{line:split}, the states of $u$ are known hence Proposition~\ref{prop:hybrid-degen-case4} can be applied.
%
These tests at line~\ref{line:PP-PE} can be done in constant time in the data-structure we use,
as well as the updates required in these simplest cases (Proposition~\ref{prop:unique-node-cases123} and Proposition~\ref{prop:unique-edge-case56}). They will not be considered again in the implementation.%
\medskip

Now, this section fills out the framework by specifying procedures for the state assignment, node-split, node-join, cleaning, and contraction involved at lines \ref{line:cases}, \ref{line:split}, and \ref{line:contraction} of Algorithm \ref{alg:vertexinsertion}. We first describe our data-structures which is partly based on \emph{union-find}~\cite{CLR01}. We then provide a complexity analysis of the insertion algorithm parameterized by elementary union-find requests. An amortized complexity analysis is developed in the next section.

%As in the previous section, we assume throughout that $ST(G) = (T,\mathcal{F})$, and that leaves and marker vertices in $ST(G)$ are assigned states according to the set $N(x)$; then we consider the changes required to form $ST(G+x)$.    

%----------------------------------------------------------------------------------------------------------------------
\subsection{The data-structure}

In order to achieve the announced time complexity, we implement a GLT $(T,\mathcal{F})$ with the well-known union-find data-structure~\cite{CLR01}, making $T$ a rooted tree.
%(the leaf corresponding to the first inserted vertex will serve as the root of $T$). 
There are two reasons for this choice. First, identifying empty and perfect subtrees is easier if the tree $T$ is rooted. Second, in the contraction step, we'll need to union the neighbourhoods of two nodes to perform a node-join. 
%EME useless below
We first present how the tree $T$ will be encoded with a union-find data-structure. We then detail how each node and its labels are represented.
\medskip

A union-find data-structure maintains a collection of disjoint sets.  Each set maintains a distinguished member called its \emph{set-representative}.  Union-find supports three operations:

\begin{enumerate}
\item \texttt{initialize}$(x)$: creates the singleton set $\{x\}$;
\item \texttt{find}$(x)$: returns the set-representative of the set containing $x$;
\item \texttt{union}$(S_1,S_2)$: forms the union of $S_1$ and $S_2$, and returns the new set-representative, chosen from amongst those of $S_1$ and $S_2$.
\end{enumerate}

The initialization step takes $O(1)$ time, and a combination of $k$ union and find operations takes time $O(\alpha(N) \cdot k+N)$, where $N$ is the number of elements in the collection of disjoint sets 
%EME-Ack-update
and $\alpha$ is the inverse Ackermann function~\cite{CLR01}. 
The complexity of the algorithms described in this section will be parameterized by
\texttt{initialize-cost}, \texttt{find-cost}, \texttt{union-cost} the respective costs of the above requests.
\medskip

Our algorithm will store a GLT $(T,\mathcal{F})$ as a \emph{rooted GLT} where 
%We call \emph{rooted GLT} the data-structure used in the algorithm to implement a GLT $(T,\mathcal{F})$.
%
a leaf of $T$ will serve as \emph{the root} (it is the leaf corresponding to the first inserted vertex).
Each node or leaf of $T$, except the root, has a \emph{parent pointer} to its parent (which is a node or the root leaf) in $T$ with respect to the root. 
To each prime node is associated a \emph{children-set} containing 
the set of its 
%DGC5 (combinatorial) 
children  in $T$ with respect to the root (nodes or leaves).
The children-sets of prime nodes form the collection of disjoint sets
maintained by the union-find data-structure.
Every children-set has a set representative, which is a 
%DGC5(combinatorial) 
child of the node in $T$.
%
%Depending on the node, its
A parent pointer may be active or not. The nodes or leaves with an \emph{active} parent pointer are: the child of the root, the children of degenerate nodes (clique or star), and the nodes or leaves that are the set representative of the children-set of a prime node.
A non-active parent pointer is just one that will never be used again; there is no need to update information for it.

\begin{remark}\label{rem:traversal}
A traversal of a rooted GLT $(T,\mathcal{F})$  can be implemented in time $O((1+\emph{\texttt{find-cost}}) \cdot |T|)$.
%A traversal of a GLT $(T,\mathcal{F})$ represented by a data-structure based on union-find as described can be implemented in time $O((1+\emph{\texttt{find-cost}}).|T|)$.
\end{remark}

Union-find is required only to update the tree structure (child-parent relationship) efficiently.  As the node-splits only apply to degenerate nodes (lines~\ref{line:split}, \ref{line:contraction}), union-find is not required here. Union operations are performed after the cleaning step, during the contraction step (line~\ref{line:contraction}). 
%On the other hand, when a node-join  is performed, we need to union the children sets of the two nodes involved.

\begin{remark}\label{rem:data}
The data-structure described here concerns a split-tree, whose labels are either prime or degenerate, since after each step of the construction it is such a GLT that will be obtained. 
%It is extended directly to general GLTs by dealing with non-degenerate nodes the same way than with prime nodes.
%
%Still, node-join can of course be allowed in the data-structure.
%
Still, we need to allow node-joins in the data-structure.
In what follows, during the successive node-joins in the contraction step (Proposition \ref{prop:fully-mixed-case-7}), the GLT has one non-degenerate node whose 
label graph will eventually become prime only after the final insertion step.
%becomes prime only after the final insertion step.
% 
The data-structure for such a GLT remains the same by recording the type of this non-degenerate node as prime, and dealing with it the same way as a prime node. 
%Briefly: one can consider that the prime type in the data structure stands for non-degenerate.
\end{remark}

%EME below modified
%Our application of union-find maintains for each prime node a set containing its children. Each node (or leaf) of $T$, except the root, has 	a \emph{parent pointer} to another node or leaf. Depending on the node, the parent pointer may be active or not. A non-active parent pointer is just one that will never be used again, there is no need to update some information for that. The nodes with an active parent pointer are the children of degenerate nodes (clique or star) and the nodes that are the set representative of the children of a prime node (see Figure~\ref{fig:union-find}).

It is important to note that removing elements from sets is not supported in the union-find data-structure. It follows that when a node-join is performed and the children-set of a node $u'$ is unioned with the children-set of its parent $u$, the node object corresponding to $u'$ still exists in the children-set of $u$. As we will see later, the persistence of these \emph{fake} nodes is not a problem. Indeed, 
their total number will be suitably bounded, 
 they will never be selected again as set representatives, and the data-structure we develop below guarantees that they will never be accessed again. 
In particular, no active parent pointer points to a fake node. 
That is why the children-set of a prime node may strictly contain its set of 
%DGC5(combinatorial) 
children in $T$. 
%DGC5In what follows, we simply call child a (combinatorial) child in $T$.
%In what follows, we will now omit the precision ``(combinatorial)''. 
%So in our application, a set representing the children of a prime node $u$ may contain some fake nodes, which are no longer children of $u$.

%Figure~\ref{fig:union-find} (left part) illustrates this rooted GLT data-structure.  
%Note the introduction there of the term \emph{root marker vertex}.
\medskip

\begin{figure}[htbh]
\begin{center}
\includegraphics[scale=0.9]{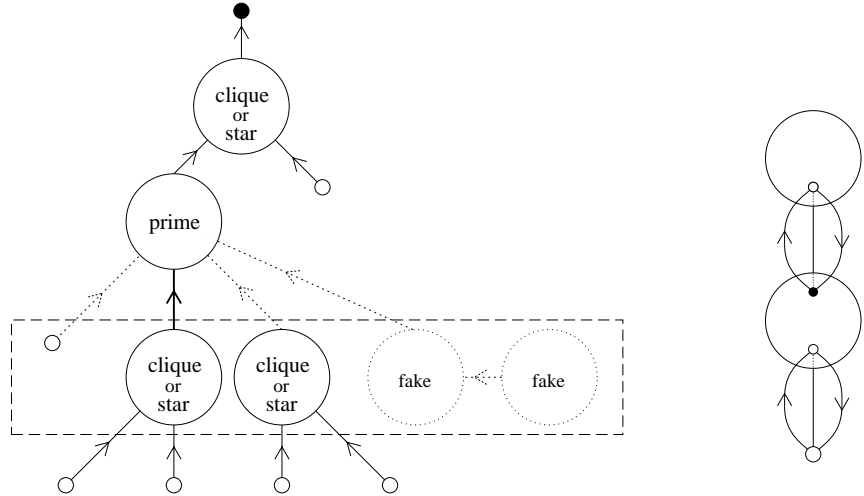}
\end{center}
%\vspace{-0.2in}
\caption{Representation of a rooted GLT with the union-find data-structure. In the left picture: full arrows going up represent active parent pointers, the dashed arrows are non-active; the dashed rectangle represents the children-set of the prime node; small circles outside nodes represent leaves, the black one is the root.
In the right picture: small circles inside nodes represent marker vertices, the black one is a root marker vertex; curved arrows represent pointers to the opposite marker vertex or leaf.}
%\caption{Representation of $T$, the tree of GLT $(T,\mathcal{F})$, with the union-find data-structure}
\label{fig:union-find}
\end{figure}

To complete our data-structure we define a data-object for every node and leaf in the rooted GLT. These data-objects will maintain several fields:
\begin{itemize}
\item We already mentioned that every node and leaf, except the root, has a parent pointer. % toward another node or leaf object. 
Each node $u$ has a distinguished marker vertex, called the \emph{root marker vertex}, which is the extremity of the tree-edge between $u$ and its parent. Every node maintains a pointer to its root marker vertex. As already implied, we need to store the \emph{type} of every node 
(prime, clique, or star),
%(prime, clique, or star, where prime may stand for non-degenerate), 
and we also store the \emph{number of its children}.
\end{itemize}

\noindent 
On the top of that, depending on its type, every node $u$ maintains the following fields:
\begin{itemize}
\item if $u$ is prime: an \emph{adjacency-list representation} of $G(u)$; a pointer to the \emph{last marker vertex} in $\sigma[G(u)]$; a pointer to its \emph{universal marker vertex} (if it exists);
\item if $u$ is degenerate: a \emph{list of its marker vertices} $V(u)$; and a pointer to its \emph{centre} if it is a star;
\end{itemize}

\noindent
To each leaf and marker vertex, we associate: 
\begin{itemize}
\item a pointer, called the \emph{opposite pointer}, to its opposite marker vertex; 
%EME below confusing since only perfect states are calculated
%a field for its \emph{state} (perfect, empty, mixed), 
a field for its \emph{perfect-state} (at each new vertex insertion, perfect states, but no other state, 
will be computed and recorded,
%will be calculated and recorded by further procedures, 
and the content of these fields from previous vertex insertions is not reused); and, for every root marker vertex, a pointer, called the \emph{node pointer}, toward the node to which it belongs.
\end{itemize}

Figure~\ref{fig:union-find} illustrates this rooted GLT data-structure.  
%
%As an illustration, 
For instance,
notice how a prime node accesses its children-set using this data-structure: pick a non-root marker vertex of the node, then its opposite marker vertex, then the node to which this marker vertex belongs, then the \texttt{find} on this node gives the set-representative of the corresponding children-set.
%----------------------------------------------------------------------------------------------------------------------
\subsection{State assignment and case identification}
\label{subsec:cases}

Prior to any update, a preprocessing of the split-tree is required to identify which of the cases of Theorem~\ref{th:cases} holds. This preprocessing is based on state assignment and tree traversals. The LBFS ordering will play an important role here. The procedure we use for the empty subtrees detection (Algorithm~\ref{alg:emptyPruning}) differs from that for perfect subtrees (Algorithm~\ref{alg:perfectPruning}).

%----------------------------------------------------------------------------------------------------------------------
\subsubsection{Empty subtrees} \label{sec:emptyPruning}

If a marker vertex or leaf $q$ is empty, then, by definition, $T(q)$ contains no leaf in $N(x)$, and thus it will be unchanged under $x$'s insertion. 
For the sake of complexity issues, we will want to avoid ``touching'' $q$ and any part of $T(q)$. %The intuition is that there is nothing to ``charge'' this work.  
We are fortunate that identifying such empty subtrees can be simulated indirectly:

\begin{lemma} \label{emptyPruning}
Consider $ST(G) = (T,\mathcal{F})$, and let $T(N(x))$ be the smallest connected subtree of $T$ spanning the leaves $N(x)$.  Then $q$ is an empty leaf or an empty marker vertex if and only if $T(q)$ and $T(N(x))$ are node disjoint.
\end{lemma}

\begin{proof}
It's important to recall that if $q\in V(u)$, then the node $u$ is not part of $T(q)$.  If $q$ is empty, then $L(q) \cap N(x) = \emptyset$, meaning $T(q)$ and $T(N(x))$ must be disjoint.  If $T(q)$ and $T(N(x))$ are disjoint, then $L(q) \cap N(x) = \emptyset$, meaning $q$ is empty.
\end{proof}

%The previous result says that pruning empty leaves and marker vertices can be simulated by the computation of $T(N(x))$.  The only difference is that pruning replaces $T(q)$ with a leaf from $A(q)$, while $T(N(x))$ merely removes $T(q)$ and doesn't replace it with anything.  We'll see that this difference won't matter.  Its ultimate benefit is efficiency.  We can compute $T(N(x))$ using the procedure specified in~\cite{}, which is repeated here as Algorithm~\ref{alg:emptyPruning}.    

%--------------------
\begin{algorithm} 
\KwIn{A tree $T$ rooted at a leaf and a subset $N(x)$ of its leaves  (assuming $|N(x)|>1$).}
\KwOut{The tree $T(N(x))$, the smallest connected subtree of $T$ spanning the leaves $N(x)$.}
%, rooted at a leaf or a node.}

\BlankLine
Mark each leaf of $N(x)$ as \emph{active} (other nodes and leaves are considered inactive)\;

\BlankLine

\While{[the root is not visited and there are at least two active leaves or nodes] OR [the root is visited and there is at least one active leaf or node]}{
         Let $L$ be the current set of active leaves or nodes\;
	\ForEach{element of $L$, $u$}{
		$u$ is no longer \emph{active}, it becomes \emph{visited}\;
		%mark the parent $v$ of $u$ as \emph{visited} \;
		\lIf{$u$ is not the root and its parent is not visited}{$u$'s parent is marked \emph{active}\;}
		%set the parents just marked as \emph{active} if they are not the root and their parent is not marked\;
		}
	}

\BlankLine

Let $T'$ be the subtree of $T$ induced by the visited leaves and nodes\;
\uIf{$t'$, the root of $T'$, has a unique visited child but $t'$ does not belong to $N(x)$}{
	remove in $T'$ the path from $t'$ to the closest node with at least two visited children\;
	}
\KwRet{$T'$}\;

\caption{\cite{GP07, GP08}~Detection of empty subtrees: computing the smallest connected subtree spanning a set of leaves} \label{alg:emptyPruning}
\end{algorithm}
%--------------------

%Recall that state assignment and pruning is the responsibility of Algorithm~\ref{alg:perfectPruning}, which includes a call to Algorithm~\ref{alg:emptyPruning}. 

We can compute $T(N(x))$ using the procedure specified in~\cite{GP07, GP08}, which is repeated here as Algorithm~\ref{alg:emptyPruning}.    
It was proved in~\cite{GP07, GP08} that  a call to Algorithm~\ref{alg:emptyPruning} runs in time $O(|T(N(x))|)$, assuming each node maintains a pointer to its parent.  Therefore, given the data-structure proposed above, a \texttt{find()} request is needed to move from a node to its parent, when prime. So the following holds:

\begin{lemma} \label{emptyPruningTime}
Given a GLT $(T,\mathcal{F})$, Algorithm~\ref{alg:emptyPruning} returns a subtree of $T$ that is node disjoint from every empty subtree, and runs in time $O((1+$\emph{\texttt{find-cost}}$) \cdot |T(N(x))|)$.
%Algorithm~\ref{alg:emptyPruning}prunes all the subtree $T(q)$ such that $q$ is an empty marker and runs in time $O((1+$\emph{\texttt{find-cost}}$).|T(N(x))|)$.
\end{lemma}

%----------------------------------------------------------------------------------------------------------------------
\subsubsection{Perfect subtrees}

As $T$ is rooted, the subtree $T'=T(N(x))$ has a root which is a leaf or a node of $T$. For every node $u$ of $T'$, the marker vertices of $u$ which are the extremity of a tree-edge in $T'$ form the set $NE(u)$ (recall Definition~\ref{def:states}).
The next task is to identify the perfect subtrees and derive the case identification used at line \ref{line:cases} of Algorithm 
\ref{alg:vertexinsertion}. 
Our procedure, described in detail as Algorithm~\ref{alg:perfectPruning}, outputs either a tree-edge %as in the first two cases of Theorem~\ref{th:cases}, 
or a hybrid node, or the fully-mixed subtree of $ST(G)$. (Recall that a fully-mixed subtree is maximal, by definition.)  It works in three main steps:

%
%Given $T'=T(N(x))$, the next task in the case identification is to identify the perfect subtrees. To select a leaf as the root of $T$ will help us too. Our procedure, described in details as Algorithm~\ref{alg:perfectPruning}, outputs either a tree-edge as in the first two cases of Theorem~\ref{th:cases}, or an hybrid node, or the fully-mixed subtree of $ST(G)$. It works in $3$ main steps:
\begin{enumerate}
\item  First it traverses the subtree $T(N(x))$ in a bottom-up manner to identify the pendant perfect subtrees: for each non-leaf node $u$, we test if the marker $q$ opposite $u$'s root marker is perfect and if so remove $u$ from $T'$ and move to $u$'s parent.

%so prune $T(q)$ and move to $u$'s parent. 
\item Then, if the root of the remaining subtree $T'$ has a unique child $v$ in $T'$, we check whether $v$'s root marker vertex $r$ is perfect and if so remove the root from $T'$ and move to $v$. This test is repeated until the current root of $T'$ neighbours at least two nodes in $T'$ or is the unique remaining node of $T'$. 

\item In the former case, we are done and the resulting tree $T'$ is fully-mixed. 
%Note that at each step we know which marker vertex has to be checked. 
In the latter case, we still need to test whether the remaining node is hybrid or contains a marker vertex opposite a perfect leaf, in which case the output is this edge. As we will see, using the LBFS ordering allows us to test only two marker vertices. 
\end{enumerate}

%----------
% SEE LATER WHERE TO PUT THIS
%----------
%The pruning of perfect marker vertices is unaffected by the charging cost problem observed for empty leaves and marker vertices.  If $q$ is a perfect marker vertex, then $T(q)$ and $N(x)$ have a non-empty intersection, and the pruning of $T(q)$ can be ``charged'' to $N(x)$ in this case.  We don't ``double-charge'' because for any other perfect marker vertex $q'$, we can be sure that $T(q)$ and $T(q')$ differ by at least one node.  We use the following result to identify and prune perfect marker vertices:

The next remark explains how one can test whether a given marker vertex is perfect and will be used in the bottom-up and top-down traversal of $T(N(x))$.
%
%Thanks to the data-structure, each test requires constant time.

\begin{remark} \label{twinTest}
Let $q \in V(u)$ be a marker vertex in $ST(G)$, and let $r$ be $q$'s opposite.  Then $r$ is perfect if and only if:
\begin{enumerate}
\item $P(u) = N_{G(u)}(q)$, or $P(u) = N_{G(u)}[q]$; and
\item $NE(u) \setminus P(u) = \emptyset$ or $NE(u) \setminus P(u) = \{q\}$.
\end{enumerate}
\end{remark}  

Remark~\ref{twinTest} must not be applied, at the third step of our procedure, to every marker vertex of the unique remaining node $u$. Indeed, consider $q,q' \in V(u)$, and let $r$ and $r'$ be their opposites, respectively.  To test if $r$ and $r'$ are perfect using Remark~\ref{twinTest} requires us to test if $P(u) = N(q)$ and $P(u) = N(q')$.  But if $N(q) \cap N(q') \ne \emptyset$, then this involves ``touching'' marker vertices of $u$ multiple times.  In general, we cannot bound the number of times marker vertices in $u$ will have to be ``touched''.  

The solution for a degenerate node follows from the next lemma. In the case of a prime node, we will use the LBFS Lemmas~\ref{LBFStwin} and~\ref{inducedLBFS} of Section~\ref{sec:LBFS}.

%to restrict the number of times Remark~\ref{twinTest} is applied.  We'll show that it only needs to be applied a constant number of times at each node.  For degenerate nodes, this follows from the next result, which rephrases Remark~\ref{twinTest} in the degenerate setting:

% MT 08/27/12 - Reviewer wanted this changed to a lemma with the proof supplied.
\begin{lemma} \label{degenerateTwinTest}
Let $u$ be a degenerate node of $ST(G)$, the marker vertices of which are all either perfect or empty (i.e. $P(u) = NE(u)$).
There exists a marker vertex $q \in V(u)$ whose opposite $r$ is perfect if and only if one of the following conditions holds:

\begin{enumerate}
\item $P(u) = V(u)$ (in this case, if $u$ is a clique then any $q \in V(u)$ is suitable, and if $u$ is a star then $q$ is its centre);

\item $P(u) = V(u)\setminus \{q\}$ and, when $u$ is a star, $q$ is the centre of $u$; 

\item $P(u) = \{c\}$ and $u$ is a star with centre $c$ (in this case any $q \in V(u)\setminus \{c\}$ is suitable);

\item $P(u) = \{c,q\}$ and $u$ is a star with centre $c$.
\end{enumerate}
\end{lemma}  

\begin{proof}
% MT 08/27/12 - Christophe asks if "descendant" is correct in this context rather than marker vertex.  I think it is, especially if we are to apply lemma~\ref{lem:hereditary}, but I could be wrong.  Please review.
Choose some $q \in V(u)$ and let $r$ be its opposite marker vertex.  Assume that $r$ is perfect.  Then if $u$ is a clique or a star with centre $q$, all marker vertices in $V(u) - \{q\}$ are accessible descendants of $r$ and are therefore perfect, by Lemma~\ref{lem:hereditary}-1.  Therefore either items 1 or 2 of the lemma hold.  

So assume that $u$ is a star with centre $c \ne q$.  Then $c$ is an accessible descendant of $r$, and all marker vertices in $V(u) - \{c,q\}$ are inaccessible descendants of $r$.  Thus, $c$ is perfect and the marker vertices in $V(u) - \{c,q\}$ are empty, by Lemma~\ref{lem:hereditary}-1.  It follows that either items 3 or 4 of the lemma hold.  

Now assume that one of items 1-4 of the lemma holds.  If items~1 or~2 hold, then all marker vertices in $V(u) - \{q\}$ are accessible descendants of $r$, and all are perfect.  So by Lemma~\ref{lem:hereditary}-1, $r$ is perfect as well.  And if items~3 or~4 hold, then only $c$ is perfect, and only $c$ is an accessible descendent of $r$.  It follows that $r$ is perfect, once again by Lemma~\ref{lem:hereditary}-1.
\end{proof}

%EME remark reformulated since of confusing use, and hypthesis missing
%\begin{remark} \label{degenerateTwinTest}
%Let $q \in V(u)$ be a marker vertex in $ST(G)$ at the degenerate node $u$, and let $r$ be its opposite.  
%\begin{enumerate}
%\item If $q$ is not a degree-1 vertex of a star, then $r$ is perfect if and only~if:
%	\begin{enumerate}
%	\item $|P(u)| = |NE(u)|$;
%	\item $|P(u)| = |V(u)|$, or $|P(u)| = |V(u)| - 1$ and $q \notin P(u)$.
%	\end{enumerate}
%\item If $q$ is a degree-1 vertex of a star whose centre is $c$, then $r$ is perfect if and only if $|P(u)| = |NE(u)|$, and $P(u) = \{c\}$ or $P(u) = \{c,q\}$.
%\end{enumerate}
%\end{remark}  

%To apply Remark~\ref{degenerateTwinTest} at node $u$, we start by comparing $|P(u)|$ and $|NE(u)|$ for equality, or in other words, determine if $u$ has any mixed marker vertices.  If not, meaning $|P(u)| = |NE(u)|$, then in case 1 we go on to compare $|P(u)|$ and $|V(u)|$.  If they are equal, then any $r$ opposite a $q \in V(u)$ will be perfect. Thereby as $T(r)$ is perfect, node $u$ can be discarded from further consideration.  But if $|P(u)| = |V(u)| - 1$, then only the $r$ opposite the $q \notin P(u)$ can be perfect, and then node $u$ can be discarded.  The idea for case 2 is similar.

% MT 08/27/12 - No longer needed now that the a proof is supplied above.
%This remark comes directly from a case by case analysis and Corollary \ref{cor:hereditary}.
Applying Lemma~\ref{degenerateTwinTest} at node $u$ is straightforward as soon as $P(u)$ has been computed. Notice that in cases 2 and 3, the marker vertex $q$ is empty, but it can be determined without considering other empty marker vertices. Hence at most one empty marker vertex is involved in this step of the procedure.

Let us now turn to prime nodes. First observe the following remark, which is a straightforward application of the definitions:

\begin{remark} \label{twins}
Let $q \in V(u)$ be a marker vertex in $ST(G)$, and let $r$ be its opposite.  Let $t$ be a marker vertex added to $u$, made adjacent precisely to $P(u)$.  Then $r$ is perfect if and only if $q$ and $t$ are twins.
\end{remark}

%The preceding remark is a straightforward application of the definitions.  
%--------------
%----- Xtof
% FORMER LEXBFS LEMMAS.
%----- Xtof
%--------------

The next lemma merely translates Lemma~\ref{LBFStwin} to the split-tree; its corollary is the important result for our purposes:

\begin{lemma} \label{LBFStwinNode}
Let $\sigma$ be an LBFS of the connected graph $G + x$ in which $x$ appears last, and let $u$ be a prime node in $ST(G)$.  Let $r$ be the opposite of some $q \in V(u)$.  If $r$ is perfect, then $q$ is universal in $G(u)$ or $q$ appears last in $\sigma[G(u)]$.
%DGC6$\sigma(u)$.
\end{lemma}

\begin{proof}
Let $u'$ be the same as $u$ but with a new marker vertex $t$ adjacent precisely to $P(u)$.  Consider the GLT $(T',\mathcal{F}')$ that results from replacing $u$ with $u'$, and adding a new leaf $\ell$ opposite $t$.  Let $\sigma_{\ell}$ be the same as $\sigma$ but with $x$ replaced by $\ell$.  Since $t$ is only adjacent to $P(u)$, we have $N(\ell) \subseteq N(x)$.  Therefore $\sigma_{\ell}$ is an LBFS of $G+\ell$ in which $\ell$ appears last, and $\sigma_{\ell}[G(u')]$ 
%DGC6$\sigma_{\ell}(u')$ 
is an LBFS of $G(u')$ in which $t$ appears last, by Lemma~\ref{inducedLBFS} applied to the split-tree.  

Assume that $r$ is perfect.  Then $q$ and $t$ are twins, by Remark~\ref{twins}.  Therefore $u'$ is not prime.  But recall that $u$ was prime.  So by Lemma~\ref{LBFStwin}, either $q$ is universal in $G(u)$ or it is the penultimate vertex in $\sigma_{\ell}[G(u')]$.
%DGC6$\sigma_{\ell}(u')$.  
If it is the penultimate vertex in $\sigma_{\ell}[G(u')]$,
%DGC6$\sigma_{\ell}(u')$
 then it must be the last vertex in $\sigma[G(u)]$.
%DGC6$\sigma(u)$.  
\end{proof}

\begin{corollary} \label{primeTwinTest}
Let $\sigma$ be an LBFS of $G+x$ in which $x$ appears last.  Let $\ell$ be a leaf adjacent to a prime node $u$ in $ST(G)$, and let $q \in V(u)$ be $\ell$'s opposite.  Then $\ell$ is perfect if and only if $q$ is universal in $G(u)$ or $q$ appears last in $\sigma[G(u)]$.
%DGC6$\sigma(u)$.
\end{corollary}

\begin{proof}
A direct consequence of Remark~\ref{twins} and Lemma~\ref{LBFStwinNode}.
\end{proof}

Therefore to perform the third step of the case identification procedure, at most two marker vertices of the remaining prime node $u$ can have opposites that are perfect. Thus, Remark~\ref{twinTest} needs to be applied at most twice. Remember that our data-structure keeps track of these two marker vertices.  
%The importance of this simplification, along with the similar one for degenerate nodes, will be clear when we come to discuss the running-time of our algorithm.  For now, we remark that with these simplifications it becomes possible to describe how $T(q)$ is pruned for every perfect leaf or marker vertex $q$.  That is summarized in Algorithm~\ref{alg:perfectPruning}, whose correctness is straightforward.
%--------------------
\begin{algorithm} 
\KwIn{The rooted split-tree $ST(G)=(T,\mathcal{F})$ and an LBFS ordering $\sigma$ of $G+x$ where $x$ is the last vertex.}
\KwOut{A tree-edge $e$, with one extremity perfect and the other empty or perfect (if one exists); or a hybrid node $u$ (if one exists); or the fully-mixed subtree $T'$ of $T$ (if one exists).}

\BlankLine
\lnl{pruned:init} $T'\gets T(N(x))$, the result of Algorithm~\ref{alg:emptyPruning} with input $T$ and $N(x)$\;
Set the non-root marker vertices opposite leaves of $N(x)$ to \emph{perfect} (these markers belong to nodes of $T'$)\;

\BlankLine
%\ForEach{node $u$ in $T_p$ having states assigned to all its non-root marker vertices} {
\tcp{bottom-up traversal:  discard the pendant perfect subtrees}
\uIf{$T'$ contains more than one node}{
\lnl{pruned:botoom-up}	\While{there exists a non-processed node $u$ in $T'$ all of whose children in $T'$ are leaves}{
		let $q$ be the (non-root) marker vertex opposite $u$'s root marker vertex\;
		determine $q$'s state by applying Remark~\ref{twinTest}\;
%		\lIf{$q$ is perfect} {prune $T(q)$ from $T_p$\;}
		\lIf{$q$ is perfect} {remove $u$ from $T'$\;}
	}
}
\BlankLine
%\uIf{the root of $T_p$ is a leaf}{
%	let $u$ be its only child\;
%	set $u$'s root marker vertex to \emph{perfect}\;
%	}
%\Else{
%	let $u$ be the root of $T_p$\;
%	set $u$'s root marker vertex to \emph{empty}\;
%	}

%\lIf{the root of $T'$ is a leaf of $T$}{let $u$ be is unique child\;}
%\lElse{ 
let $u$ be the root of $T'$\;
%}

\BlankLine
\tcp{top-down traversal: discard the perfect subtree containing the root}
\lnl{pruned:top-down}\While{node $u$ has exactly one non-leaf child $v$ in $T'$ and no fully-mixed edge has been identified}{
%	let $q$ be the root marker vertex of $v$\;
	apply Remark~\ref{twinTest} to determine the state of  the root marker vertex $q$ of $v$\;
%	\lIf{$q$ is perfect} {prune $T(q)$ from $T_p$ and $u\gets v$\;}
%	\lIf{$q$ and its opposite are perfect} {\KwRet $e$, the tree-edge between $u$ and $v$\;}
	\lIf{$q$ is perfect} {remove $u$ from $T'$ and $u\gets v$\;}
	\lElse{the tree-edge $uv$ is fully-mixed\;}
	}

\BlankLine
\tcp{case identification}
\lnl{prune:identification}
\uIf{$T'$ contains a unique node $u$}{
	\uIf{$u$ is degenerate} {
		apply Lemma~\ref{degenerateTwinTest} to determine if there is a $q \in V(u)$ whose opposite $r$ is 		perfect\;
%		\lIf{such a $q$ exists} {prune $T(r)$\;}
		\lIf{such a $q$ exists}{
	%		\tcp{case 1 and 2 of Theorem~\ref{th:cases}}
			\KwRet the tree-edge $e$ of $T$ whose extremities are $q$ and $r$\;
			}
		\lElse{
	%		\tcp{case 3 of Theorem~\ref{th:cases}}
			\KwRet the hybrid node $u$\; }
	}
	\uElse{
	\tcp{twin-test in a prime node}
		let $q$ be the last marker vertex in $\sigma[G(u)]$ and let $r$ be its opposite\;
%DGC6$\sigma(u)$
		let $q'$ be $u$'s universal vertex (if it exists), and let $r'$ be its opposite\;
%		let $e$ be the tree-edge with extremities $r$ and $q\in V(u)$ (last marker vertex in $\sigma(u)$)\;
%		\If{$V(u)$ contains a universal vertex $q'$}{
%			\KwRet the tree-edge $e'$ whose extremities are $r'$ and $q'$\;
%			}
%		let $q$ be the marker vertex appearing last in $\sigma(u)$, and let $r$ be its opposite\;
%		let $q'$ be $u$'s universal vertex (if it exists), and let $r'$ be its opposite\;
		apply Remark~\ref{twinTest} to determine the states of $r$ and $r'$\;
%		\lIf{$r$ (respectively $r'$) is perfect} {prune $T(r)$ (respectively $T(r')$) from $T_p$\;} 
		\uIf{$r$ (respectively $r'$) is perfect} {
	%		\tcp{case 1 and 2 of Theorem~\ref{th:cases}}
			\KwRet the tree-edge $e$ of $T$ whose extremities are $q$ and $r$ (respectively $q'$ and $r'$)\;
			}
		\lElse{
	%		\tcp{case 3 of Theorem~\ref{th:cases}}
			\KwRet the hybrid node $u$\;
			}
		}
	}
\smallskip	
\lElse{\KwRet the fully-mixed subtree $T'$\;}

%\ForEach{node $u$ in $T_p$ having states assigned to all its marker vertices} {
%\uIf{$u$ is degenerate} {
%apply Remark~\ref{degenerateTwinTest} to determine if there is a $q \in V(u)$ whose opposite is perfect\;
%\lIf{such a $q$ exists} {prune $T(q)$\;}
%}
%\Else{
%let $q$ be the marker vertex appearing last in $\sigma(u)$, and let $r$ be its opposite\;
%let $q'$ be $u$'s universal vertex (if it exists), and let $r'$ be its opposite\;
%apply Remark~\ref{twinTest} to determine the states of $r$ and $r'$\;
%\lIf{$r$ is perfect} {prune $T(r)$ from $T_p$\;} 
%\lIf{$r'$ is perfect} {prune $T(r')$ from $T_p$\;}
%}
%}

%\KwRet{$T_p$}\;
\caption{Detection of perfect subtrees, and split-tree case identification} \label{alg:perfectPruning}
\end{algorithm}
%--------------------

\begin{lemma} \label{lem:PruningCorrectness}
Given an LBFS ordering of a connected graph $G+x$ and the split-tree $ST(G)$,  Algorithm~\ref{alg:perfectPruning} returns: 
\begin{itemize}
\item a tree-edge $e$, one of whose extremities is perfect and the other is either empty or perfect, if case 1, 2, 5 or 6 of Theorem~\ref{th:cases} applies; 
\item a hybrid node $u$, if case 3 or 4 of Theorem~\ref{th:cases} applies; 
\item the full-mixed subtree $T'$ of $T$, if case 7 of Theorem~\ref{th:cases} applies; 
\end{itemize}
It can be implemented to run in time $O((1+$\emph{\texttt{find-cost}}$) \cdot |T(N(x))|)$.
\end{lemma}
\begin{proof}
By Lemma~\ref{emptyPruningTime}, we know that $T(N(x))$, and thus $T'$, is node disjoint from every empty subtree of $ST(G)$. Clearly, before the while loop at line~\ref{pruned:top-down}, the current subtree $T'$ is disjoint from every pendant perfect subtree of $ST(G)$. Thus, if the root of the tree belongs to a perfect subtree $T(p)$, then the nodes of $T(p)\cap T'$ can only form a path of $T'$: otherwise $T(p)$ should contain a node $u$ with two non-root marker vertices $q$ and $r$ which are neither perfect nor empty, as $T(q)$ and $T(r)$ have not been removed so far, contradicting the fact that $T(p)$ is perfect.

%The fact that $T_p$ does not contains the  leaves opposite empty marker vertices is an obvious consequence of Lemma~\ref{emptyPruning}. And it follows from Lemma~\ref{emptyPruningTime} that $T_p$ is disjoint from every $T(q)$ where $q$ is an empty marker.

%It is obvious that any pendant subtree $T(q)$ such that $q$ is a perfect marker is pruned during the first loop of Algorithm~\ref{alg:perfectPruning}. Once done, if the root of the tree belongs to a perfect subtree $T(p)$, then the nodes of $T(p)$ can only form a path. Since otherwise $T(p)$ should contains a node $u$ with two non-root marker vertices $q$ and $r$ which are neither perfect nor empty, as they $T(q)$ and $T(r)$ have not been pruned so far: contradicting the fact that $T(p)$ is perfect.

Concerning the correctness of the third step (case identification), first observe that if, at line~\ref{prune:identification}, $T'$ contains more than one node, then every tree-edge in $T'$ is fully-mixed (otherwise one of its extremities would have been removed during the tree traversals). Thus case 7 of Theorem~\ref{th:cases} holds. So assume that $T'$ consists of one node. It follows from Lemma~\ref{degenerateTwinTest}, Remark \ref{twins} and  Corollary~\ref{primeTwinTest}, that a single tree-edge is returned by the algorithm if there exists a tree-edge with a perfect extremity and the other extremity either perfect or empty. Notice that we do not test the uniqueness of such a tree-edge. If none of the previous cases applies, then, by Theorem~\ref{th:cases}, $ST(G)$ contains a hybrid node which is correctly identified by the algorithm. This corresponds to cases 3 and 4 of Theorem~\ref{th:cases}. 

By Remark~\ref{rem:traversal}, the cost of performing the tree traversals is $O((1+\texttt{find-cost}) \cdot |T(N(x))|)$.
During the algorithm, either Remark~\ref{twinTest} or Lemma~\ref{degenerateTwinTest} is applied a constant number of times at each node.  The cost to apply each remark/lemma in the context of the data-structure presented above is clearly $O(|NE(u)|)$, where $NE(u) = P(u) \cup M(u)$ (recall Definition~\ref{def:states}).  But every $q \in NE(u)$ has its corresponding edge in $T(N(x))$, by Lemma~\ref{emptyPruning} and the definition of $NE(u)$.  So the total cost of applying Remark~\ref{twinTest} and Lemma~\ref{degenerateTwinTest} is $O(|T(N(x))|)$.  Of course, once $q$ is found to be perfect, then $u$ can be removed from $T'$ in constant time. 
\end{proof}

%It follows that the tree returned by Algorithm~\ref{alg:perfectPruning} can be used in place of $pr(T,\mathcal{F})$.  In particular, it can be used in Lemma~\ref{th:cases}, Propositions~\ref{lem:cases12},~\ref{lem:case3}, and~\ref{lem:case4}.  The previous remark says the two trees have identical nodes, and therefore the cleaning and contraction required by cases 3 and 4 of Lemma~\ref{th:cases} can proceed no differently.  %We focus on cleaning first and then turn to contraction.  

Before describing how the split-tree is updated in each of the different cases, let us point out how the states (perfect, empty, mixed) of marker vertices are computed (or not) during Algorithm~\ref{alg:perfectPruning}:

\begin{remark} \label{rem: state-algo}
Algorithm~\ref{alg:perfectPruning} assigns a state to a marker vertex $q$ and updates the data-structure accordingly only if $q$ is perfect. From these recorded \emph{perfect-state} fields, the states of all marker vertices involved in the output of Algorithm \ref{alg:perfectPruning} can be deduced.
\end{remark}

It is not a problem to avoid explicitly computing the state of all the marker vertices. Indeed, computing them would affect our complexity. Moreover, notice that in every case (see Propositions~\ref{prop:unique-node-cases123}, \ref{prop:hybrid-degen-case4},~\ref{prop:unique-edge-case56} and~\ref{prop:fully-mixed-case-7}) the knowledge of the perfect marker vertices is enough to determine the state of every other marker vertex which will be affected in the successive steps of the updates. For example, in a hybrid node, the non-perfect marker vertices are by definition empty. Similarly if $ST(G)$ contains a fully-mixed subtree $T'$, then a marker vertex of a node of $T'$ is empty if and only if it is not perfect and not incident to a tree-edge of $T'$. If follows that, once Algorithm~\ref{alg:perfectPruning} has been performed, 
we can conclude that the state of every useful marker vertex has been determined.
%we can abusively consider that the states of every marker vertex has been determined.

%----------------------------------------------------------------------------------------------------------------------
\subsection{Node-split and cleaning}

The node-split operation (see Definition~\ref{def:node-split}) is required when cases 4 and 7 of Theorem~\ref{th:cases} hold. Case 4, existence of a degenerate hybrid node $u$ (Proposition \ref{prop:hybrid-degen-case4}), only requires the node-split of $u$ according to $(P^*(u),V(u)\setminus P^*(u))$. Case 7 potentially implies a large number of node-splits since, before the contraction step, the cleaning of $ST(G)$ is necessary (Proposition \ref{prop:fully-mixed-case-7}). Notice that degenerate nodes are the only nodes that are ever node-split.

To be as efficient as possible, to perform a node-split we won't create two new nodes as seemingly required by the definition.  Instead, we will reuse the node being split so that only one new node has to be created.  This is presented in Algorithm~\ref{alg:nSplit}.

%--------------------
\begin{algorithm}
\KwIn{A rooted GLT with a node $v$ such that $G(v)$ contains the split $(A,B)$ having frontiers $A'$ and $B'$.}
\KwOut{The rooted GLT with nodes $u$ and $u'$ resulting of the node-split of $v$ with respect to $(A,B)$.}

\BlankLine
replace the vertices in $A$ with a new marker vertex $q$ adjacent precisely to $B'$\;
call the result $u'$\;

\BlankLine
create a new node $u$ consisting of the vertices in $A$, plus one new marker vertex $r$ adjacent precisely to $A'$\;

\BlankLine
add an internal tree-edge between $u$ and $u'$ having extremities $q$ and $r$\;
\lIf{the root marker vertex of $v$ belongs to $A$}{make $u'$ a child of $u$\;}
\lElse{make $u$ a child of $u'$\;}

\KwRet{the resulting GLT with $u$, $u'$ and their child relation}\;
\caption{Node-split$(v,A,B)$} \label{alg:nSplit}
\end{algorithm}
%--------------------

% A simple examination of Algorithm~\ref{alg:nSplit} in the context of the data-structure presented earlier reveals the folllowing when applied to a degenerate node:
 
\begin{lemma} \label{nSplitTime}
Algorithm~\ref{alg:nSplit} performs a node-split $(A,B)$ for a degenerate node in time $O(|A|)$.
\end{lemma}  
\begin{proof}
The correctness follows from the definitions. The time complexity is obvious as well from a simple examination of our data-structure. Recall that every child of a degenerate node maintains a parent pointer (unlike for prime nodes which use the union-find). Depending on whether the root marker vertex of $v$ belongs to $A$ or $B$, $u'$ becomes a child of $u$ or vice-versa. The root marker vertices (and their respective node pointers) of the resulting nodes are updated accordingly. As these two nodes are degenerate, they need to have a list of their marker vertices: $u'$ inherits the list of $v$ in which $A$ has been removed plus the new marker vertex $q$, while $u$'s list is created and contains $A\cup\{r\}$. Meanwhile, the marker vertices of $A$ update their node pointer. This work requires $O(|A|)$ time. Other information such as type of the node, number of children, pointer to the centre (if it is a star) or opposite pointer, perfect-states, is easily updated in constant time. 
\end{proof}

% MT 08/27/12 - Reviewer wanted reference to definition of cleaning, which appeared many pages earlier in the paper.
Cleaning was introduced in Definition~\ref{def:cleaning} along with the notation $c\ell(ST(G))$.  Notice that it amounts to 
%Cleaning amounts to 
repeated application of the node-split operation. The cleaning step will proceed according to Algorithm~\ref{alg:cleaning}.
To determine if a degenerate node needs to be node-split according to $(E^*(u),V(u) \setminus E^*(u))$, Algorithm~\ref{alg:cleaning} takes advantage of the equivalence: $E^*(u) = V(u) \setminus (P^*(u) \cup M(u))$, where $P^*(u)$ and $M(u)$ are deduced directly from perfect-states fields and the fully-mixed subtree structure. As before, the reason is one of efficiency; we want to avoid ``touching'' empty subtrees.  To this end, it is important that the node-split is performed only ``touching'' perfect and mixed marker vertices of $V(u) \setminus E^*(u)$. This is the reason for defining Algorithm~\ref{alg:nSplit} as we did. The rest of Algorithm~\ref{alg:cleaning} is a direct implementation of the definition.

%EME below was not clear
%It is worth focusing on the manner by which Algorithm~\ref{alg:cleaning} determines if a degenerate node needs to be node-split according to $(E^*(u),V(u) - E^*(u))$; the rest of the algorithm is a direct implementation of the definition.  The algorithm takes advantage of the equivalence: $E^*(u) = V(u) - (P^*(u) \cup M(u))$.  As before, the reason is one of efficiency: we want to avoid ``touching'' empty subtrees.  To this end, it is important that the node-split is performed only ``touching'' perfect and mixed marker vertices of $V(u) - E^*(u)$. This is the reason for defining Algorithm~\ref{alg:nSplit} as we did.

%--------------------
\begin{algorithm}[h]
\KwIn{The rooted split-tree $ST(G)=(T,\mathcal{F})$ marked with respect to $N(x)$ and the fully-mixed subtree $T'$ of $T$.}
%all of whose leaves and marker vertices have been assigned states.}
\KwOut{The rooted GLT $cl(ST(G))$ resulting from the cleaning of $(T,\mathcal{F})$ and the fully-mixed subtree $T_c$ of $cl(ST(G))$.}

\BlankLine
Assume all nodes are marked \emph{unvisited}\;
$T_c \gets T'$\;

\BlankLine
\ForEach{unvisited degenerate node $v$ in $T_c$ } {
	%\uIf{$|P^*(v)| > 1$ and $|V(v)| - |P^*(v)| > 1$} {
	\uIf{$|P^*(v)| > 1$} {
		node-split $v$ according to the split $(P^*(v),V(v) \setminus P^*(v))$, using Algorithm~\ref{alg:nSplit}\;
		mark \emph{visited} the two nodes that result from the node-split\;
		keep in $T_c$ only the node containing the marker vertices in $V(v) \setminus P^*(v)$\;
%		make $u$ (remind $P^*(u)\subseteq V(u)$) a child of $u'$\;
		}
	}

\BlankLine
reset all the marks in $T_c$\;

\BlankLine
\ForEach{unvisited degenerate node $v$ in $T_c$} {
	%\uIf{$|V(v)| - (|P^*(v)| + |M(v)|) > 1$ and $|P^*(v)| + |M(v)| > 1$} {
	\uIf{$|V(v)| \setminus (|P^*(v)| + |M(v)|) > 1$} {
		node-split $u$ according to the split $(V(v) \setminus E^*(v),E^*(v))$, using Algorithm~\ref{alg:nSplit}\;
		mark \emph{visited} the two nodes that result from the node-split\;
		keep in $T_c$ only the node containing the marker vertices in $V(v) \setminus E^*(v)$\;
%		make $u'$ (remind $E^*(u')\subseteq V(u')$) a child of $u$\;
		}
	}

\KwRet{the updated GLT, $c\ell(ST(G))$ together with its fully-mixed subtree $T_c$}\;
\caption{Cleaning$((T,\mathcal{F}),T')$} \label{alg:cleaning}
\end{algorithm}
%--------------------

\begin{lemma} \label{cleaningTime}
Given the split-tree $ST(G)=(T,\mathcal{F})$ marked with respect to $N(x)$ and the fully-mixed subtree $T'$ of $T$, Algorithm~\ref{alg:cleaning} computes $c\ell(ST(G))$ together with its fully-mixed subtree, and it runs in time $O((1+\mbox{\emph{\texttt{find-cost}}}) \cdot |T(N(x))|)$.
%Given a GLT $(T,\mathcal{F})$ and the subtree $T'$ of $T$ containing the nodes of $pr(T,\mathcal{F})$, Algorithm~\ref{alg:cleaning} computes $cl(pr(T,\mathcal{F}))$ and it runs in time $O((1+\mbox{\emph{\texttt{find-cost}}}).|T(N(x))|)$.
\end{lemma}

\begin{proof}
Correctness is clear as Algorithm~\ref{alg:cleaning} traverses $T'$ twice. In both traversals, a single test is performed at each node, and then if it succeeds, a node-split is applied.  
Recall that $(P^*(u),V(u) \setminus P^*(u))$, respectively $(E^*(u),V(u) \setminus E^*(u))$, is a split of $u$ if and only if $| P^*(u) |>1$, respectively $| E^*(u) |>1$, by Remark \ref{rk:cleaning}. The resulting subtree $T_c$ is the fully-mixed subtree of $cl(ST(G))$ by construction. 

The traversals require time $O((1+\texttt{find-cost}) \cdot |T(N(x))|)$. To perform the test required by Algorithm~\ref{alg:cleaning}, we first need to compute $P^*(u)$.  But this clearly can be done in time bounded by $O(|P(u)|)$, which is bounded by $O(|NE(u)|)$.  So by the argument already used in the proof of Lemma~\ref{lem:PruningCorrectness}, the total cost of computing the sets $P^*(u)$ is $O(|T(N(x))|)$.  Once $P^*(u)$ is computed, the test performed at each node can be carried out in constant time.  If the degenerate node $u$ is node-split during the first pass, then the set $P^*(u)$ plays the role of $A$ in the input to Algorithm~\ref{alg:nSplit}; if $u$ is split during the second pass, then the set $V(u) \setminus E^*(u)$ plays the role of $A$.  The size of both of these sets is bounded by $|NE(u)|$.  But every $q \in NE(u)$ has its corresponding edge in $T(N(x))$, by Lemma~\ref{emptyPruning} and the definition of $NE(u)$.  So the total cost of the node-splits performed during cleaning is $O(|T(N(x))|)$, by Lemma~\ref{nSplitTime}.  
\end{proof}

No \texttt{union()} operation has so far been needed; it has been sufficient to employ \texttt{find()} operations while traversing various trees.  
%The \texttt{union()} operation is only needed by contraction, which is discussed next.

%----------------------------------------------------------------------------------------------------------------------
\subsection{Node-joins and contraction}
\label{sec:contraction}

Contraction amounts to repeated application of the node-join, which requires the  \texttt{union()} operation.  In the same way that we reused one node for the node-split, we will want to reuse one node for the node-join (we arbitrarily choose to reuse the parent node).  Algorithm~\ref{alg:nJoin} provides the details of the implementation of the node-join between a node $u'$ and its parent $u$. Notice that in the case where $u'$ is a star and its root marker vertex $q'$ has degree one, a node-join is performed differently. Here we reuse the marker vertex $q$ of $u$ adjacent to $u'$ to play the role of the unique neighbour $t$ of $q'$. We do so for reasons of efficiency that will become clear later.
In other cases, the label-edges adjacent to the two marker vertices disappearing in the node-join are not reused.
%--------------------
\begin{algorithm}
\KwIn{A rooted GLT with two adjacent nodes $u$ and $u'$, where $u'$ is the child of $u$.}
\KwOut{The rooted GLT resulting from the node-join of $u$ and $u'$.}

\BlankLine
let $q \in V(u)$ and $q' \in V(u')$ be the extremities of the tree-edge between $u$ and $u'$\;

\BlankLine
\uIf{$u'$ is a star node and its root marker vertex $q'$ has degree one}{
\lnl{line:type2}
	\ForEach{non-neighbour $t'$ of $q'$ in $G(u')$}{
		move $t'$ to $G(u)$ and make it adjacent to $q$\;
%DGC7 ($q'$ is no longer a neighbour of $t'$)\;
		let $v$ be the child of $u'$ containing the marker vertex opposite $t'$\;
		update $v$'s parent pointer to $u$\;
		}
	}

\BlankLine
\lnl{line:type1-3}
\uElse{
	move all the marker vertices of $G(u')$ except $q'$ to $G(u)$ and remove $q$\;
	add adjacencies in $G(u)$ between every neighbour of $q'$ and every neighbour of $q$\;
	\uIf{$u'$ is prime}{
		let $v$ be the representative of the children-set of $u'$\;
		update the parent pointer of $v$ to $u$\;
		}
	\lElse{
		update the parent pointer of every child $v$ of $u'$ to $u$\;
		}
	}
\smallskip	
\lnl{line:children-set}
create a single children-set containing the children of $u$ and $u'$ (by the way of a series of \texttt{initialize()} and \texttt{union()})\;

\KwRet{the resulting rooted GLT}\;
\caption{Node-join$(u, u')$} \label{alg:nJoin}
\end{algorithm}

The node resulting from the node-join of $u$ and $u'$ in Algorithm~\ref{alg:nJoin} may be neither prime nor degenerate.
In the data-structure of the resulting rooted GLT, this resulting node  is 
nevertheless marked as prime (standing for non-degenerate), as explained in Remark \ref{rem:data}, since it will finally be prime after %the contraction and 
insertion of the new vertex.
The children-set associated with this node is created at line \ref{line:children-set}, and may require several \texttt{initialize()} and \texttt{union()} operations when $u$ and/or $u'$ were degenerate and thus were not associated with a children-set.

\begin{lemma} \label{nJoinTime}
Let $u$ and $u'$ be two adjacent nodes of a GLT and let $q$ and $q'$ be the respective extremities of the tree-edge between $u$ and $u'$.
Algorithm~\ref{alg:nJoin} computes the GLT resulting from the node-join of $u$ and $u'$. It can be implemented to run in time  $\texttt{\emph{join-cost}}=$O(\emph{\texttt{\#new-label-edge}})$+\texttt{\emph{tree-update-cost}}$ where:
\begin{itemize}
\item \emph{\texttt{\#new-label-edge}} denotes the number of newly created label-edges
(i.e. at line \ref{line:type1-3}: the number of neighbours of $q$ multiplied by the number of neighbours of $q'$).
%if $u'$ is a star with $d'$ leaves and its root marker vertex $q'$ has degree one and $O(d.d')$ otherwise.
\item the \texttt{\emph{tree-update-cost}} amounts to
\begin{itemize}
\item $O(d \cdot (\texttt{union-cost}+\texttt{initialize-cost}))$ if $u$ is a degenerate node with $d$ children, plus
\item $O(d' \cdot (\texttt{union-cost}+\texttt{initialize-cost}))$ if $u'$ is a degenerate node with $d'$ children, plus
\item $O(\texttt{union-cost})$ to perform the union of children-sets of $u$ and $u'$.
\end{itemize}
\end{itemize}
\end{lemma}
\begin{proof}
The correctness easily follows from the definitions.
%
%Observe that by Proposition~\ref{prop:fully-mixed-case-7}, the node resulting from the contraction phase is always prime. As a node-join only occurs during contraction, when Algorithm~\ref{alg:nJoin} is applied, 
The resulting node is assigned the prime type and therefore is associated with a children-set, 	and its graph label has to be represented by an adjacency list. Concerning the series of \texttt{union()} and \texttt{initialize()} requests to build the children-set: if a degree $d$ degenerate node is involved in a node-join, then the cost to create a set containing its children amounts to the cost of $d$ \texttt{initialize()} and $d$ \texttt{union()} requests. Concerning the adjacency-lists: we first create those not already present at degenerate nodes.  Then the existing ones can be combined to create one for the result of the node-join.  This can be done in the obvious way. 
\end{proof}

Recall that node-joins on a given set of tree-edges can be performed in any order (Remark~\ref{rem:node-join-commutative}).
But to ease the amortized time complexity (developed in Section~\ref{sec:runningTime}), during the contraction step, different types of node-joins are performed before others. This is reflected in Algorithm~\ref{alg:contraction}, which separates the node-joins into three phases, defined as \emph{phase 1}, \emph{phase 2}, and \emph{phase 3} node-joins,
and dealing with different types of node-joins.
%The first phase performs node-joins involving star nodes whose root marker vertex is their centre. Those node-joins will be called \emph{phase 1} node-joins.  The second phase performs node-joins, called \emph{phase 2} node-joins, involving nodes whose root marker vertex has degree one.  The third phase recursively performs whatever node-joins, called \emph{phase 3} node-joins, are required to contract the tree into a single node.  

%--------------------
\begin{algorithm}[h]
\KwIn{The rooted GLT $(T,\mathcal{F})=c\ell(ST(G))$ and the fully-mixed subtree $T'$ of $T$} %such that the perfect and mixed maker vertices of every node $u\in T'$ are identified.}
%all of whose leaves and marker vertices have been assigned states, a new vertex $x$ to be inserted.}
\KwOut{The rooted GLT resulting from the contraction of $T'$ into a single node $u$ to which the new leaf $x$ has to be attached, with $x$'s opposite made adjacent in $G(u)$ precisely to $P(u)$.}
%A single node $u+x$.  The neighbours of $u+x$ are $x$, plus the leaves of $(T,\mathcal{F})$.  The label of $u+x$ is the accessibility graph of $(T,\mathcal{F})$ with $x$'s opposite added and made adjacent precisely to $P(u)$.}

\BlankLine
\tcp{Phase 1 node-joins}
\ForEach{star node $u$ in $T'$ whose root marker vertex is its centre}{
	perform node-join$(u,u')$ (Algorithm~\ref{alg:nJoin}) for every non-leaf child $u'$ of $u$ in $T'$\;
	 update $T'$ accordingly\;
	}

\BlankLine
\tcp{Phase 2 node-joins}
\ForEach{node $u'$ in $T'$ whose root marker vertex $r$ has degree 1}{
	\lIf{the parent node $u$ of $u'$ is not a leaf and is in $T'$}{perform node-join$(u,u')$ (Algorithm~\ref{alg:nJoin})\;}
		 update $T'$ accordingly\;
	}

\BlankLine
\tcp{Phase 3 node-joins}
recursively perform node-joins to contract $T'$ into a single node $u$, using Algorithm~\ref{alg:nJoin}\;

\BlankLine
\tcp{New vertex insertion}
add a marker vertex $q$ to $u$, adjacent precisely to $P(u)$, then make $x$ opposite~$q$\;
let $u+x$ be the resulting node and mark $q$ as the last marker vertex of $\sigma[G(u+x)]$\;
%DGC6 $\sigma(u+x)$
$x$'s parent is $u+x$\;
\KwRet{the resulting rooted GLT}\;
\caption{Contraction$((T,\mathcal{F}),T')$}\label{alg:contraction}
%DGC8 \hfill - contains the new vertex insertion -} 
\end{algorithm}

\begin{lemma} \label{contractionTime}
%For a split-tree $ST(G)=(T,\mathcal{F})$, Algorithm~\ref{alg:contraction} computes the node $u+x$ with which $pr(c\ell(pr(T,\mathcal{F})))$ is substituted to obtain $ST(G+x)$. If $k$ is the number of node-joins required for the insertion of $x$, then Algorithm~\ref{alg:contraction} runs in time 
If $ST(G)$ contains a fully-mixed subtree, given $c\ell(ST(G))=(T,\mathcal{F})$ and its fully-mixed subtree $T'$, Algorithm~\ref{alg:contraction} computes $ST(G+x)$. It can be implemented to run in time:
$$O(|T(N(x))| \cdot (1+\emph{\texttt{find-cost}}) + k \cdot \emph{\texttt{join-cost}} + \emph{\texttt{intialize-cost}}+\emph{\texttt{union-cost}})$$
where $k$ is the number of fully-mixed tree-edges of $T'$.
\end{lemma}

\begin{proof}
From Remark~\ref{rem:node-join-commutative}, the GLT resulting from the node-joins between nodes incident to the tree-edges of $T'$ is independent of the order in which they are applied. If $u$ is the node resulting from the %contraction of $T'$ (together with the new vertex insertion), 
contraction of $T'$ (together with the new vertex insertion), 
then observe that $\cup_{q\in P(u)} A(q)=N(x)$. If follows that the accessibility graph of the resulting GLT is $G+x$. Now from
Proposition~\ref{prop:fully-mixed-case-7}, $u$ is prime thereby showing that Algorithm~\ref{alg:contraction} computes $ST(G+x)$ since its result is reduced.

%Because the split-tree is reduced, it is impossible for a star node whose centre is its root marker vertex to have a child that is a star whose centre is its root marker vertex.  Therefore phase~1 node-joins can be performed in any order.  Phase 2 node-joins require further discussion.  First, we observe that the node $u'$ involved must be a star. Therefore the node-join between $u'$ and its parent $u$ is performed as in step~\ref{line:type2} of Algorithm~\ref{alg:nJoin}. It follows that every tree-edge of $pr(c\ell(pr(T,\mathcal{F})))$  is contracted by Algorithm~\ref{alg:contraction}. Correctness follows from Lemma~\ref{nJoinTime}.

The nodes participating in each phase can be located by performing a single pass over the tree.  This traversal can be performed in time $O((1+\texttt{find-cost}) \cdot |T(N(x))|)$. Once the nodes participating in each phase have been located, the node-joins can proceed.  The cost of each one is described by Lemma~\ref{nJoinTime}.  The outcome of these node-joins is a single node $u$.  At this point, all of $u$'s children are organized into a children-set.  

The new vertex $x$ is made a neighbour of $u$, which in our situation means it is made a child of $u$.  This requires $x$ to be added to the set that represents $u$'s children.  To do so we need one \texttt{initialization()} operation and one  \texttt{union()} operation.  The opposite of $x$ is a new marker vertex $q$ added to $G(u)$ and made adjacent precisely to $P(u)$.  Recall that our split-tree algorithm inserts vertices according to an LBFS ordering, say $\sigma$.  Notice that $q$ clearly becomes the last vertex in $\sigma[G(u)+q]$.  
%DGC6' $\sigma(u)$
So given the data-structure assumed earlier, the cost of adding $q$ is $O(|P(u)|)$.  But every marker vertex in $NE(u)$ has its corresponding edge in $T(N(x))$, by Lemma~\ref{emptyPruning} and the definition of $NE(u)$.  Therefore the total cost of adding $q$ is $O(|T(N(x))|)$.  
\end{proof}

At this point of the paper, the reader can completely compute the split decomposition of a graph with our algorithm.
It is the number of node-joins involved in successive uses of Lemma \ref{contractionTime}
for contraction that prevents us from concluding 
its running time.
We have already seen one example where the number of node-joins required by contraction is linear in the size of the split-tree (see Figure~\ref{fig:badExample}).  But later we emphasized that this was worst-case behaviour.  We promised that our LBFS ordering would make it possible to amortize the cost of contraction.  We finally prove this in the next section.  

%EME I shortened a lot the paragraph below since I am usually not fond of comments at the end of a section to introduce the next one...
%It is the number of node-joins involved in successive uses of Lemma \ref{contractionTime}
%for contraction that prevents us from concluding 
%its running time.
%%%%%%the running-time of our algorithm at this point.  
%So far the number of \texttt{find()} operations required by our algorithm has always been bounded by $O(T(|N(x)|)$.  But right now we do not know enough to bound the number of  \texttt{union()} and  \texttt{initialization()} operations.  Nor do we know the cost of building the adjacency-lists for degenerate nodes participating in contraction.  Further still, we do not know the cost of creating the new label edges resulting from the node-join (see Lemma~\ref{nJoinTime}).  
%
%We have already seen one example where the number of node-joins required by contraction is linear in the size of the split-tree (see Figure~\ref{fig:badExample}).  But later we emphasized that this was worst-case behaviour.  We promised that our LBFS ordering would make it possible to amortize the cost of contraction.  We finally prove this in the next section.  

%----------------------------------------------------------------------------------------------------------------------
%\newpage
%----------------------------------------------------------------------------------------------------------------------
%----------------------------------------------------------------------------------------------------------------------
\section{An ammortized running time analysis} \label{sec:runningTime}

This section completes the proof for the running time of our algorithm, described completely in  Section \ref{sec:algorithm}.
Our main goal is to amortize the cost of contraction (Lemma \ref{contractionTime}), 
involving the number of updates and requests to the union-find data-structure, and the number of created label-edges and vertices involved in the adjacency lists of label graphs. 
%We also need to bound the number of degenerate marker vertices created at the cleaning step.
% which depends on the \texttt{initialization()} and  \texttt{union()} operations performed over the entire split-tree construction. 
\smallskip

So far, the number of \texttt{find()} operations required by our algorithm has always been bounded by $O(T(|N(x)|)$. This will directly imply a suitable bound (by Lemma \ref{lem:numSubtree} in Subsection \ref{sub:final}).

The \texttt{initialization()} routine is always performed just prior to a \texttt{union()} during a node-join operation and it involves a child of a degenerate node or the new vertex to be inserted (Algorithm \ref{alg:nJoin} line \ref{line:children-set}). It follows that the number of  \texttt{initialization()} operations is bounded by the number of non-root marker vertices belonging to a degenerate node that appear at some step of the algorithm (when $x$ is inserted or when a node-split is performed).
Bounding the number of such vertices is also required since they participate in the data-structure (see Subsection \ref{sub:degenerate-marker}).

The  \texttt{union()} operations are performed during a node-join (Algorithm \ref{alg:nJoin} line \ref{line:children-set})
once together with each \texttt{initialization()} operation, and once to finalize the node-join.
The total number for the first part is bounded the same way as \texttt{initialization()} operations, and the total number for the second part is bounded by the total number of node-joins, which is bounded by the total number of created label-edges, since each node-join implies the creation of a label-edge
(Lemma \ref{nJoinTime}).

Therefore, bounding the number of created label-edges is the key to the complexity analysis and is the difficult part of our complexity argument.
%
%Recall that a degenerate node does not store any label edge and that every node-join eventually leads to the creation of a new prime node. 
%EME beloaw false !!!
%It follows that label edges are formed during contraction only and are never deleted. 
To count and bound the number of label-edges created during the whole algorithm, we use a charging argument in which the role of LBFS is critical (see Subsection \ref{sub:charging}). 

%----------------------------------------------------------------------------------------------------------------------
\subsection{Bounding the number of degenerate marker vertices}
\label{sub:degenerate-marker}

We prove that the number of non-root marker vertices  that are created in some degenerate node during the process of building $ST(G)$ is linearly bounded by the number of vertices of the input graph $G$. The idea  is to show that at each vertex insertion, only a constant number of such marker vertices are generated by our incremental algorithm.
%A bound on the total number on non-root marker vertices that are created in some degenerate node during the process of building $ST(G)$ provides a bound on the number \texttt{initialization()} operations.

\begin{lemma} \label{lem:degenerate-marker}
Let $G$ be a connected graph. The insertion of vertex $x$ in the process of building $ST(G+x)$ creates at most two new non-root marker vertices belonging to a degenerate node in the rooted GLT data-structure. 
%Moreover the degree of each of these at most three new marker vertices is bounding by $|N(x)|$.
\end{lemma}
\begin{proof}
Consider forming the split-tree $ST(G+x)$, where $x$ is some new vertex not already in $G$.  We consider the changes required of $ST(G)$ to form $ST(G+x)$, as described by Theorem~\ref{th:cases} and the subsequent propositions.  Notice that the set of marker vertices belonging to some degenerate node is modified in three different ways:

\begin{itemize}
\item \textit{the leaf $x$ is attached to a degenerate node.} This occurs when cases 1, 2, 4, 5 and 6 of Theorem~\ref{th:cases} apply. Two subcases are to be considered. If the degenerate node $u$ neighbouring leaf $x$ has degree $3$ (that is case 4, 5 or 6 holds and $u$ is a new node), then exactly two new non-root marker vertices have been created. It also follows that the degree in $G(u)$ of these two marker vertices is at most two. Otherwise (case 1 or 2), the only new non-root marker vertex $q\in V(u)$ is the opposite of $x$. 
%Notice that the degree of $q$  in $G(u)$ is exactly $|P(u)|$ which is bounded by $|N(x)|$.

\item  \textit{a node-split is performed on a degenerate node.} This occurs when cases 4 or 7 (during cleaning) of Theorem~\ref{th:cases} applies. Observe that every split creates exactly two new marker vertices, one of which is the root of its node. In case 4, only one node-split is performed, thereby creating one extra non-root marker vertex in a degenerate node. 

So let us consider the node-splits performed during the cleaning step when case 7 holds. Each degenerate node $u$ of the fully-mixed subtree is involved in at most two node-splits (see Figure \ref{fig:case4cleaning}). Among the two nodes resulting from each node-split, one will eventually be node-joined to form a prime node and the other remains degenerate in $ST(G+x)$.  Let us call $v$ such a created degenerate node.
All marker vertices inherited by $v$ through the node-split are reused, and remain non-root marker vertices if they were non-root marker vertices in $u$. Hence the only case where a non-root marker vertex is created in $v$ is when $v$ inherits the root marker of $u$ and thus a non-root marker vertex is created as the extremity of the new tree-edge resulting from the split.
%DGC9 I don't think this last phrase is right.
%Xtof : the created marker vertex in $v$ is non-root in $v$; that is, 
%when the root marker vertex of $u$ is inherited by $v$ and becomes the root marker vertex of $v$. 
Of course, this case can happen at most once for any degenerate node $u$ affected by a node-split. Moreover, this can only happen at the node at the root of the fully-mixed subtree. Thus at most one non-root degenerate marker is created during the series of node-split required by $x$'s insertion.

%Since $ST(G)$ is rooted, at most two of these resulting degenerate nodes contains a new non-root marker vertex. 

%Moreover observe that when a node-split is performed, the degree of a persisting marker vertex always decreases. So we proceed as follows: the edge-label 
%every non-root persisting marker vertex

\item \textit{a node-join is performed and involves a degenerate node.} This only occurs during the contraction step while a prime node is being formed. In this case, marker vertices of a degenerate node are lost. The invariant trivially holds.
\end{itemize}
\end{proof}

From the previous lemma, we can conclude the following:

\begin{lemma} \label{initializationCost}
%The total cost of \texttt{initialization()} operations to the construction of $ST(G+x)$ is $O(n+m)$, where $n$ is the number of vertices in $G+x$, and $m$ is the number of edges.
The total cost of \emph{\texttt{initialization()}} operations to the construction of $ST(G)$ is $O(n)$, where $n$ is the number of vertices in $G$.
\end{lemma}

\begin{proof}
Each \texttt{initialization()} operation takes constant time.  They are only employed prior to a node-join involving a degenerate node (see Lemma~\ref{nJoinTime}), and when the vertex $x$ to be inserted is made a neighbour of a newly formed prime node (see Lemma~\ref{contractionTime}).  Of the latter, there can be at most $O(n)$.  The \texttt{initialization()} operations of the former are dealt with below. Every such \texttt{initialization()} operation corresponds to a child of a degenerate node, or equivalently to a non-root marker vertex of a degenerate node. Thus by Lemma~\ref{lem:degenerate-marker}, building $ST(G)$ requires $O(n)$ calls to \texttt{initialization()}.
\end{proof}

%----------------------------------------------------------------------------------------------------------------------
\subsection{Bounding the number of label-edges}
\label{sub:charging}

We shall first recall that, in our data-structure, degenerate nodes do not store any label-edge. label-edges belong to prime nodes, which are only formed by contraction.  So, the label-edges in the resulting prime node either existed previously or were created by a node-join during contraction (Lemma \ref{nJoinTime}). 
Recall also that label-edges adjacent to marker vertices that disappear during a node-join are lost since these are not reused, but of course they count in the total number of created label-edges. 
%As already claimed, bounding the number of created label edges is required to bound the cost of updating adjacency lists of label graphs as well as the total cost of the node-joins and thereby the number of calls to \texttt{union()} required by contraction.  

%To do this we develop a charging argument driven by a stamping scheme. 
To bound the number of created label-edges, we develop a charging argument driven by a stamping scheme. 
The stamps will help us to spread and distribute the charge over the successive steps of our algorithm. The charging argument depends on our LBFS ordering.  Keep in mind that the split-tree construction algorithm does not involve the stamping scheme, nor the subsequent charging argument. These are only defined for the sake of counting created label-edges and of the amortized complexity analysis.

Let us sketch the construction. First, we need to show as a preliminary result that our LBFS ordering regulates how stars are formed during the construction of the split-tree: see Subsubsection \ref{sub:lbfs_and_stars}. Second, every non-root marker vertex is associated with some \emph{stamps} (0, 1 or 2 depending on its type), which are vertices of the input graph: see Subsubsection \ref{sub:stamping}. Independently, marker vertices are associated with \emph{Charge lists} such that spreading units of charge in the lists during the incremental process serves to count created label-edges:  see Subsubsection \ref{sub:charging}.
Lastly, the way stamps and Charge lists are associated with marker vertices will allow us to evaluate the total number of units of charge in terms of parameters of the input graph (number of edges and vertices), and hence to get the awaited complexity bound: see Lemma \ref{chargeBound} in Subsection \ref{sub:final}.

%----------------------------------------------
\subsubsection{LBFS and stars}
\label{sub:lbfs_and_stars}

Let us introduce extra notation and definitions related to an LBFS ordering $\sigma$ (see Section~\ref{sec:LBFS}). First we will abusively use $x_i$ instead of $\sigma^{-1}(i)$ to denote the $i$-th vertex in $\sigma$. Then $G_i$ stands for the subgraph induced by the subset $\{x_1,\dots x_i\}$ of vertices and we denote by $B_i$ the set of vertices appearing before $x_i$ (not including $x_i$).

\begin{definition}
Let $\sigma$ be an LBFS ordering of a connected graph $G$. A subset of consecutive vertices $S=\{x_i,\dots x_j\}$ (with $i \leq j$) is a \emph{slice} of $\sigma$ if for every $y\in S$, $N(y)\cap B_i=N(x_i)\cap B_i$. The set $Slice(x_i)$ denotes the largest slice starting at vertex $x_i$.
\end{definition}

In addition to the properties proved in Subsection~\ref{sub:LBFS}, LBFS controls the formation of star nodes.  During the split-tree construction process, some new marker vertices appear (e.g. the one opposite the new leaf, or when a node-split is performed), some disappear (when a node-join is performed) and some others are kept. More formally: 

\begin{definition}
Let $\sigma$ be an LBFS ordering of a connected graph $G$. 
Building on the definition given in Subsection \ref{subsec:split-tree}, 
we say that $ST(G_{i+1})$
\emph{inherits} a marker vertex $q$ of a node in $ST(G_i)$ 
if $q$ is not the extremity of a fully-mixed tree-edge of $ST(G_i)$ marked by the neighbourhood of $x_{i+1}$. \\
%(indeed, such a $q$ is inherited by the resulting GLT through any node-join operation involving a fully-mixed edge). %\\
By extension, $ST(G_j)$, with $j>i+1$, \emph{inherits} the marker vertex $q$ from $ST(G_i)$ if $ST(G_{j-1})$ inherits $q$ from $ST(G_i)$ and $q$ is not the extremity of a fully-mixed tree-edge of $ST(G_{j-1})$ marked by the neighbourhood of $x_{j}$.
\end{definition}

Recall that in a split-tree the centre of a star is never the opposite of a degree-1 marker vertex (since otherwise the split-tree wouldn't be reduced). For our charging argument, we need to extend this property over the life time of a marker vertex that was created as the centre of a star, assuming the vertex insertion follows an LBFS ordering.

\begin{lemma} \label{lem:star}
Let $\sigma$ be an LBFS ordering of a connected graph $G$. Assume that to insert vertex $x_i$, a %DGC10degree 
degree three node $u_i$ labelled by a star has been created. Let $c_i$ be the centre of $u_i$ and $q_i$ be the degree-1 marker vertex of $u_i$ not opposite $x_i$. If $x_j\in Slice(x_i)$, %(the largest slice starting at $x_i$), 
then $ST(G_j)$ contains a star node $u_j$ which contains $c_i$ as centre and $q_i$ as one of its degree-1 marker vertices. Moreover $u_j$ contains a degree-1 marker vertex $q_j$ such that $x_j\in L(q_j)$.
%Assume that $x_i$ is adjacent to a star $u_i$ in $ST(G_i)$. Let $c_i$ the centre of $u_i$ and $q_i$ be $x_i$'s opposite. If $x_j\in S_i$, then $ST(G_j)$ contains a star node $u_j$ which contains $c_i$ as centre and $q_i$ as one of its degree one marker vertices. Moreover $u_j$ contains a degree one marker vertex $q_j$ such that $x_j\in L(q_j)$.
\end{lemma}

\begin{proof}
The result clearly holds if $i = j$. Consider the case $j=i+1$. Notice that $x_{i+1}$ is a twin of $x_i$ since $x_{i+1}\in Slice(x_i)$. This implies that, if $ST(G_i)$ is marked by the neighbourhood of $x_{i+1}$, then $c_i$ is perfect and $q_i$ is empty. In other words the tree-edges respectively incident to $c_i$ and $q_i$ are not fully-mixed. So by definition, $c_i$ and $q_i$ are inherited by $ST(G_{i+1})$. Obviously, the opposite of $x_i$ is either perfect or empty in $ST(G_i)$ and is inherited by $ST(G_{i+1})$. It follows that $ST(G_{i+1})$ contains the desired  star node $u_{i+1}$.

Assume that for $i<k<j$, $ST(G_k)$ has a star node $u_k$ in which $c_i$ is the centre and $q_i$ is the degree-1 marker vertex identified at the creation of $u_i$. Notice that by the definition of a star, $i>2$. It follows that for every $k$ such that $i<k<j$, $(B_i,\{x_i,\ldots,x_k\})$ is a split of $G_k$. Moreover, observe that $B_i=L(c_i)\cup L(q_i)$. 
%
%EME the sentence below used results from christophe's draft that disappeared, and even with these results was not understandable to me... had to be rewritten, still maybe not enough clear
%In other words, $(B_i,\{x_i,\ldots,x_k\})$ is preserved by the insertion of $x_{k+1}$. Thereby Corollary~\ref{cor:preserved-degenode-split} implies that, in $ST(G_k)$ marked by the neighbourhood of $x_{k+1}$, $c_i$ is perfect and $q_i$ is empty. 
%It follows that the tree-edges respectively incident to $c_i$ and $q_i$ are not fully-mixed. 
%So by definition, $c_i$ and $q_i$ are inherited by $ST(G_{k+1})$ and, by our incremental split-tree construction, $ST(G_{k+1})$ contains the desired star node $u_{k+1}$.
%
Consider the GLT obtained by a node-split of $u_k$, creating a tree-edge $e$ corresponding to the split $(B_i,\{x_i,\ldots,x_k\})$.
Since $(B_i,\{x_i,\ldots,x_k,x_{k+1}\})$ is also a split of $G_{k+1}$,
the extremity $q$ of $e$ such that $L(q)=B_i$ is perfect or empty in
this GLT marked with respect to the neighbourhood of $x_{k+1}$.
So, by our incremental split-tree construction, the marker vertices in $T(q)$, and in particular  $c_i$ and $q_i$, are inherited by $ST(G_{k+1})$. Hence  $ST(G_{k+1})$ contains the desired star node $u_{k+1}$.
%the tree-edge $e$ is not affected by the insertion of $x_{k+1}$,
%and neither is the star created by this node-split and having marker vertices $c_i$ and $q_i$. 
%More precisely:
%the tree-edge $e$ has an extremity which is a descendant of an empty or perfect marker vertex in the above GLT marked with respect to the neighbourhood of $x_{k+1}$.
%
\end{proof}

\begin{lemma} \label{lem:phase1Limit}
Let $\sigma$ be an LBFS ordering of a connected graph $G$. Let $c_i$ be the centre marker vertex of a star $u_i$ in $ST(G_i)$. If $c_i$ is inherited by $ST(G_j)$, with $i\leqslant j$, then it is not opposite a degree-1 marker vertex of a star in $ST(G_j)$.
\end{lemma}

\begin{proof}
Assume without loss of generality that $c_i$ has been generated by $x_i$'s insertion, that is $u_i$ is a 
%DGC10degree 
degree three node $u_i$ labelled by a star.
%three star. 
Let $q_i$ be the degree-1 marker vertex of $u_i$ not opposite $x_i$ in $ST(G_i)$. Let $k$ be the smallest index such that $c_i$ is inherited by $ST(G_k)$ and is opposite a degree-1 marker vertex. Clearly as $ST(G_i)$ is reduced, $i<k$. By assumption, in $ST(G_{k-1})$, $c_i$'s opposite has degree at least two. By Propositions~\ref{prop:unique-node-cases123}, \ref{prop:hybrid-degen-case4}, \ref{prop:unique-edge-case56} 	and~\ref{prop:fully-mixed-case-7}, the only way to make $c_i$ the opposite of a degree-1 marker vertex in $ST(G_k)$ is to subdivide the tree-edge $e$ incident to $c_i$ in $ST(G_{k-1})$ by a star node adjacent to $x_k$. That is, case 6 of Theorem~\ref{th:cases} holds and $e$ was the unique tree-edge with one perfect and one empty extremity (the perfect extremity being $c_i$). Observe that %by Theorem~\ref{th:cases}, 
the node $u$ containing $c_i$  in $ST(G_{k-1})$ cannot be a star node
(Corollary \ref{cor:hereditary}-5). %(see Claim~\ref{cl:PP-PE}). 
We now contradict this fact. To that aim observe that as $c_i$ is perfect in $ST(G_{k-1})$ and in $ST(G_i)$, we have $N(x_k)\cap B_i \subseteq N(x_i)\cap B_i$. As $i<k$, for $\sigma$ to be a LBFS ordering, we have $N(x_k)\cap B_i = N(x_i)\cap B_i$; in other words, $x_k\in Slice(x_i)$. But now Lemma~\ref{lem:star} implies that $u$ is a star: contradiction.
\end{proof}

%EME par below rewritten
% One way of interpreting the previous results is to say that once a star is created, it is then expanded maximally (assuming an LBFS ordering is followed).  However, we are more interested in what the last result says about phase 1 node-joins.  Recall that phase 1 node-joins involve two nodes, one of which is a star.  Also recall that in a reduced GLT, only stars contain degree one marker vertices.  Of course, the result of contraction is always a prime node.  Lemma~\ref{lem:phase1Limit} therefore restricts the number of phase 1 node-joins a node can undergo.  We need this fact to bound the number of new label edges created during contraction.

The last results rely on the crucial assumption that an LBFS ordering is followed.
One way of interpreting them is to say that once a star is created, it is then expanded maximally. 
The important consequence is what the last result says about the phase 1 node-joins defined by Algorithm~\ref{alg:contraction}.  Recall that phase 1 node-joins involve a star whose root marker vertex is its centre, and one of its children.
Lemma~\ref{lem:phase1Limit} therefore restricts the number of phase 1 node-joins a node can undergo.  
We need this fact to bound the number of new label-edges created during contraction.

%----------------------------------------------
\subsubsection{Stamping schemes}
\label{sub:stamping}

The amortized complexity analysis relies on two stamping schemes. First, every non-root marker vertex $q$ of a degenerate node is stamped with a vertex of the input graph $G$, called the \emph{degenerate stamp} of $q$. As a consequence of Lemma~\ref{lem:degenerate-marker}, degenerate stamps can be assigned such that every vertex of $G$ is used at most three times. Intuitively, the role of degenerate stamping is to amortize the cost of the creation of the label-edges of degenerate nodes prior to some node-join operation.

In addition, another stamping scheme is developed to amortize the cost of the creation of the label-edges generated by the node-join operations during contraction. Consider the following inductive procedure which, given a reduced GLT, assigns a stamp $s(q) = (s_1(q),s_2(q))\in V(G)^2$ to every non-root marker vertex $q$ that is not the centre of a star:

\begin{enumerate}
\item If $q$ is opposite the leaf $y$, then $s(q) = (y,y)$.
\item Let $uv$ be an internal tree-edge in the split-tree $ST(G)$ with extremities $q \in V(u)$ and $r \in V(v)$, where $u$ is the parent of $v$:
\begin{enumerate}
\item if $d(r) > 1$, then set $s(q) = (s_2(t),s_2(t'))$ for two (arbitrary) neighbours $t$ and $t'$ of $r$;
\item if $d(r) = 1$, and therefore $v$ is a star with centre $c$, then set $s(q) = s(c)$, and then remove $c$'s stamp.
\end{enumerate}
\end{enumerate}

We will refer to $s_1(q)$ as $q$'s \emph{primary stamp} and $s_2$ as $q$'s \emph{secondary stamp}.  We are interested in primary stamps; secondary stamps only exist to be ``passed up'' in step 2(a) above.  The procedure guarantees the following properties of these stamps:

% MT 08/27/12 - Reviewer wanted the following three results changed from observations to lemmas.  All future references to them have been updated without further comment.
\begin{lemma} \label{FACT1}
At the end of the procedure, centres of stars are the only non-root marker vertices without a stamp.
\end{lemma}
\begin{proof}
This follows from the observation that only step 2(b) removes a stamp.
\end{proof}

\begin{lemma} \label{FACT2}
At the end of the procedure, if the leaf $y$ is the primary stamp of the marker vertex $q$, then $y \in A(q)$.
\end{lemma}
\begin{proof}
An easy inductive argument shows this, applying the fact that $q$ only receives a stamp via its
accessibility paths.
%DGC11neighbours.
\end{proof}

\begin{lemma} \label{FACT3}
At the end of the procedure, every leaf is a primary stamp at most twice.
\end{lemma}
\begin{proof}
Let $uv$ be an arbitrary edge in $T$, where $u$ is the parent of $v$, and let $q \in V(u)$ and $p \in V(v)$ be arbitrary non-root marker vertices, where $p$ is accessible from $q$. We let $y$ be an arbitrary vertex and examine how occurrences of $y$ in $s(p)$ can be transmitted to $s(q)$. Note that the stamp $(y,y)$ applies to the marker vertex opposite $y$, and thus a bottom-up argument starts with $y$ having appeared once as a primary stamp.

First, we observe that no step of the algorithm allows a primary occurrence of $y$ in $s(p)$ to be a secondary occurrence of $y$ in $s(q)$. Suppose for contradiction that a primary occurrence of $y$ in $s(p)$ is also a primary occurrence of $y$ in $s(q)$.  This can only happen by execution of step 2(b), but now the stamp is removed from $p$.  Thus this case does not allow an increase in the number of times that $x$ appears as a primary stamp.  

Finally, suppose that a secondary occurrence of $y$ in $s(p)$ becomes a primary occurrence of $y$ in $s(q)$.  Step 2(a) allows this to happen thereby increasing by one the number of times that $y$ can appear as a primary stamp.  The preceding argument shows that this cannot occur again.
\end{proof}

%EG1
%{\small
%{\bf Emeric's comment:} The paragraph below is not at all understandable to me. As well as the first paragraph of the next subsection. They seem to repeat each other (and to repeat previous introductions) and despite of this they are not precise enough (or too precise if their aim was just to give an idea). Indeed I cannot understand at this point of the reading the link between undefined forthcoming charges and the above constructive observations... Moreover - and more importantly - these observations are never cited again in what follows ! References to them are surely missing in some later proofs... but I did not fix this.
%}\smallskip

What these properties will allow us to do, after the next subsubsection, is to transfer the charge assigned to marker vertices (that are not centres of stars)  to their primary stamps.  Lemma~\ref{FACT2} allows us to associate charge with an edge (incident to the primary stamp) in the underlying accessibility graph.  Lemma~\ref{FACT3} allows us to associate the charge with a vertex (i.e. the primary stamp) in the underlying accessibility graph (and to do this at most twice for each vertex). 
Then we will be able to bound the total charge  in terms of the input graph parameters.

%----------------------------------------------------
\subsubsection{The charging apparatus}
\label{sub:charging}

The idea of the charging argument is to charge the creation of each new label-edge to one of its incident marker vertices.  
To that aim, each marker vertex $q$ is associated throughout its lifetime with a list of vertices $Charge(q)$ that can be given units of charge.

\begin{definition}
Let $q$ be a marker vertex of a node $u$ in a (rooted) split-tree $ST(G)$ of a connected graph $G$. The list $Charge(q)$ of vertices of $G$ contains a set of vertices of $G$ such that:
\begin{itemize}
\item each element in the list can be given a number of units of charge;
\item the vertices are divided into groups, one for each of $q$'s neighbours in $G(u)$;
\item the vertices in neighbour $t$'s group are the vertices in $A(t)$;
\item the root marker vertex's group (if it exists) is called the \emph{root group}.
\end{itemize}
For a leaf $x$ of $ST(G)$ (i.e. a vertex of $G$), the list $Charge(x)$ contains the vertices of $N(x)=A(x)$.
\end{definition}

The way we assign charge during the algorithm is described precisely in
forthcoming Lemmas
\ref{prop:amortized1}, \ref{prop:phase1}, \ref{prop:phase2} and \ref{prop:phase3}. 
Of course, the Charge lists are not static during the construction of $ST(G)$: as new vertices are inserted, new elements must be added to some $Charge$ lists; and as node-joins and node-splits occur, new groups are created and destroyed, respectively. 
However, throughout these changes, the following invariant will be maintained:

\begin{invariant} \label{invariant1}
Let $ST(G)$ be the (rooted) split-tree of a connected graph $G$. 

\begin{itemize}
\item if $q$ is a marker vertex of a node $u$ of $ST(G)$, then
%Let $q$ be a marker vertex of a node $u$ in a (rooted) split-tree $ST(G)$ of a connected graph. Then:
\begin{enumerate}
\item the list $Charge(q)$ is free of charge if 
\begin{enumerate}
\item $u$ is a degenerate node or at some intermediate step of the contraction, $q$ has degree one in $G(u)$;
\item $q$ is a root marker vertex and has never been the centre of a star at some prior step;
\end{enumerate}
\item if $q$ is adjacent to the root of $G(u)$, then the root group in $Charge(q)$ is free of charge;
\item at most one vertex in each group in $Charge(q)$ has been assigned charge;
\item every vertex in $Charge(q)$ has been assigned at most one unit of charge if $q$ is a root marker vertex, and  at most three units of charge otherwise.
\end{enumerate}

\item if $x$ is a leaf of $ST(G)$, then each vertex in $Charge(x)$ is assigned at most three units of charge.
\end{itemize}
Moreover the number of label-edges created during the process of constructing $ST(G)$ is bounded by the total charge on all the $Charge()$ lists.
\end{invariant}

Let us observe that a marker vertex can have degree one (condition 1(a)) 
and not belong to a degenerate node 
only at some intermediate step of the contraction prior to the vertex insertion. 
%Meaning that at such step, the current GLT may have some nodes which are not degenerate, but contain degree one marker vertices. 
Of course, once the new vertex is inserted, this is no longer possible since the current GLT is a split-tree (see Proposition~\ref{prop:fully-mixed-case-7}).

%\begin{invariant} \label{invariant1}
%Let $q$ be a marker vertex.  Then,
%\begin{enumerate}
%\item if $q \in V(u)$ for degenerate node $u$, then $Charge(q)$ is free of charge;
%\item if $d(q) = 1$, then $Charge(q)$ is free of charge;
%\item the root group in $Charge(q)$ (if it exists) is free of charge;
%\item if $q$ is a root marker vertex, then $Charge(q)$ is free of charge, unless $q$ was once the centre of a star but is no longer;
%\item at most one vertex in each group in $Charge(q)$ has been assigned charge;
%\item if $q$ is a root marker vertex, then no vertex in $Charge(q)$ has been assigned more than 1 unit of charge;
%\item if $q$ is a non-root marker vertex, then no vertex in $Charge(q)$ has been assigned more than 3 units of charge.
%\end{enumerate}
%\end{invariant}

%\begin{invariant} \label{invariant2}
%The total number of label edges created during contraction is equal to the total of the charge on all $Charge$ lists.
%\end{invariant}

Assume the invariants hold for the split-tree $ST(G)$.  Now consider forming the split-tree $ST(G+x)$, where $x$ is 
is the last vertex in an LBFS ordering of $G+x$. We will consider the changes required of $ST(G)$ to form $ST(G+x)$, as described by Propositions~\ref{prop:unique-node-cases123}, \ref{prop:hybrid-degen-case4}, \ref{prop:unique-edge-case56} and \ref{prop:fully-mixed-case-7}.

\begin{lemma} \label{prop:amortized1}
Let $x$ be the last vertex of an LBFS ordering of the connected graph $G+x$. If Invariant~\ref{invariant1} is satisfied by $ST(G)$ and if $ST(G)$ does not contain a fully-mixed edge, then Invariant~\ref{invariant1} is satisfied by $ST(G+x)$.
\end{lemma}
\begin{proof}
In every case, except if $ST(G)$ contains a unique hybrid prime node (case 3 of Theorem~\ref{th:cases}), the modifications performed on $ST(G)$ to obtain $ST(G+x)$ only involve degenerate nodes. Thus no label-edge is created. By condition 1 of Invariant~\ref{invariant1} every list $Charge(q)$ for a marker vertex of a degenerate node is free of charge, and  Invariant~\ref{invariant1} is still valid after $x$'s insertion.

In the case $ST(G)$ contains a unique hybrid prime node $u$, then by Proposition~\ref{prop:unique-node-cases123} new label-edges are created incident to $x$'s opposite, namely $q\in V(u)$ the new created marker vertex. Clearly $q$ is not the root marker of $u$ and $Charge(q)$ is divided in $|P(u)|=d(q)$ groups. One vertex of each of these groups receives a unit charge. If one of these group is the root group, then the charge assigned to one of its vertices is shifted to one of the other already charged vertices of $Charge(q)$. It follows that Invariant~\ref{invariant1} it still satisfied.
\end{proof}

We now deal with the case where $ST(G)$ contains a fully-mixed subtree $T'$. By the arguments used in the proof above, since the cleaning step only involves degenerate nodes and thus does not create any label-edges, the GLT $c\ell(ST(G))$ still satisfies Invariant~\ref{invariant1}. Our split-tree algorithm uses Algorithm~\ref{alg:contraction} to perform contraction.  Recall that it separates node-joins into three phases.  Phase 1 node-joins involve star nodes whose root marker vertex is its centre.  Phase 2 node-joins involve nodes whose root marker vertex has degree one.  Phase 3 node-joins are all those remaining.  No matter the phase, a node-join creates new label-edges.  We need to assign charge to account for every one of these edges.  However, this is done differently for each of contraction's three phases, as explained below.

\begin{lemma}[Phase 1 node-joins] \label{prop:phase1}
Let $x$ be the last vertex of an LBFS ordering of the connected graph $G+x$ and assume $c\ell(ST(G))$ satisfies Invariant~\ref{invariant1}. If the node-join$(u,u')$ is performed on $c\ell(ST(G))$ between a star node $u$, whose root marker vertex is its centre, and a child $u'$ of $u$, then the resulting GLT satisfies Invariant~\ref{invariant1}.
\end{lemma}
\begin{proof}
Let $c$ be the centre of $u$. Notice that $Charge(c)$ is free of charge, by condition~4 of Invariant~\ref{invariant1}.  Let $t\in V(u)$ and $t'\in V(u')$ be the extremities of the tree-edge between $u$ and $u'$. Notice that $t$ is a degree one marker vertex.

Prior to the node-join, we need to create the label-edges of the graph $G(u)$ since it is degenerate and of $G(u')$ if $u'$ is degenerate. Every such 
label-edge $e$ is incident to a non-root marker vertex $q$ whose degenerate stamp is a vertex $z$ of $G$. Let $y$ be a leaf of $A(q')$, where $q'$ is the other marker vertex incident to $e$. Observe that $y$ belongs to $Charge(z)$. Then one unit is charged to $y$'s entry in $Charge(z)$ for the cost of the creation of $e$. As $z$ appears at most three times as a degenerate stamp (see Lemma~\ref{lem:degenerate-marker}), Invariant~\ref{invariant1} is satisfied.

So assume the label-edges of $G(u)$ and $G(u')$ exist. If $d(t')>1$, then the node-join of $u$ and $u'$ results in $d(t')$ extra label-edges being created.  But notice that it also results in $t$'s group in $Charge(c)$ being replaced by $d(t')$ new groups, each free of charge.  So to each of these new groups we assign one unit of charge. If $t'$ is a degree one marker vertex, then recall the node-join is handled differently (see Algorithm~\ref{alg:nJoin}). Only one new label-edge is added between $c$ and $t'$'s unique neighbour. Again, for this we assign one unit of charge to what was $t$'s group in $Charge(c)$.

Now, since $d(t) = 1$, we know $Charge(t)$ is free of charge, by condition~1(a) of Invariant~\ref{invariant1}.  We also know that $t'$ is not the centre of a star, because $ST(G)$ is reduced.  Moreover, $t'$ can never have been the centre of a star, because of Lemma~\ref{lem:phase1Limit} and the fact that $d(t) = 1$.  It follows from condition~1(b) of Invariant~\ref{invariant1} that $Charge(t')$ is free of charge.  In other words, no charge is lost in deleting $Charge(t)$ and $Charge(t')$ along with $t$ and $t'$.  It follows easily that the number of label-edges created so far is bounded by the total charge on all the $Charge()$ lists.

It is also easy to verify that every condition of Invariant~\ref{invariant1} continues to hold, although we single out condition~1(b) for comment.  The key for condition~1(b) is that $c$ is no longer the centre of a star after the node-join is performed, and thus is allowed to have charge.
\end{proof}

We can now assume that all phase 1 node-joins are complete. Let us turn to phase 2 node-joins. Recall from Algorithm~\ref{alg:contraction} that a phase 2 node-join involves a node $u$ and one of its children $u'$ such that the root of $u'$ is a degree one marker vertex. This type of node-join is implemented differently (see Algorithm~\ref{alg:nJoin}). Observe also that $u'$ could not have resulted from any node-join in phase 1.

\begin{lemma}[Phase 2 node-joins] \label{prop:phase2}
Let $x$ be the last vertex of an LBFS ordering of the connected graph $G+x$. Assume that all the phase 1 node-joins have been performed on $c\ell(ST(G))$ and that the resulting GLT satisfies Invariant~\ref{invariant1}. If the node-join$(u,u')$ is performed between a node $u$ and one of its children
$u'$ which is a star node rooted at a degree one marker vertex, then the resulting GLT satisfies Invariant~\ref{invariant1}.
\end{lemma}
\begin{proof}
First observe that as $u'$ is a star node and $u$ is possibly a clique node, the label-edges of $G(u)$ and $G(u')$ have to exist prior to the node-join. The cost of this creation can be charged, as described in the proof of Lemma~\ref{prop:phase1}, to the degenerate stamps of their non-root marker vertices.

Let $r$ be the root of $u'$ ($r$ has degree one in $G(u')$) and let $c$ be the centre of $u'$. Let $q$ be the opposite of $r$.  The node-join proceeds by deleting $c$'s neighbours (other than $r$) from $u'$, and adding them to $u$ as neighbours of $q$.  Suppose that $k$ neighbours of $c$ are moved in this way.  Then $k$ new label-edges are created by the node-join.  But notice that adding the new edges incident to $q$ creates $k$ new groups in $Charge(q)$, each being free of charge by virtue of being new.  So to each of these groups we assign one new unit of charge. 

By condition~1(a) of Invariant~\ref{invariant1}, both $Charge(c)$ and $Charge(r)$ are free of charge.  So no charge is lost deleting $Charge(c)$ and $Charge(r)$ along with $c$ and $r$.  Therefore the number of label-edges created so far is bounded by the total charge on all the $Charge()$ lists.
\end{proof}

We now consider phase 3 node-joins. This means that all the node-join involving a star node have been performed.

\begin{lemma}[Phase 3 node-joins] \label{prop:phase3}
Let $x$ be the last vertex of an LBFS ordering of the connected graph $G+x$. Assume that all the phase 1 and phase 2 node-joins have been performed on $c\ell(ST(G))$ and that the resulting GLT satisfies Invariant~\ref{invariant1}. If the node-join$(u,u')$ is performed, then the resulting GLT satisfies Invariant~\ref{invariant1}.
\end{lemma}

\begin{proof}
As discussed in the proof of Lemma~\ref{prop:phase1}, if $u$ or $u'$ is a degenerate node (it must be a clique in this case), then the cost of creating the corresponding label-edges can be charged to the list of some degenerate stamps while preserving Invariant~\ref{invariant1}.

Let $q \in V(u)$ and $r \in V(u')$ be the extremities of the edge $uu'$. Since all phase 1 and phase 2 joins have been performed, we can assume that $d(q) > 1$ and $d(r) > 1$.  Before assigning charge to account for the $d(q) \cdot d(r)$ new label-edges that are created, we will want to redistribute any charge on $Charge(q)$ and $Charge(r)$.

First consider the redistribution of the charge on $Charge(q)$.  Let $t$ be a neighbour of $q$ in $G(u)$.  If $t$ is the root marker vertex, then $t$'s group in $Charge(q)$ is free of charge by condition~2 of Invariant~\ref{invariant1}, and so no charge in this group needs to be redistributed.  So let $t$ be a non-root marker vertex that is a neighbour of $q$.  Then by condition~3 of Invariant~\ref{invariant1}, we know that at most one vertex in $t$'s group in $Charge(q)$ has been assigned charge.  If such a vertex exists, then call it $\lambda$.  
%DGC12
Notice that during the $u,u'$ node-join, $Charge(t)$ will lose $q$'s group but will gain the $d(r) > 1$ groups in $Charge(r)$.  By condition~3 of Invariant~\ref{invariant1}, at least one of these groups (say $\gamma$) will be free of charge.
 The charge on $\lambda$ is reassigned to one of the vertices in $\gamma$ in this case.  Continuing this for all such $t$ removes all charge on $Charge(q)$, and so it can be deleted along with $q$ and no charge is lost.  

We now turn to the redistribution of the charge on $Charge(r)$.  Let $t'$ be a neighbour of $r$ in $G(u')$, and notice that by condition~3 of Invariant~\ref{invariant1}, at most one vertex in $t'$'s group in $Charge(r)$ has been assigned charge.  If such a vertex exists, call it $\lambda'$.  By condition~4 of Invariant~\ref{invariant1}, we know that $\lambda'$ has been assigned no more than one unit of charge.  Now, once more by condition~2 of Invariant~\ref{invariant1}, we know that $r$'s group in $Charge(t')$ is free of charge.  
Furthermore, during the join, $Charge(t')$ will lose $r$'s group but will gain the $d(q) > 1$ groups in $Charge(q)$.  
  Let $\gamma'$ be one of these new groups.  In this case we reassign $\lambda'$'s one unit of charge (if it exists) to a label in $\gamma'$.  Continuing this for all such $t'$ removes all charge on $Charge(r)$, and so it can be deleted along with $r$ and no charge is lost.

We finally assign $d(q) \cdot d(r)$ new units of charge.  Let $t'$ and $\gamma'$ be as above.  So $r$'s former group in $Charge(t')$ is replaced by $d(q) > 1$ new groups, one of them called $\gamma'$.  Only $\gamma'$ (possibly) has any charge assigned to it, and only one unit at that; the other $d(q) - 1$ groups are free of charge.   To $d(q) - 2$ of the groups free of charge we assign one unit of charge, and to the remaining group free of charge we assign two units of charge.  Lastly, if one of the groups corresponds to the root marker vertex of $u$, then the charge just assigned to that group is shifted to another, which must exist.  The result is that only one vertex in each group contains charge, none having been assigned more than three units, and the root group becomes free of charge.

It is easy to verify that Invariant~\ref{invariant1} continues to hold under this (re)assignment of charge. 
\end{proof}

%----------------------------------------------
\subsection{Bounding the total charge, and the running time of our algorithm}
\label{sub:final}

Lemmas~\ref{prop:amortized1}, \ref{prop:phase1}, \ref{prop:phase2} and \ref{prop:phase3} guarantee that Invariant~\ref{invariant1} holds during the LBFS incremental construction of the split-tree of $G$.  Consequently, the total number of label-edges created all along the construction is bounded by the total charge residing on all $Charge$ lists.  

%The first invariant allows us to bound this charge:

\begin{lemma} \label{chargeBound}
Let $G$ be a connected graph. The total number of label-edges created during our LBFS incremental construction of $ST(G)$ is $O(n+m)$, where $n$ is the number of vertices in $G$ and $m$ is the number of its edges.
\end{lemma}

\begin{proof}
Assign primary stamps to the marker vertices in $ST(G)$ as described earlier.  By condition~1(a) of Invariant~\ref{invariant1}, we can focus on the $Charge$ lists residing on marker vertices in prime nodes and on leaves of $ST(G)$.

By Invariant~\ref{invariant1}, the elements of the lists $Charge(x)$ received at most $3$ units of charge. As these lists contain exactly $\sum_{x\in V(G)} d_G(x)= \sum_{x\in V(G)} |A(x)|$ elements, the total charge on these lists is bounded by $O(m)$.

Let $u$ be a prime node, and let $r \in V(u)$ be $u$'s root marker vertex, and let $q$ be one of $u$'s non-root marker vertices.  If $q$ and $r$ are adjacent, then download all the charge from $q$'s group in $Charge(r)$ to $r$'s group in $Charge(q)$.  By condition~4 of Invariant~\ref{invariant1}, the total charge in $r$'s group in $Charge(q)$ is not more than four units.  By the same condition~4, the total charge in the other groups in $Charge(q)$ is no more than three units.  

Now, by choice of $q$ and Lemma~\ref{FACT1}, we can assume that $q$ has been assigned a primary stamp.  Moreover, the vertex $s_1(q)$ acting as primary stamp is in $A(q)$, by Lemma~\ref{FACT2}.  So by definition of $Charge(q)$, we know $s_1(q)$ is adjacent to every label in $Charge(q)$.  The total charge on all such $Charge$ lists is therefore $O(n+m)$, by Lemma~\ref{FACT3} and our discussion above.
\end{proof}

More important than the bound above is what it implies:

\begin{lemma} \label{numNjoins}
Let $G$ be a connected graph. The total number of node-joins and \emph{\texttt{union()}} operations performed by our LBFS incremental construction of $ST(G)$ is $O(n+m)$, where $n$ is the number of vertices in $G$ and $m$ is the number of its edges.
\end{lemma}  

\begin{proof}
Notice that during every node-join at least one new label-edge is created (Lemma \ref{nJoinTime}).  The bound on the number of node-joins now follows from Lemma~\ref{chargeBound}: it is $O(n+m)$.

There are exactly three ways our algorithm applies a \texttt{union()} operation: once after each \texttt{initalization()} in a node-join, once for finalizing every node-join, and once when a new prime node is formed.  The number of such applications of the first way is $O(n)$ by Lemma \ref{initializationCost}, 
the number of the second way is $O(n+m)$ as said above, and the number of the third way is $O(n)$, since at most one new prime node is formed for each vertex inserted.
\end{proof}

%\begin{corollary} \label{numUnions}
%The total number of $union$ operations performed by our algorithm in its construction of $ST(G+x)$ is $O(n+m)$, where $n$ is the number of vertices in $G+x$ and $m$ is the number of its edges.
%\end{corollary}
%\begin{proof}
%There are exactly two ways our algorithm applies a $union$ operation: once for every node-join, and once when a new prime node is formed.  The number of the former is $O(n+m)$, by corollary~\ref{numNjoins}.  The number of the latter is $O(n)$, because at most one new prime node is formed for each vertex inserted.
%\end{proof}

% MT 08/27/12 - Reviewer wanted to clarify what was meant by "our algorithm".
\begin{lemma}\label{lem:numSubtree}
Let $G_n$ be a connected graph with $n$ vertices and $m$ edges 
whose split-tree is 
incrementally constructed by repeated application of Algorithm~\ref{alg:vertexinsertion}; 
%built by our algorithm, 
that is $ST(G_{i+1})=ST(G_i+x_{i+1})$ is built from $ST(G_i)=(T_i, \mathcal{F}_i)$ for $1\leq i\leq n-1$.
Then the sum of 
%EME i removed the O below, since rigoursously we would then need to precise that the O is obtained with the same constant for all i... otherwise the sum could explode.
%$O(|T_i(N_{G_{i+1}}(x_{i+1}))|)$ 
$|T_i(N_{G_{i+1}}(x_{i+1}))|$ 
over all $1\leq i\leq n-1$  is $O(n+m)$.
\end{lemma}

\begin {proof}
For a fixed $i$, let $G$ denote $G_i$, $T$ denote $T_i$, $x$ denote $x_{i+1}$, and $N(x)$ denote $N_{G_{i+1}}(x)$.
We can divide the nodes of $T(N(x))$ into two groups: those that remain in $ST(G+x)$, and those that do not.  The number of those that remain is $O(|N(x)|)$, by Lemma~\ref{treeSize}.  So the total number of nodes in the first group over the entire execution of our algorithm is $O(n+m)$.  Every node in the second group participates in at least one node-join.  So the total number of nodes in the second group over the entire execution of our algorithm is also $O(n+m)$, by Lemma~\ref{numNjoins}.
\end{proof}

\begin{lemma} \label{numFinds}
Let $G$ be a connected graph. The total number of \emph{\texttt{find()}} operations performed by our LBFS incremental construction of $ST(G)$ is $O(n+m)$, where $n$ is the number of vertices in $G+x$ and $m$ is the number of its edges.
\end{lemma}  

\begin{proof}
Our algorithm uses \texttt{find()} operations to traverse the split-tree.  These traversals can take place during case identification and state assignment, cleaning, and contraction.  Case identification and state assignment take place according to Algorithm~\ref{alg:perfectPruning}; cleaning takes place according to Algorithm~\ref{alg:cleaning}; and contraction takes place according to Algorithm~\ref{alg:contraction}.  So by Lemmas~\ref{lem:PruningCorrectness},~\ref{cleaningTime}, and~\ref{contractionTime}, the total number of $find$ operations required for the insertion of $x$ is $O(|T(N(x))|)$.
So the total number for the whole algorithm is $O(n+m)$ by Lemma \ref{lem:numSubtree}.
%We can divide the nodes of $T(N(x))$ into two groups: those that remain in $ST(G+x)$, and those that do not.  The number of those that remain is $O(|N(x)|)$, by Lemma~\ref{treeSize}.  So the total number of nodes in the first group over the entire execution of our algorithm is $O(n+m)$.  Every node in the second group participates in at least one node-join.  So the total number of nodes in the second group over the entire execution of our algorithm is also $O(n+m)$, by Lemma~\ref{numNjoins}.
\end{proof}

%\begin{corollary} \label{sumTreeSizes}
%Let $ST(G_i) = (T_i,\mathcal{F}_i)$.  Then $\Sigma_{1 \le i \le n}|T_{i-1}(N(x_i))|$ is $O(n+m)$.
%\end{corollary}
%\begin{proof}
%The second paragraph in the proof of Corollary~\ref{numFinds} suffices.
%\end{proof}

\begin{lemma} \label{numNodes}
Let $G$ be a connected graph. The total number of nodes (including fake nodes) created in the rooted GLT data-structure used by our algorithm, and the total number of elements in the union of children-sets, is $O(n+m)$.
\end{lemma}  

\begin{proof}
A node is created either when the new vertex $x$ is inserted, or when a node-split is performed. At most one node-split is performed in the case where there is no fully-mixed subtree.
Therefore, the total number of nodes created in this case and by the insertion of $x$ is $O(n)$ over the course of the LBFS construction of $ST(G)$.
If there is a fully-mixed subtree, then node-splits are performed during the cleaning step and they involve nodes in $T(N(x))$. At most two node-splits are performed at each such node. So the total number of created nodes for the whole algorithm is $O(n+m)$ by Lemma \ref{lem:numSubtree}. Elements of children-sets in the data-structure are either leaves or some nodes created at some step of the algorithm, hence their total number is also $O(n+m)$.
\end{proof}

The previous lemmas culminate in the theorem below, which is the main result of our paper:

\begin{theorem} \label{bigTheorem}
%Xtof--04-11------
%Given a connected graph $G$ with $n$ vertices and $m$ edges, the split-tree $ST(G)$ can be constructed in time
The split-tree $ST(G)$ of a graph $G=(V,E)$ with $n$ vertices and $m$ edges can be built incrementally according to an LBFS ordering in time
%Xtof--04-11------
%EME-Ack-update
%$O((n+m)\alpha(n+m))$, where 
$O(n+m)\alpha(n+m)$, where 
%$O(\alpha(C.(n+m))(n+m))$, where $C$ is a constant  and 
%EME-Ack-update
%$\alpha$ is the inverse of Ackermann's function.
$\alpha$ is the inverse Ackermann function.
\end{theorem}

\begin{proof}
Lemma~\ref{lem:degenerate-marker} establishes an $O(n)$ bound on the number of non-root degenerate marker vertices. As every node has degree at least $3$, the total number of degenerate marker vertices created during the LBFS incremental construction is $O(n)$. The total number of label-edges created during the LBFS incremental construction is $O(n+m)$, by Lemma~\ref{chargeBound}. Therefore our algorithm generates an $O(n+m)$ size data-structure.

In addition to the cost of computing an LBFS, which takes time $O(n+m)$~\cite{Gol80, HMP00}, we have to bound the cost of the tree traversals, based on \texttt{find()} operations, plus the total cost of the \texttt{initialization()}, \texttt{union()} operations involved in the contraction steps.  
The sum, over all the algorithm, of the term $|T(N(x))|$ in Lemmas~\ref{lem:PruningCorrectness},~\ref{cleaningTime}, and~\ref{contractionTime}, is $O(n+m)$, 
%by the argument used in proof of Lemma~\ref{numFinds}.
by Lemma~\ref{lem:numSubtree}.
 The total cost of \texttt{initialization()} operations required by the algorithm is $O(n)$, by Lemma~\ref{initializationCost}.  The total number of $union$ and $find$ operations is $O(n+m)$, by Lemmas~\ref{numNjoins} and~\ref{numFinds}, respectively.
 
%Thereby time cost the union-find requests amounts to $O(\alpha(n+m)(n+m))$. Combining the statements above supplies the result.
Finally, the cost of the union-find requests amounts to $O(\alpha(N)(n+m)+N)$,
where $N$ is the total number of elements in the union of children-sets. 
%
%This number $N$ is $O(n+m)$ by Lemma \ref{numNodes}, which is bounded by $C \cdot (n+m)$ for some constant $C$. Since $\alpha$ is increasing,
%the total cost of our algorithm is $O(\alpha(C \cdot (n+m))(n+m))$.
This number $N$ is $O(n+m)$ by Lemma \ref{numNodes}. 
%
%
%EME-Ack-update
It is easy to prove that $\alpha(C\cdot (n+m))\sim \alpha(n+m)$ for any fixed constant $C$.
Indeed, because of the way the Ackermann function increases,
for any $n$, $\alpha(n)=k-1$ implies $\alpha(C.n)\leq k$ for every $k$ large enough.
Hence $\alpha(C.n)-\alpha(n)\leq 1$ for $n$ large enough, which implies $\alpha(C.n)\sim \alpha (n)$.
In particular, we have that $\alpha(O(n+m))=O(\alpha(n+m))$.

So to conclude, the total cost of our algorithm is $O((n+m)\alpha(n+m))$, which can also be written $O(n+m)\alpha(n+m)$.
\end{proof}

%{\small
%{\bf Emeric's comment on the above proof.}
%To me two things are missing :\\
%- in the litterature, I have found that
%``a series of $k$ union/find requests in a universe of $N$ elements is done in time
%$O( \alpha(N) . k + N )$ '', but we do not care about this $N$ (which should notably take into account the fake nodes)... it disturbs me! Note that if an update is needed, the subsection data-struture should be updated too (the bound $O(\alpha(k)k)$ we use is claimed there)\\
%- applying the bound $O(\alpha(k)k)$ to $k=O(n+m)$ provides a bound $O(\alpha(O(n+m))(n+m)$. We need to justify that $\alpha(O(n+m))=O(\alpha(n+m))$
%}
%\smallskip

%\medskip
%To conclude this section and the paper, 
To conclude,
let us mention that we are prevented from achieving linear time only by the node-join.  Let $u$ and $u'$ be two adjacent nodes, with $u$ the parent of $u'$.  To effect their node-join, the children of $u'$ must be made children of $u$.  That is the bottleneck.  Our implementation does its best to avoid it by using union-find,
%Eme-2011-04-10 
but the optimal complexity for union-find involves the inverse Ackemann function. 
%but the optimal complexity for the union-find problem involves the inverse Ackemann function. 
%but only quasi-linear time can be obtained this way.  
It seems to us that our charging argument can not be extended to cover the cost of reassigning $u'$'s children, thereby eliminating union-find and achieving linear time.
However, it is worth emphasizing that, 
from the practical viewpoint, the inverse Ackermann function can be thought of as a constant, and that every other aspect of our algorithm is consistent with linear time.  

This paper's companion~\cite{GPTC11b} extends our split decomposition algorithm to recognize circle graphs in same time.  
It is the first sub-quadratic circle recognition algorithm, and the first progress on the problem in fifteen years. 

\vspace{-0.2cm}
%%%%%%%%%%%%%%%%%%%%%%%%%

\bibliographystyle{plain}

\end{document}